\documentclass[runningheads,envcountsect,envcountsame,orivec]{llncs}
\pdfoutput=1
\usepackage{silence}
\usepackage{thm-restate}

\usepackage[T1]{fontenc} % Enables the generation of copyable text in the PDF.
\usepackage[utf8]{inputenc} % Allows for UTF-8 input encoding.

\usepackage{amsmath} % Provides \allowdisplaybreaks
\allowdisplaybreaks
\usepackage{cite} % Sorts and compress bibliographic references
\usepackage{cmll}
\usepackage{multicol}
\usepackage{fontawesome}
\usepackage{proof}
\usepackage{float}
\usepackage{amssymb}
\usepackage{stmaryrd}
\SetSymbolFont{stmry}{bold}{U}{stmry}{m}{n} % This solves a font problem in
\usepackage{etoolbox}
\makeatletter
\let\llncs@addcontentsline\addcontentsline
\patchcmd{\maketitle}{\addcontentsline}{\llncs@addcontentsline}{}{}
\patchcmd{\maketitle}{\addcontentsline}{\llncs@addcontentsline}{}{}
\patchcmd{\maketitle}{\addcontentsline}{\llncs@addcontentsline}{}{}
\makeatother
\usepackage{bookmark}
\usepackage{hyperref} % Provides links for references
\hypersetup{
  colorlinks=true,
  linkcolor=blue,
  anchorcolor=blue,
  citecolor=blue,
  urlcolor=blue,
  filecolor=blue,
  breaklinks=true
}
\usepackage{todonotes}
\usepackage{quantikz} % always load this package last

\floatstyle{boxed}
\restylefloat{figure}

\newcommand\OC{\ensuremath{{\mathcal L}_2^{\mathcal S}}}
\newcommand\fv{\ensuremath{\mathsf{fv}}}
\newcommand\FV{\ensuremath{\mathsf{FV}}}
\newcommand\zero{\ensuremath{\mathfrak 0}}
\newcommand\one{\ensuremath{\mathfrak 1}}
\newcommand\abstr[1]{#1.}

\newcommand\inl{\mbox{\small \it inl}}
\newcommand\inr{\mbox{\small \it inr}}
\newcommand\elimone{\delta_{\one}}
\newcommand\elimzero{\delta_{\zero}}
\newcommand\elimwith{\delta_{\with}}
\newcommand\elimplus{\delta_{\oplus}}

\newcommand\elimtens{\delta_{\otimes}}
\newcommand\elimbang{\delta_\oc}

\newcommand\plus{\mathbin{\text{\normalfont\scalebox{0.7}{\faPlus}}}}
\newcommand\SN{\ensuremath{\mathsf{SN}}}

\newcommand\pair[2]{\langle #1, #2 \rangle}

\newcommand\topintro{\langle\rangle}

\newcommand\lra{\longrightarrow}
\newcommand\lla{\longleftarrow}
\newcommand\lras{\longrightarrow^*}
\newcommand\llas{\mathrel{{}^*{\longleftarrow}}}

\newcommand\hatmultimap{\ensuremath{\mathbin{\hat\multimap}}}
\newcommand\hatotimes{\ensuremath{\mathbin{\hat\otimes}}}
\newcommand\hatand{\ensuremath{\mathbin{\hat\with}}}
\newcommand\hatoplus{\ensuremath{\mathbin{\hat\oplus}}}
\newcommand\hatbang{\ensuremath{\mathop{\hat\oc}}}

\newcommand\qedHERE{{\tag*{$\qed$}}}
\newcommand\qedHEREHERE{\eqno\qed}
\newcommand\citeintitle[2]{\cite[#1]{#2}}

\floatstyle{boxed}
\restylefloat{figure}

\begin{document}

\title{A linear proof language for second-order intuitionistic linear logic\protect\footnote{Supported by PICT 2021-I-A-00090 and 2019-1272, PIP 11220200100368CO, the French-Argentinian IRP SINFIN, and CSIC-UdelaR I+D-22520220100073UD.}
}

\titlerunning{A linear proof language for second-order ILL}

\author{Alejandro D\'{\i}az-Caro\inst{123}
  \and
  Gilles Dowek\inst{4}
  \and\\
  Malena Ivnisky\inst{235}
  \and
  Octavio Malherbe\inst{6}
}

\authorrunning{A.~D\'{\i}az-Caro, G.~Dowek, M.~Ivnisky, and O.~Malherbe}
\institute{
  Universidad Nacional de Quilmes. DCyT. Bernal, PBA, Argentina\\
  \and
  Universidad de Buenos Aires. FCEyN. DC. Buenos Aires, Argentina\\
  \and
  Universidad de Buenos Aires-CONICET. ICC. Buenos Aires. Argentina\\
  \and
  Inria, LMF, ENS Paris-Saclay, Gif-sur-Yvette, France\\
  \and
  Universidad de la Rep\'ublica--MEC, PEDECIBA.  Montevideo, Uruguay
  \and
  Universidad de la Rep\'ublica.  IMERL, FIng.  Montevideo, Uruguay
}

\maketitle

\begin{abstract}
  We present a polymorphic linear lambda-calculus as a proof language for
  second-order intuitionistic linear logic. The calculus includes addition and
  scalar multiplication, enabling the proof of a linearity result at the
  syntactic level.


\end{abstract}

\section{Introduction}\label{sec:intro}

Linear Logic~\cite{Girard87} is named as such because it is modelled by vector
spaces and linear maps, and more generally by monoidal
categories~\cite{EilenbergKelly66,Benabou63}. These types of categories also
include the so-called Cartesian categories, generating a formal place of
interaction between purely algebraic structures and purely logical structures,
i.e.~between algebraic operations and the exponential connective ``$\oc$''.  In the
strictly linear fragment (without $\oc$),
functions between two propositions are linear functions. However, expressing
this linearity within the proof-term language itself is challenging. Properties
such as $f(u+v) = f(u)+f(v)$ and $f(a \cdot u) = a \cdot f(u)$ require
operations like addition and scalar multiplication, which are typically absent
in the proof language.

In~\cite{DiazcaroDowekFSCD22,DiazcaroDowek2023}, this challenge has been addressed. The
Intuitionistic Multiplicative Additive Linear Logic is considered and extended 
with addition and scalar multiplication within the proof-terms. The resulting
calculus, the $\mathcal L^{\mathcal S}$-calculus, does not alter the
provability of formulas but allows us to express linear properties. It has been proved
that any proof-term $t~(u\plus v)$, where $t$ is a proof-term of $A \multimap B$ and
$u$ and $v$ are proof-terms of $A$, is extensionally equivalent to the proof-term
$t~u\plus t~v$. Similarly, $t~(a\bullet u)$ is equivalent to $a\bullet t~u$.

This extension involves changing the proof-term $\star$ of proposition $\one$ into a
family of proof-terms $a.\star$, one for each scalar $a$ in a given fixed semiring
$\mathcal S$:%\quad
\[
  \infer[\mbox{\small $\one_i$($a$)}]{\vdash a.\star:\one}{}
\]

The following two deduction rules have also been added:
\[
  \infer[{\mbox{\small sum}}]{\Gamma\vdash t\plus u:A}{\Gamma\vdash t:A & \Gamma\vdash u:A}
  \qquad\qquad
  \infer[{\mbox{\small prod($a$)}}]{\Gamma\vdash a\bullet t:A}{\Gamma\vdash t:A}
\]

Incorporating these rules requires adding commuting rules to preserve
cut-elimination. Indeed, the new rules may appear between an introduction and
an elimination of some connective. For example, consider the following
derivation.
\[
  \infer[\mbox{\small $\with_{e1}$}]{\Gamma\vdash C}
  {
    \infer[\mbox{\small prod($a$)}]{\Gamma\vdash A\with B}
    {
      \infer[\mbox{\small $\with_i$}]{\Gamma\vdash A\with B}
      {
	\Gamma\vdash A
	&
	\Gamma\vdash B
      }
    }
    &
    \Gamma,A\vdash C
  }
\]

To achieve cut-elimination, we must commute the rule prod($a$) either with the
introduction or with the elimination, as follows: 

\[
  \infer[\mbox{\small$\with_{e1}$}]{\Gamma\vdash C}
  {
    \infer[\mbox{\small$\with_i$}]{\Gamma\vdash A\with B}
    {
      \infer[\mbox{\small prod($a$)}]{\Gamma\vdash A}{\Gamma\vdash A}
      &
      \infer[\mbox{\small prod($a$)}]{\Gamma\vdash B}{\Gamma\vdash B}
    }
    &
    \Gamma,A\vdash C
  }
  \infer[\mbox{\small prod($a$)}]{\Gamma\vdash C}
  {
    \infer[\mbox{\small$\with_{e1}$}]{\Gamma\vdash C}
    {
      \infer[\mbox{\small$\with_i$}]{\Gamma\vdash A\with B}
      {
	\Gamma\vdash A
	&
	\Gamma\vdash B
      }
      &
      \Gamma,A\vdash C
    }
  }
\]

Both of these are reducible. We refer to the sum and prod($a$) rules
as \emph{interstitial rules}, as they can appear in the interstice between an
introduction and an elimination. We choose to commute these rules with the
introductions as much as possible.  This means we introduce the following
commutation rule
\(
  a\bullet\pair{t}{u}\lra\pair{a\bullet t}{a\bullet u}
\)
instead of the alternative rule
\(
  \elimwith^1(a\bullet t,x^A.u) \lra a\bullet\elimwith^1(t,x^A.u)
\).
This choice provides a better introduction property: A closed irreducible proof-term
of a proposition $A\with B$ is a pair.

In this new paper, we extend the proof system to second-order intuitionistic linear logic, adding the
exponential connective and a universal quantifier. We prove that the linearity
result still holds for second order. This is an intermediate step in a long-term research towards the definition of second-order categorical semantics in linear categories, building on the work of~\cite{Maneggia} and extending the semantics defined in~\cite{DiazcaroMalherbe22}. The midterm objective is to define a second-order linear model in which quantum computing can be encoded.

%\paragraph{Related work} 
The initial development of the proof
language $\mathcal L^{\mathcal S}$~\cite{DiazcaroDowekFSCD22,DiazcaroDowek2023}, paved the way
for the development of the second-order version presented in this paper.  While
our primary focus is on introducing a minimal extension to the proof language
within the realm of second-order intuitionistic linear logic, our work draws
inspiration from various domains, particularly quantum programming languages. These languages were trailblazers in
merging programming constructs with algebraic operations, such as addition and
scalar multiplication.

QML~\cite{AltenkirchGrattageLICS05} introduced the concept of superposition of terms through the
$\mathsf{if}^\circ$ constructor, allowing the representation of linear
combinations $a.u+b.v$ by the expression $\mathsf{if}^\circ\ a.\ket 0 + b.\ket
1\ \mathsf{then}\ u\ \mathsf{else}\ v$. The linearity (and even unitarity)
properties of QML were established through a translation to quantum circuits.

The ZX calculus~\cite{ZXBook17}, a graphical language based on a categorical
model, lacks direct syntax for addition or scalar multiplication but defines a
framework where such constructs can be interpreted. This language is extended
by the Many Worlds Calculus~\cite{Chardonnet23} which allows for
linear combinations of diagrams.

The algebraic lambda-calculus~\cite{Vaux2009} and Lineal~\cite{Lineal} exhibit
syntax similarities with $\mathcal L^{\mathcal S}$-calculus. However, the
algebraic lambda-calculus lacks a proof of linearity in its simple
intuitionistic type system. In contrast, Lineal enforces linearity without a
dedicated type system, relying on explicit definitions like $f(u+v) = f(u) +
f(v)$ and $f(a.u) = a.f(u)$.  Several type systems have been proposed for
Lineal~\cite{ArrighiDiazcaroLMCS12,DiazcaroPetitWoLLIC12,ArrighiDiazcaroValironIC17,LambdaS,DiazcaroGuillermoMiquelValironLICS19},
including some polymorphic ones. However, none of these systems are related to
linear logic, and their purpose is not to prove linearity, as we do, but rather to enforce
it.

Our contributions are as follows:
\begin{itemize}
  \item We extend the $\mathcal L^{\mathcal S}$-calculus to Church-style
    second-order intuitionistic linear logic, resulting in the $\OC$-calculus 
    (Section~\ref{sec:secls}).
  \item We prove its correctness (Section~\ref{sec:correctness}), namely,
    Subject Reduction (Theorem~\ref{thm:SR}), Confluence
    (Theorem~\ref{thm:Confluence}), Strong Normalisation 
    (Theorem~\ref{cor:ST}), and the Introduction Property
    (Theorem~\ref{thm:introductions}). In particular, the proof of strong
    normalisation involves applying two techniques due to Girard: ultra-reduction
    and reducibility candidates.
  \item Since it is a conservative extension, the encodings for vectors
    (Section~\ref{sec:secvectors}) and matrices
    (Section~\ref{sec:secmatrices}), already present in the $\mathcal
    L^{\mathcal S}$-calculus are still valid. We provide detailed explanations
    of these encodings for self-containment.  Since we have polymorphism and
    the exponential connective, we also show an example of an iterator in the $\OC$-calculus (Section~\ref{sec:ExPoly}).
  \item Finally, we prove that the linearity result is also valid for the
    second order without exponentials (Section~\ref{sec:seclinearity}).
\end{itemize}

\section{The \texorpdfstring{$\OC$}{LS2}-calculus}
\label{sec:secls}

The $\OC$-logic is the second-order intuitionistic linear logic (we follow the presentation of $F_{\mathsf{DILL}}$~\cite{Maneggia}, extended with the additive linear connectives).

\[
  A = X \mid \one \mid A \multimap A \mid A \otimes A \mid \top \mid \zero
  \mid A \with A \mid A \oplus A \mid \oc A \mid \forall X.A
\]
The $\alpha$-equivalence relation and the free and bound variables of a
proposition are defined as usual, we write as $\FV(A)$ the set of free
variables of $A$.  Propositions are defined modulo $\alpha$-equivalence.  A
proposition is closed if it contains no free variables.  We write $(B/X)A$ for
the substitution of $B$ for $X$ in $A$.

Let ${\mathcal S}$ be a semiring of {\it scalars}, for instance $\{\star\}$,
$\mathbb N$, ${\mathbb Q}$, ${\mathbb R}$, or ${\mathbb C}$.  The proof-terms
of the $\OC$-calculus are given in Figure~\ref{fig:proofterms}, where $a$ is a scalar
in $\mathcal S$.
\begin{figure}[t]
  \[
    \begin{array}{rllc}
      &\textrm{Introductions} & \textrm{Eliminations} & \textrm{Connective}\\
      t = x \mid t \plus u \mid a \bullet t \\
      & \mid a.\star & \mid \elimone(t,u) & (\one)\\  
      & \mid \lambda x^A.t & \mid t~u &(\multimap)\\
      &\mid t \otimes u & \mid \elimtens(t, x^A y^B.u) & (\otimes)\\
      & \mid \langle \rangle  && (\top)\\
      & &\mid\elimzero(t) & (\zero)\\
      &\mid \pair{t}{u} &\mid \elimwith^1(t,x^A.u) \mid \elimwith^2(t,x^B.u) & (\with)\\
      &\mid \inl(t)\mid \inr(t) & \mid \elimplus(t,x^A.u,y^B.v) & (\oplus)\\
      & \mid \oc t & \mid \elimbang(t, x^A.u) & (\oc) \\
      & \mid \Lambda X.t &\mid t~A & (\forall)
    \end{array}
  \]
  \caption{The proof-terms of the $\OC$-calculus.\label{fig:proofterms}}
\end{figure}

The $\alpha$-equivalence relation and the free and bound variables of a
proof-term are defined as usual, we write as $\fv(t)$ the set of free variables
of $t$.  Proof-terms are defined modulo $\alpha$-equivalence.  A proof-term is
closed if it contains no free variables.  We write $(u/x)t$ for the
substitution of $u$ for $x$ in $t$ and if $\fv(t) \subseteq \{x\}$, we also use
the notation $t\{u\}$, when there is at most one free variable in $t$.

A sequent has the form $\Xi;\Gamma\vdash t:A$, where $\Xi$ is the non-linear
context and $\Gamma$ the linear one.  The deduction rules are those of
Figure~\ref{fig:figuretypingrules}.  These rules are exactly the deduction rules of
second-order intuitionistic linear natural deduction, with proof-terms, with
two differences: the interstitial rules and the scalars (see Section~\ref{sec:intro}).

\begin{figure}[t]
  \[
    \infer[{\mbox{\small {lin-}ax}}]{{\Xi;} x^A \vdash x:A}{}
    \qquad
    \infer[{\mbox{\small ax}}]{\Xi, x^A;\varnothing \vdash x:A}{}
    \qquad
    \infer[{\mbox{\small sum}}]{{\Xi;} \Gamma \vdash t \plus u:A}{{\Xi;} \Gamma \vdash t:A & {\Xi;} \Gamma \vdash u:A}
  \]
  \[
    \infer[{\mbox{\small prod}(a)}]{{\Xi;} \Gamma \vdash a \bullet t:A}{{\Xi;} \Gamma \vdash t:A}
    \qquad
    \infer[{\mbox{\small $\one_i$}(a)}] {{\Xi;\varnothing} \vdash a.\star:\one}{}
    \qquad
    \infer[{\mbox{\small $\one_e$}}]{{\Xi;}\Gamma, \Delta \vdash \elimone(t,u):A}{{\Xi;}\Gamma \vdash t:\one & {\Xi;}\Delta \vdash u:A}
  \]
  \[
    \infer[{\mbox{\small $\multimap_i$}}]{{\Xi;}\Gamma \vdash \lambda x^A.t:A \multimap B}{{\Xi;}\Gamma, x^A \vdash t:B}
    \qquad
    \infer[{\mbox{\small $\multimap_e$}}]{{\Xi;}\Gamma, \Delta \vdash t~u:B}{{\Xi;}\Gamma \vdash t:A\multimap B & {\Xi;}\Delta \vdash u:A}
  \]
  \[
    \infer[{\mbox{\small $\otimes_i$}}]{{\Xi;}\Gamma, \Delta \vdash t \otimes u:A \otimes B}{{\Xi;}\Gamma \vdash t:A & {\Xi;}\Delta \vdash u:B}
    \qquad
    \infer[{\mbox{\small $\otimes_e$}}]{{\Xi;}\Gamma, \Delta \vdash \elimtens(t, x^A y^B.u):C}{{\Xi;}\Gamma \vdash t:A \otimes B & {\Xi;}\Delta, x:A, y:B \vdash u:C}
  \]
  \[
    \infer[{\mbox{\small $\top_i$}}]{{\Xi;} \Gamma \vdash \langle \rangle:\top}{}
    \qquad
    \infer[{\mbox{\small $\zero_e$}}]{{\Xi;}\Gamma, \Delta \vdash \elimzero(t):C}{{\Xi;}\Gamma \vdash t:\zero}
    \qquad
    \infer[{\mbox{\small $\with_i$}}]{{\Xi;}\Gamma \vdash \pair{t}{u}:A \with B}{{\Xi;}\Gamma \vdash t:A & {\Xi;}\Gamma \vdash u:B}
  \]
  \[
    \infer[{\mbox{\small $\with_{e1}$}}]{{\Xi;}\Gamma, \Delta \vdash \elimwith^1(t,x^A.u):C}{{\Xi;}\Gamma \vdash t:A \with B & {\Xi;}\Delta, x^A \vdash u:C}
    \qquad
    \infer[{\mbox{\small $\with_{e2}$}}]{{\Xi;}\Gamma, \Delta \vdash \elimwith^2(t,x^B.u):C}{{\Xi;}\Gamma \vdash t:A \with B & {\Xi;}\Delta, x^B \vdash u:C}
  \]
  \[
    \infer[{\mbox{\small $\oplus_{i1}$}}]{{\Xi;}\Gamma \vdash \inl(t):A \oplus B}{{\Xi;}\Gamma \vdash t:A}
    \qquad
    \infer[{\mbox{\small $\oplus_{i2}$}}]{{\Xi;}\Gamma \vdash \inr(t):A \oplus B}{{\Xi;}\Gamma \vdash t:B}
  \]
  \[
    \infer[{\mbox{\small $\oplus_e$}}]{{\Xi;}\Gamma, \Delta \vdash \elimplus(t,x^A.u,y^B.v):C}{{\Xi;}\Gamma \vdash t:A \oplus B & {\Xi;}\Delta, x^A \vdash u:C & {\Xi;}\Delta, y^B \vdash v:C}
  \]
  \[
    \infer[{\mbox{\small $\oc_i$}}]{\Xi; \varnothing \vdash \oc t: \oc A}{\Xi; \varnothing \vdash t:A}
    \qquad
    \infer[{\mbox{\small $\oc_e$}}]{\Xi; \Gamma, \Delta \vdash \elimbang(t, x^A.u):B}{\Xi; \Gamma \vdash t:\oc A & \Xi, x^A; \Delta \vdash u:B}
  \]
  \[
    \infer[{\mbox{\small $\forall_i$}}]
    {\Xi;\Gamma \vdash \Lambda X.t:\forall X.A}{\Xi;\Gamma \vdash t:A & X \notin \FV(\Xi,\Gamma)}
    \qquad
    \infer[{\mbox{\small $\forall_e$}}]
    {\Xi;\Gamma \vdash t~A:(A/X)B}{\Xi;\Gamma \vdash t:\forall X.B}
  \]
  \caption{The deduction rules of the $\OC$-calculus.\label{fig:figuretypingrules}}
\end{figure}

\begin{figure}[t]
  \[
    \begin{array}{rl@{\hspace{-8mm}}rl}
    \elimone(a.\star,t) & \lra  a \bullet t 
    	&\elimplus(\inl(t),x^A.v,y^B.w) & \lra  (t/x)v\\
    (\lambda x^A.t)~u & \lra  (u/x)t 
    	&\elimplus(\inr(u),x^A.v,y^B.w) & \lra  (u/y)w\\
    \elimtens(u \otimes v,x^A y^B.w) & \lra  (u/x,v/y)w
    	&\elimbang(\oc t, x^A.u) & \lra (t/x)u\\
    \elimwith^i(\pair{t_1}{t_2}, x^A.v) & \lra  (t_i/x)v
    	&(\Lambda X.t)~A & \lra (A/X)t
	\\[1mm]
	\hline
	\\[-3mm]
      {a.\star} \plus b.\star&\lra  (a+b).\star
      &\pair{t}{u} \plus \pair{v}{w} & \lra  \pair{t \plus v}{u \plus w}\\
      (\lambda x^A.t) \plus (\lambda x^A.u) & \lra  \lambda x^A.(t \plus u)
      &\elimplus(t \plus u,x^A.v,y^B.w) & \lra\\
      \elimtens(t \plus u,x^A y^B.v) & \lra
      &\elimplus(t,x^A.v,y^B.&\!\!w) \plus \elimplus(u,x^A.v,y^B.w)\\
      \elimtens(t,&\!x^A y^B.v) \plus \elimtens(u,x^A y^B.v)
      &\oc t \plus \oc u & \lra \oc (t \plus u)\\
      \langle \rangle \plus \langle \rangle & \lra  \langle \rangle
      &(\Lambda X.t) \plus (\Lambda X.u) &  \lra \Lambda X.(t \plus u)
	\\[1mm]
	\hline
	\\[-3mm]
    a \bullet b.\star&\lra  (a \times b).\star
    	&a \bullet \pair{t}{u} &\lra  \pair{a \bullet t}{a \bullet u}\\
    a \bullet \lambda x^A. t &\lra  \lambda x^A. a \bullet t
    	&\elimplus(a \bullet t,x^A.v,y^B.w) & \lra a \bullet \elimplus(t,x^A.v,y^B.w)\\
    \elimtens(a \bullet t,x^A y^B.v) & \lra a \bullet \elimtens(t,x^A y^B.v)
    	&a \bullet \oc t & \lra \oc (a \bullet t)\\
    a \bullet \langle\rangle &\lra  \langle\rangle
    	&a \bullet \Lambda X.t & \lra \Lambda X.a \bullet t
    \end{array}
  \]
  \caption{The reduction rules of the $\OC$-calculus.\label{fig:figureductionrules}}
\end{figure}

The reduction rules are those of Figure~\ref{fig:figureductionrules}.  As
usual, the reduction relation is written $\lra$, its inverse $\lla$,
its reflexive-transitive closure $\lras$, the reflexive-transitive
closure of its inverse $\llas$, and its reflexive-symmetric-transitive
closure $\equiv$.  The first nine rules correspond to the reduction of
cuts on the connectives $\one$, $\multimap$,
$\otimes$,
$\with$, $\oplus$, $\oc$, and $\forall$.
The sixteen others enable to commute the interstitial rules sum and prod($a$)
with the introduction rules of the connectives $\one$, $\multimap$,
$\top$, $\with$, $\oc$, and $\forall$, and with the elimination rule of the connectives
$\otimes$ and $\oplus$.
For instance, the rule
\(
  \pair{t}{u} \plus \pair{v}{w} \lra  \pair{t \plus v}{u \plus w}
\)
pushes the symbol $\plus$ inside the pair.  The zero-ary commutation rules add
and multiply the scalars:
\(
  {a.\star} \plus b.\star \lra  (a+b).\star
\),
\(
  a \bullet b.\star \lra  (a \times b).\star
\).

\section{Correctness}\label{sec:correctness}
We now prove the subject reduction, confluence, strong normalisation, and introduction
properties of the $\OC$-calculus.

\subsection{Subject reduction}
\label{sec:SR}
The subject reduction property is not completely trivial.  As noted in the
introduction, we commute the sum rule with the introductions \emph{as much as
possible}. It is not possible in the case of the connectives $\otimes$ and
$\oplus$, since it would break subject reduction.  For example, the rule $(t\otimes u)\plus(v\otimes w)\lra (t\plus v)\otimes(u\plus w)$ would not be valid, as we have
\(
  \varnothing;x^A,y^A\vdash (x\otimes y)\plus (y\otimes x):A\otimes A
\)
but
\(
  \varnothing;x^A,y^A \nvdash (x \plus y)\otimes (y \plus x):A\otimes A
\),
since the rule $\otimes_i$ is multiplicative. 

\begin{restatable}[Subject reduction]{theorem}{SR}
%\begin{theorem}[Subject reduction]
  \label{thm:SR}
  If $\Xi; \Gamma \vdash t:A$ and $t \lra u$, then $\Xi; \Gamma \vdash u:A$.
%\end{theorem}
\end{restatable}
\begin{proof}
  By induction on the relation $\lra$. The proof is given in Appendix~\ref{proof:SR}.
  \qed
\end{proof}

\subsection{Confluence}

\begin{theorem}[Confluence]
\label{thm:Confluence}
The $\OC$-calculus is confluent.
\end{theorem}

\begin{proof}
The reduction system of Figure~\ref{fig:figureductionrules} applied to
well-formed proof-terms is left linear and has no critical pairs.  By
\cite[Theorem 6.8]{Nipkow}, it is confluent.\qed
\end{proof}

\subsection{Strong normalisation}\label{sec:ST}
We now prove that all reduction sequences are finite.  To handle the symbols
$\plus$ and $\bullet$ and the associated reduction rules, we prove the strong
normalisation of an extended reduction system, in the spirit of Girard's
ultra-reduction\footnote{Ultra-reduction is used, in particular, in the
adequacy of the connectives $\plus$ and $\bullet$.}~\cite{GirardPhDThesis},
whose strong normalisation obviously implies that of the rules of
Figure~\ref{fig:figureductionrules}.

\begin{definition}[Ultra-reduction]
  Ultra-reduction is defined with the rules of Figure~\ref{fig:figureductionrules},
  plus the rules
  \[
    t \plus u  \lra  t \qquad\qquad
    t \plus u  \lra  u \qquad\qquad
    a \bullet t  \lra  t
  \]
\end{definition}

Our proof is an extension from the proof of the $\mathcal L^{\mathcal
S}$-calculus~\cite{DiazcaroDowekFSCD22,DiazcaroDowek2023} and that of System F, using the methods
introduced by Tait~\cite{TaitJSL67} for Gödel's System T and generalised to
System F by Girard~\cite{GirardPhDThesis}.

\begin{definition}
  $\SN$ is the set of strongly normalising proof-terms and $Red(t)$ is the set of one-step reducts of $t$. That is,
  $\SN = \{ t \mid t\textrm{ strongly normalises}\}$ and $Red(t) = \{u\mid t\lra u\}$.
\end{definition}

\begin{definition}[Reducibility candidates]
  A set of proof-terms $E$ is a reducibility candidate if and only if the following conditions are satisfied.

    (CR1) $E \subseteq \SN$.

    (CR2) If $t \in E$ and $t \lra t'$, then $t' \in E$.

    (CR3) If $t$ is not an introduction and $Red(t) \subseteq E$, then $t \in E$.

    (CR4) If $t \in E$, then for every $X$ and $A, (A/X)t \in E$.
\end{definition}

\begin{definition}
  Let $E,F$ be sets of proof-terms. We define the following sets.
  \begin{align*}
    E \hatmultimap F &= \{t \in \SN \mid \textrm{if }t \lras \lambda x^A.u\textrm{, then for every }v \in E, (v/x)u \in F\}\\
    E \hatotimes F &= \{t \in \SN \mid \textrm{if }t \lras u \otimes v\textrm{, then }u \in E\textrm{ and }v \in F\}\\
    E \hatand F &= \{t \in \SN \mid \textrm{if }t \lras \langle u,v \rangle\textrm{, then }u \in E\textrm{ and }v \in F\}\\
    E \hatoplus F &= \{t \in \SN \mid \textrm{if }t \lras \inl(u)\textrm{,\,then\,}u \in E\textrm{\,and if }t \lras \inr(v)\textrm{,\,then }v \in F\}\\
    \hatbang E &= \{t \in \SN \mid \textrm{if }t \lras \oc u\textrm{, then }u \in E\}
  \end{align*}
\end{definition}

The set of all reducibility candidates is called $\mathcal R$. A valuation $\rho$ is a map from proposition variables to $\mathcal R$.

\begin{definition}
  For any proposition $A$ and valuation $\rho$, the set of proof-terms $\llbracket A \rrbracket_\rho$ is defined as follows:
%  \begin{align*}
%    \llbracket X \rrbracket_\rho &= \rho(X)\\
%    \llbracket \one \rrbracket_\rho &= \SN\\
%    \llbracket A \multimap B \rrbracket_\rho &= \llbracket A \rrbracket_\rho \hatmultimap \llbracket B \rrbracket_\rho\\
%    \llbracket A \otimes B \rrbracket_\rho &= \llbracket A \rrbracket_\rho \hatotimes \llbracket B \rrbracket_\rho\\
%    \llbracket \top \rrbracket_\rho &= \SN\\
%    \llbracket \zero \rrbracket_\rho &= \SN\\
%    \llbracket A \with B \rrbracket_\rho &= \llbracket A \rrbracket_\rho \hatand \llbracket B \rrbracket_\rho\\
%    \llbracket A \oplus B \rrbracket_\rho &= \llbracket A \rrbracket_\rho \hatoplus \llbracket B \rrbracket_\rho\\
%    \llbracket \oc A \rrbracket_\rho &= \hatbang\llbracket A\rrbracket_\rho\\
%    \llbracket \forall X.A \rrbracket_\rho &=
%    \begin{aligned}[t]
%      \{ & t \in \SN \mid \text{if }t \lras \Lambda X.u\text{, then for every }E \in \mathcal R \\
%      &\text{and every proposition }B, (B/X)u \in \llbracket A \rrbracket_{\rho, E/X}\}
%    \end{aligned}
%  \end{align*}
  \[
    \begin{array}{rl@{\quad}rl}
      \llbracket X \rrbracket_\rho &= \rho(X) & 
      \llbracket \top \rrbracket_\rho = \llbracket \zero \rrbracket_\rho &= \SN\\
      \llbracket \one \rrbracket_\rho &= \SN &
      \llbracket A \with B \rrbracket_\rho &= \llbracket A \rrbracket_\rho \hatand \llbracket B \rrbracket_\rho\\
      \llbracket A \multimap B \rrbracket_\rho &= \llbracket A \rrbracket_\rho \hatmultimap \llbracket B \rrbracket_\rho &
      \llbracket A \oplus B \rrbracket_\rho &= \llbracket A \rrbracket_\rho \hatoplus \llbracket B \rrbracket_\rho\\
      \llbracket A \otimes B \rrbracket_\rho &= \llbracket A \rrbracket_\rho \hatotimes \llbracket B \rrbracket_\rho & 
      \llbracket \oc A \rrbracket_\rho &= \hatbang\llbracket A\rrbracket_\rho\\
      \llbracket \forall X.A \rrbracket_\rho & 
      \multicolumn{3}{l}{=
	\begin{aligned}[t]
	  &\{ t \in \SN \mid \text{if }t \lras \Lambda X.u\text{, then for every }E \in \mathcal R \\
	  &\text{and every proposition }B, (B/X)u \in \llbracket A \rrbracket_{\rho, E/X}\}
	\end{aligned}
      }
    \end{array}
  \]
\end{definition}

\begin{restatable}{lemma}{typeinterpretationsarerc}
%\begin{lemma}
  \label{lem:typeinterpretationsarerc}
  For any proposition $A$ and valuation $\rho$, $\llbracket A \rrbracket_\rho \in \mathcal R$.
%\end{lemma}
\end{restatable}
\begin{proof}
  By induction on $A$. The proof is given in Appendix~\ref{proof:ST}.
  \qed
\end{proof}

\begin{lemma}[Variables]
  \label{lem:Var}
  For any proposition $A$ {and any valuation $\rho$}, the set $\llbracket A \rrbracket_{\rho}$ contains all the proof-term variables.
\end{lemma}
\begin{proof}
  By Lemma~\ref{lem:typeinterpretationsarerc}, $\llbracket A \rrbracket_\rho \in \mathcal R$. Since $Red(x)$ is empty, by CR3, $x \in \llbracket A \rrbracket_\rho$.\qed
\end{proof}

\begin{restatable}[Adequacy]{theorem}{adequacy}
  \label{thm:adequacy}
  If $\Xi;\Gamma\vdash t:A$, then for any $\rho$ and $\sigma$ such that for each $x^B\in\Xi\cup\Gamma$, $\sigma(x)\in\llbracket B\rrbracket_\rho$, we have
%  Let $t$ be a proof-term of $A$ in a context ${\Xi;}\Gamma = x_1:A_1, ...,
%  x_n:A_n$, { $\rho$ a valuation} and $\sigma$ be a substitution mapping each
%  variable $x_i$ to an element of $\llbracket A_i \rrbracket_{\rho}$, then
  $\sigma t \in \llbracket A \rrbracket_{\rho}$.
\end{restatable}
  \begin{proof}
  By induction on $t$.
If $t$ is a variable, then, by the definition of $\sigma$, $\sigma t \in
\llbracket A \rrbracket_{\rho}$. For the other proof-term constructors, we use the
adequacy lemmas provided in Appendix~\ref{proof:Adequacy}, for example
%\begin{itemize}
 % \item If $t = \Lambda X.u$, where $u$ is a proof-term of $B$, then, by induction
  %  hypothesis, $\sigma u \in \llbracket B \rrbracket_{\rho,E/X}$.  Hence, by
  %  Lemma~\ref{lem:Lambda}, $\Lambda X.{\sigma u} \in \llbracket \forall X.B
  %  \rrbracket_{\rho}$, that is $\sigma t \in \llbracket A
  %  \rrbracket_{\rho}$.
%
%    \item If
if
  $t = u~B$, where $u$ is a proof-term of $\forall X.C$, then, by
    induction hypothesis, $\sigma u \in \llbracket \forall X.C
    \rrbracket_{\rho}$. Hence, by Lemma~\ref{lem:typeapplication}, $\sigma
    u~B\in \llbracket (B/X)C \rrbracket_{\rho}$, that is $\sigma t \in
    \llbracket A \rrbracket_{\rho}$. \qed
%  \end{itemize}
\end{proof}

\begin{corollary}[Strong normalisation]\label{cor:ST}
  If $\Xi;\Gamma\vdash t:A$, then,
  $t\in\SN$.
\end{corollary}

\begin{proof}
  By Lemma~\ref{lem:Var}, for each $x^B\in\Xi\cup\Gamma$, $x^{B}\in\llbracket B \rrbracket_{\rho}$.  Then, by Theorem~\ref{thm:adequacy}, $t = \mathsf{id}(t) \in \llbracket A \rrbracket_{\rho}$. Hence, by Lemma~\ref{lem:typeinterpretationsarerc}, $t\in\SN$. \qed
\end{proof}

\subsection{Introduction property}\label{sec:introductionthm}

\begin{restatable}[Introduction]{theorem}{introductions}
%  \begin{theorem}[Introduction]
  \label{thm:introductions}
  Let $t$ be a closed irreducible proof-term of $A$.
  \begin{itemize}
    \item The proposition $A$ is not $X$.

    \item If $A$ is $\one$, then $t$ has the form $a.\star$.

    \item If $A$ has the form $B \multimap C$, then $t$ has the form
      $\lambda x^B.u$.

    \item If $A$ has the form $B \otimes C$, then $t$ has the form $u \otimes v$,
      $u \plus v$, or $a \bullet u$.

    \item If $A$ is $\top$, then $t$ is $\langle \rangle$.

    \item The proposition $A$ is not $\zero$.

    \item If $A$ has the form $B \with C$, then $t$ has the form
      $\pair{u}{v}$.

    \item If $A$ has the form $B \oplus C$, then $t$ has the form $\inl(u)$, $\inr(u)$, $u \plus v$, or $a \bullet u$.
    \item If $A$ has the form $\oc B$, then $t$ has the form $\oc u$.
    \item If $A$ has the form $\forall X.B$, then $t$ has the form $\Lambda \abstr{X}u$.
  \end{itemize}
%  \end{theorem}
\end{restatable}
\begin{proof}
  By induction on $t$. The proof is given in
Appendix~\ref{proof:introductionthm}.
\qed
\end{proof}

\section{Encodings}\label{sec:secvectorsmatrices}
In this section, we present the encodings for vectors
(Section~\ref{sec:secvectors}) and matrices (Section~\ref{sec:secmatrices}),
which were initially introduced in~\cite{DiazcaroDowekFSCD22,DiazcaroDowek2023}
and are replicated here for self-containment.  Additionally, we provide one
example that uses polymorphism for iteration as an illustration
(Section~\ref{sec:ExPoly}).

Since $\mathcal S$ is a semiring, we work with semimodules.
However, to help intuition, the reader may think of $\mathcal S$ as a field,
obtaining a vector space instead.

%In this section, we treat the set $\mathcal{S}$ as a field instead of a
%semiring to enhance intuition. However, all the results, except for
%Definition~\ref{def:addinv} and item~\ref{it:neg} of
%Lemma~\ref{lem:vecstructure}, are also valid for semirings, with a semimodule
%replacing a vector space.

\subsection{Vectors}
\label{sec:secvectors}

As there is one rule $\one_i$ for each scalar $a$, there is one closed
irreducible proof-term $a.\star$ for each scalar $a$.  Thus, the closed irreducible
proof-terms $a.\star$ of $\one$ are in one-to-one correspondence with the elements
of ${\mathcal S}$.  Therefore, the proof-terms $\pair{a.\star}{b.\star}$ of $\one
\with \one$ are in one-to-one correspondence with the elements of ${\mathcal S}^2$, the
proof-terms $\pair{\pair{a.\star}{b.\star}}{c.\star}$ of $(\one \with \one) \with
\one$, and also the proof-terms $\pair{a.\star}{\pair{b.\star}{c.\star}}$ of $\one
\with (\one \with \one)$, are in one-to-one correspondence with the elements
of ${\mathcal S}^3$, etc.

\begin{definition}[The set ${\mathcal V}$]
  The set ${\mathcal V}$ is inductively defined as follows: $\one \in
  {\mathcal V}$, and if $A$ and $B$ are in ${\mathcal V}$, then so is $A
  \with B$.
\end{definition}

We now show that if $A \in {\mathcal V}$, then the set of closed
irreducible proof-terms of $A$ has a structure of $\mathcal S$-semimodule.

\begin{definition}[Zero vector]
  If $A \in {\mathcal V}$, we define the proof-term $0_A$ of $A$ by induction
  on $A$.  If $A = \one$, then $0_A = 0.\star$.  If $A = A_1 \with A_2$,
  then $0_A = \pair{0_{A_1}}{0_{A_2}}$.
  \end{definition}
  
  %\begin{definition}[Additive inverse]
  %  \label{def:addinv}
  %If $A \in {\mathcal V}$, and $t$ is a proof-term of $A$, we define the
  %proof-term $- t$ of $A$ by induction on $A$.  If $A = \one$, then $t$
  %reduces to $a.\star$, we let $- t = (-a).\star$. If $A = A_1 \with
  %A_2$, $t$ reduces to $\pair{t_1}{t_2}$ where $t_1$ is a proof-term of $A_1$
  %and $t_2$ of $A_1$. We let $-t = \pair{- t_1}{- t_2}$.
  %\end{definition}

\begin{lemma}[$\mathcal S$-semimodule structure~\citeintitle{Lemma 3.4}{DiazcaroDowekFSCD22}] \label{lem:vecstructure}
  If $A \in {\mathcal V}$ and $t$, $t_1$, $t_2$, and $t_3$ are closed proof-terms of
  $A$, then
  \begin{multicols}{2}
    \begin{enumerate}
      \item $(t_1 \plus t_2) \plus t_3 \equiv t_1 \plus (t_2 \plus t_3)$
      \item $t_1 \plus t_2 \equiv t_2 \plus t_1$
      \item $t \plus 0_A \equiv t$
      \item $a \bullet b \bullet t \equiv (a \times b) \bullet t$
      \item $1 \bullet t \equiv t$
      \item $a \bullet (t_1 \plus t_2) \equiv a \bullet t_1 \plus a \bullet t_2$
      \item $(a + b) \bullet t \equiv a \bullet t \plus b \bullet t$
	\qed
    \end{enumerate}
  \end{multicols}
\end{lemma}

\begin{definition}[Dimension of a proposition in ${\mathcal V}$]
To each proposition $A \in {\mathcal V}$, we associate a positive
natural number $d(A)$, which is the number of occurrences of the
symbol $\one$ in $A$: $d(\one) = 1$ and $d(B
\with C) = d(B) + d(C)$.
\end{definition}

If $A \in {\mathcal V}$ and $d(A) = n$, then the closed irreducible proof-terms
of $A$ and the vectors of ${\mathcal S}^n$ are in one-to-one
correspondence.

\begin{definition}[One-to-one correspondence]
\label{def:onetoone}
Let $A \in {\mathcal V}$ with $d(A) = n$.  To each closed irreducible
proof-term $t$ of $A$, we associate a vector $\underline{t}$ of ${\mathcal
  S}^n$ as follows:

  If $A = \one$, then $t = a.\star$. We let $\underline{t} =
  \left(\begin{smallmatrix} a \end{smallmatrix}\right)$.

  If $A = A_1 \with A_2$, then $t = \pair{u}{v}$.  We let
  $\underline{t}$ be the vector with two blocks $\underline{u}$ and
  $\underline{v}$: $\underline{t} = \left(\begin{smallmatrix}
    \underline{u}\\\underline{v} \end{smallmatrix}\right)$.

To each vector ${\bf u}$ of ${\mathcal S}^n$, we associate a closed irreducible proof-term $\overline{\bf u}^A$ of $A$ as follows:

If $n = 1$, then ${\bf u} = \left(\begin{smallmatrix}
  a \end{smallmatrix}\right)$. We let $\overline{\bf u}^A = a.\star$.

If $n > 1$, then $A = A_1 \with A_2$, let $n_1$ and $n_2$ be
  the dimensions of $A_1$ and $A_2$.  Let ${\bf u}_1$ and ${\bf u}_2$
  be the two blocks of ${\bf u}$ of $n_1$ and $n_2$ lines, so ${\bf u}
  = \left(\begin{smallmatrix} {\bf u}_1\\ {\bf
      u}_2\end{smallmatrix}\right)$.  We let $\overline{\bf u}^A =
    \pair{\overline{{\bf u}_1}^{A_1}}{\overline{{\bf u}_2}^{A_2}}$.
\end{definition}

We extend the definition of $\underline{t}$ to any closed proof-term of
$A$, $\underline{t}$ is by definition $\underline{t'}$ where $t'$ is
the irreducible form of $t$.

The next theorem shows that the symbol $\plus$ expresses the sum of
vectors and the symbol $\bullet$, the product of a vector by a scalar.

\begin{theorem}[Sum and scalar product of vectors~\citeintitle{Lemmas 3.7 and 3.8}{DiazcaroDowekFSCD22}]
\label{thm:parallelsum}
Let $A \in {\mathcal V}$, $u$ and $v$ two closed proof-terms of $A$, and $a\in\mathcal S$ a scalar.
Then, $\underline{u \plus v} = \underline{u} + \underline{v}$ and
$\underline{a \bullet u} = a \underline{u}$.
\qed
\end{theorem}

\subsection{Matrices}
\label{sec:secmatrices}

We now want to prove that if $A, B \in {\mathcal V}$ with $d(A) = m$
and $d(B) = n$, and $F$ is a linear function from ${\mathcal S}^m$ to
${\mathcal S}^n$, then there exists a closed proof-term $f$ of $A
\multimap B$ such that, for all vectors ${\bf u} \in {\mathcal S}^m$,
$\underline{f~\overline{\bf u}^A} = F({\bf u})$.  This can
equivalently be formulated as the fact that if $M$ is a matrix with
$m$ columns and $n$ lines, then there exists a closed proof-term $f$ of $A
\multimap B$ such that for all vectors ${\bf u} \in {\mathcal S}^m$,
$\underline{f~\overline{\bf u}^A} = M {\bf u}$.

A similar theorem has been proved also in~\cite{odot} for a non-linear
calculus. 

\begin{theorem}[Matrices~\citeintitle{Theorem 4.1}{DiazcaroDowekFSCD22}]
  \label{thm:matrices}
  Let $A, B \in {\mathcal V}$ with $d(A) = m$ and $d(B) = n$ and let $M$
  be a matrix with $m$ columns and $n$ lines, then there exists a closed
  proof-term $t$ of $A \multimap B$ such that, for all vectors ${\bf u} \in
  {\mathcal S}^m$, $\underline{t~\overline{\bf u}^A} = M {\bf u}$.
  \qed
\end{theorem}

\begin{example}[Matrices with two columns and two lines]
  The matrix
  $\left(\begin{smallmatrix} a & c\\b & d \end{smallmatrix}\right)$
  is expressed as the proof-term
    $t = \lambda x^{\one \with \one}. \elimwith^1(x,y^\one.
    \elimone(y,\pair{a.\star}{b.\star})) \plus \elimwith^2(x,z^\one.
    \elimone(z,\pair{c.\star}{d.\star}))$.
  Then,
    $t~\pair{e.\star}{f.\star} \lras
    \pair{(a \times e + c \times f).\star}{(b \times e + d \times f).\star}$.
\end{example}

\subsection{Matrix iterator}\label{sec:ExPoly}
The polymorphic extension included in the $\OC$-calculus allows us to encode, for example,
natural numbers in a usual way~\cite[Chapter 5]{Girard87}. 
\[
  \begin{array}{rl@{\qquad}rl}
    \mathsf{Nat} &= \forall X.X \multimap \oc (X \multimap X) \multimap X & 
    \mathsf{zero} &= \Lambda X.\lambda x^X.\lambda f^{\oc (X \multimap X)}.\elimbang(f, f'^{X \multimap X}.x)\\
    \multicolumn{4}{c}{\mathsf{succ} = \lambda n^\mathsf{Nat}.\Lambda X.\lambda x^X.\lambda f^{\oc (X \multimap X)}.\elimbang(f, f'^{X \multimap X}.f' (n~X~x~(\oc f')))}
  \end{array}
\]
  We can express the application $n$ times of a square matrix over a vector as
  follows, where $A\in\mathcal V$ with $d(A)=m$.

  \[
    \mathsf{Miter}= \lambda n^\mathsf{Nat}.\lambda m^{\oc(A\multimap A)}.\lambda v^A.n~A~v~m
  \]
  Let $M$ be a square matrix with $m$ columns and lines, and
  $t$ be the closed proof-term of $A\multimap A$ representing such a
  matrix.  For any vector ${\bf u}\in\mathcal S^m$ we have
  \(
    \underline{\mathsf{Miter}~\hat{n}~\oc t~\overline{\bf u}^A} = M^n{\bf u}
  \)
  where $\hat n$ is the encoding of $n\in\mathbb N$.

\section{Linearity}
\label{sec:seclinearity}
In this section, we prove the converse to Theorem~\ref{thm:matrices}, that is,
that if $A,B \in {\mathcal V}$, then each closed proof-term $t$ of $A \multimap B$
expresses a linear function.

This result is trivially false when we consider the exponential connective.  For example, the
proof-term $f = \lambda x^{\oc\one}.\elimbang(x, y^{\one}. 2.\star)$ of proposition
$!\one\multimap\one$ represents the constant function $2$, which is not linear.
Indeed,
\[
  f~(\oc (a.\star) \plus \oc (b.\star))
  \lra^*
  2.\star
  \neq
  4.\star
  \llas
  (f~\oc a.\star) \plus (f~\oc b.\star)
\]

Hence, this section refers to the $\OC$-calculus without $\oc$. That is, we remove the propositions $\oc A$ and the
proof-terms $\oc t$ and $\elimbang(t, x^A.u)$ from the syntax (cf.
Section~\ref{sec:secls}), together with their three corresponding reduction
rules (cf. Figure~\ref{fig:figureductionrules}). In addition, we remove the
deduction rules ax, $\oc_i$, and $\oc_e$ (cf.
Figure~\ref{fig:figuretypingrules}). With these changes, all interesting
sequents have the shape $\varnothing;\Gamma\vdash t:A$; thus, we simply write
$\Gamma\vdash t:A$ instead.

Note that the $\forall$ connective is meaningful only when a specific subproof-term is employed more than once. Within the linear deduction system, this can happen in two scenarios:
(1) In the multiplicative case, applying the $\oc$ constructor to the subproof-term becomes necessary for a valid proof-term.
(2) Conversely, in the additive case, the subproof-term can be utilised multiple times without requiring the $\oc$ constructor.
Hence, the polymorphic construction retains significance even in the absence of the exponential connective.

%  For example, consider the following proof-term of the proposition $(\forall X.X \multimap X) \multimap (\one \multimap \one) \with (\one \multimap \one\multimap\one)$.
%  \[
%    \lambda x^{\forall X.X \multimap X}.\langle x~(\one \multimap \one)~(\lambda y^\one.y), x~(\one \multimap \one\multimap\one)~(\lambda z^\one.\lambda w^\one.\elimone(z,w))\rangle
%  \]
%  Taking an arbitrary proof-term of $\forall X.X\multimap X$, we can use rule $\multimap_e$ to obtain a correct proof-term of $(\one \multimap \one) \with (\one \multimap \one\multimap\one)$.
%

This section extends the proof for the ${\mathcal L}^{\mathcal
S}$-calculus~\cite{DiazcaroDowekFSCD22,DiazcaroDowek2023}, adding the
polymorphic cases. For self-containment, we include all the definitions.

\subsection{Observational equivalence}

We want to prove that for any closed proof-term $t$ of $A \multimap B$, if $u_1$ and
$u_2$ are closed proof-terms of $A$, then 
\begin{equation}
  \label{eq:equiv}
  t~(u_1 \plus u_2) \equiv t~u_1 \plus t~u_2 
  \qquad\qquad\textrm{and}\qquad\qquad
  t~(a \bullet u_1) \equiv a \bullet t~u_1
\end{equation}

The property, however, is not true in general. Consider, for example,
\(
  t = \lambda x^\one.\lambda y^{\one\multimap\one}. y~x
\)
and we have 
\[
  t~(1.\star \plus\,2.\star)
  \lras \lambda y^{\one\multimap\one}. y~3.{\star}
  \not\equiv
  \lambda y^{\one\multimap\one}. {(y~1.\star)} \plus {(y~2.\star)}
  \llas
  (t~{1.\star}) \plus (t~{2.\star})
\]

Nevertheless, although the proof-terms $\lambda y^{\one\multimap\one}.
y~3.\star$ and $\lambda y^{\one\multimap\one}. {(y~1.\star)} \plus
{(y~2.\star)}$ are not equivalent, if we put them in the context $\_~\lambda
z^\one. z$, then both proof-terms
$(\lambda y^{\one\multimap\one}.
y~3.\star)~\lambda z^\one. z$ and $(\lambda y^{\one\multimap\one}.
{(y~1.\star)} \plus {(y~2.\star)})~\lambda z^\one. z$ reduce to
$3.\star$.  This leads us to introduce a notion of observational equivalence.

\begin{definition}[Observational equivalence]\label{def:obseq}
  Two proof-terms $t_1$ and $t_2$ of a proposition $B$ are {\em observationally
  equivalent}, $t_1 \sim t_2$, if for all propositions $C$ in ${\mathcal V}$ and
  for all proof-terms $c$ such that $\_^B \vdash c:C$, we have 
  \(
    c\{t_1\} \equiv c\{t_2\}
  \).
\end{definition}

We shall prove (Corollary~\ref{cor:corollary0}) that for all proof-terms $t$ of
proposition $A\multimap B$ and for all closed proof-terms $u_1$ and $u_2$ of $A$, we
have
\(
  t~(u_1 \plus u_2) \sim t~u_1 \plus t~u_2
\)
and
\(
  t~(a\bullet u_1) \sim a \bullet t~u_1
\).
However, a proof of this property by induction on $t$ does not go through and to prove
it, we first prove Theorem~\ref{thm:linearity}, expressing that for all proof-terms
$t$ of $A\multimap B$, with $B \in {\mathcal V}$, and closed proof-terms
$u_1$ and $u_2$ of $A$, we have the property~\eqref{eq:equiv}.

\subsection{Measure of a proof-term}

We define the following non-increasing measure (see Lemma~\ref{lem:mured}), over which we will make the induction to prove
the linearity theorem (Theorem~\ref{thm:linearity}).

\begin{definition}[Measure of a proof-term]
  \label{def:measureofaproof}~
  We define the measure $\mu$ as follows:
  \[
    \begin{array}{rl@{\qquad }rl}
      \mu(x) &= 0 &
      \mu(t \plus u) &= 1 + \max(\mu(t), \mu(u)) \\
      \mu(a \bullet t) &= 1 + \mu(t) &
      \mu(a.\star) &= 1 \\
      \mu(\elimone(t,u)) &= 1 + \mu(t) + \mu(u) &
      \mu(\lambda x^A.t) &= 1 + \mu(t) \\
      \mu(t~u) &= 1 + \mu(t) + \mu(u) &
      \mu(t \otimes u) &= 1 + \mu(t) + \mu(u) \\
      \mu(\elimtens(t,x^A y^B.u)) &= 1 + \mu(t) + \mu(u) &
      \mu(\topintro) &= 1\\
      \mu(\elimzero(t)) &= 1 + \mu(t) &
      \mu(\pair{t}{u}) &= 1 + \max(\mu(t), \mu(u))\\
      \mu(\elimwith^1(t,y^A.u)) &= 1 + \mu(t) + \mu(u) &
      \mu(\elimwith^2(t,y^A.u)) &= 1 + \mu(t) + \mu(u)\\
      \mu(\inl(t)) &= 1 + \mu(t) &
      \mu(\inr(t)) &= 1 + \mu(t)\\
      \multicolumn{4}{c}{\mu(\elimplus(t,y^A.u,z^B.v)) = 1 + \mu(t) + \max(\mu(u), \mu(v))} \\
      \mu(\Lambda X.t) &=1+\mu(t) &
      \mu(t~A) &=1+\mu(t)
    \end{array}
  \]
%  \[
%    \mu(s) = 
%    \left\{
%      \begin{array}{l@{\ \ }l}
%	0 & \textrm{if }s = x\\
%	1 & \textrm{if }s\in\{a.\star,\topintro\}\\
%	1+\mu(t) & \textrm{if }s\in
%	\left\{
%	  \begin{array}{c}
%	    a \bullet t,
%	    \lambda \abstr{x:A} t,
%	    \elimzero(t),
%	    \\
%	    \inl(t),
%	    \inr(t),
%	    \Lambda X.t,
%	    t~A
%	  \end{array}
%	\right\}\\
%	1 + \mu(t) + \mu(u) & \textrm{if }s\in
%	\left\{
%	  \begin{array}{c}
%	    \elimone(t,u),
%	    t~u,
%	    t \otimes u,\\
%	    \elimtens(t,\abstr{x:A}\abstr{y:B} u),\\
%	    \elimwith^1(t,\abstr{y:A}u),
%	  \elimwith^2(t,\abstr{y:A}u))
%	\end{array}
%	\right\}\\
%	1 + \max(\mu(t), \mu(u)) & \textrm{if }s\in\{t\plus u,\pair{t}{u}\}\\
%	1 + \mu(t) + \max(\mu(u), \mu(v)) & \textrm{if }s = \elimplus(t,\abstr{y:A}u,\abstr{z:B}v)
%      \end{array}
%    \right.
%  \]
\end{definition}

\subsection{Elimination contexts}

Any proof-term in the linear fragment of the $\OC$-calculus can be decomposed into a sequence of
elimination rules, forming an elimination context, and a proof-term $u$ that is
either a variable, an introduction, a sum, or a product.

\begin{definition}[Elimination context]
  An elimination context is a proof-term with a single free variable, written
  $\_$, that is a proof-term in the language
  \begin{align*}
    K =& \_
    \mid \elimone(K,u)
    \mid K~u
    \mid \elimtens(K,x^A y^B.v)
    \mid \elimzero(K)\\
    \mid& \elimwith^1(K,x^A.r)
    \mid \elimwith^2(K,x^B.r)
    \mid \elimplus(K,x^A.r,y^B.s)
    \mid K~A
  \end{align*}
  where $u$ is a closed proof-term,
  $\fv(v) = \{x,y\}$,
  $\fv(r) \subseteq \{x\}$, and $\fv(s) \subseteq \{y\}$.
\end{definition}

\begin{restatable}[Decomposition of a proof-term]{lemma}{decomp}
  %\begin{lemma}[Decomposition of a proof-term]
    \label{lem:elim}
    If $t$ is an irreducible proof-term such that $x^C \vdash t:A$, then there
    exist an elimination context $K$, a proof-term $u$, and a proposition $B$,
    such that $\_^B \vdash K:A$, $x^C \vdash u:B$, $u$ is either the
    variable $x$, an introduction, a sum, or a product, and $t = K\{u\}$.
  %\end{lemma}
  \end{restatable}
  \begin{proof}
    See Appendix~\ref{proof:seclinearity}.
    \qed
  \end{proof}

  \subsection{Linearity}
  
  We now have the tools to prove the linearity theorem.  We first prove
  (Theorem~\ref{thm:linearity}) that for any proof-term $t$ of $B\in\mathcal V$ such
  that $x^A \vdash t:B$, we have the property~\eqref{eq:equiv}.
  Then, Corollary~\ref{cor:corollary0} generalises the proof-term for any $B$ stating
  the observational equivalence. Finally, Corollary~\ref{cor:corollary2} is just a
  reformulation of it, in terms of linear functions.

\begin{restatable}[Linearity]{theorem}{linearity}
%\begin{theorem}[Linearity]
  \label{thm:linearity}
  If $A$ is a proposition, $B$ is proposition of ${\mathcal V}$, $t$ is
  a proof-term such that $x^A \vdash t:B$ and $u_1$ and $u_2$ are two closed
  proof-terms of $A$, then
  \(
    t\{u_1 \plus u_2\} \equiv t\{u_1\} \plus t\{u_2\}
  \)
  and
  \(
    t\{a \bullet u_1\} \equiv a \bullet t\{u_1\}
  \).
%\end{theorem}
\end{restatable}
\begin{proof}[Sketch, see Appendix~\ref{proof:seclinearity}]
  Without loss of generality, we can assume that $t$ is irreducible.  We
  proceed by induction on $\mu(t)$, using Lemma~\ref{lem:elim}, to decompose the proof-term $t$ into $K\{t'\}$ where
  $t'$ is either the variable $x$, an introduction, a sum, or a product. Then, we analyse case by case.
  \qed
\end{proof}

We can now generalise the linearity result, as explained at the beginning of the section, by using the observational equivalence $\sim$. 
%The proof given in Appendix~\ref{proof:seclinearity}.

\begin{restatable}[\hspace{-.1mm}{\protect\cite[Corollary 4.11]{DiazcaroDowek2023}}]{corollary}{corollarylin}
  %\begin{corollary}[Linearity~\citeintitle{Corollary 4.11}{DiazcaroDowek2023}]
    \label{cor:corollary0}
    If $A$ and $B$ are any propositions,
    $t$ a proof-term such that $x^A \vdash t:B$, and
    $u_1$ and $u_2$ two closed proof-terms of $A$, then  
    \(
      t\{u_1 \plus u_2\} \sim t\{u_1\} \plus t\{u_2\}
    \)
    and
    \(
      t\{a \bullet u_1\} \sim a \bullet t\{u_1\}
    \).
  %\end{corollary}
  \end{restatable}
  \begin{proof}
    See Appendix~\ref{proof:seclinearity}.
    \qed
  \end{proof}

Finally, the next corollary is the converse of Theorem~\ref{thm:matrices}.

\begin{corollary}
  \label{cor:corollary2}
  Let $A, B \in {\mathcal V}$, such that $d(A) = m$ and $d(B) = n$, 
  and  $t$ be a closed proof-term of $A \multimap B$.
  Then the function $F$ from ${\mathcal S}^m$ to ${\mathcal S}^n$,
  defined as
  $F({\bf u}) = \underline{t~\overline{\bf u}^A}$ is linear.
\end{corollary}

\begin{proof}
  Using Corollary~\ref{cor:corollary2} and Theorem~\ref{thm:parallelsum}, we have
  \begin{align*}
    F({\bf u} + {\bf v}) &= \underline{t~\overline{\bf u + \bf v}^A}
    = \underline{t~(\overline{\bf u}^A \plus \overline{\bf v}^A)} =
    \underline{t~\overline{\bf u}^A \plus t~\overline{\bf v}^A}
    \\
    &=
    \underline{t~\overline{\bf u}^A} +
    \underline{t~\overline{\bf v}^A}
    =
    F({\bf u}) + F({\bf v})
    \\
    F(a {\bf u}) &= \underline{t~\overline{a \bf u}^A}
    = \underline{t~(a \bullet \overline{\bf u}^A)} =
    \underline{a\bullet t~\overline{\bf u}^A}
    = a \underline{t~\overline{\bf u}^A} = a F({\bf u})
    \qedHERE
  \end{align*}
\end{proof}

\section{Conclusion}
In this paper, we have presented the $\OC$-calculus, an extension of the
$\mathcal L^{\mathcal S}$-calculus with second-order polymorphism and the
exponential connective, allowing non-linear functions and making it a more
expressive language.  We have proved all its correctness properties, including
algebraic linearity for the linear fragment.

The $\mathcal L^{\mathcal S}$-calculus was originally introduced as a core
language for quantum computing. Its ability to represent matrices and vectors
makes it suitable for expressing quantum programs when taking $\mathcal
S=\mathbb C$.  Moreover, by taking $\mathcal S=\mathbb R^+$, one can consider a
probabilistic language, and by taking $\mathcal S=\{\star\}$, a linear
extension of the parallel lambda calculus~\cite{BoudolIC94}.

To consider this calculus as a proper quantum language, we would need not only
to ensure algebraic linearity but also to ensure unitarity, using techniques
such as those in~\cite{DiazcaroGuillermoMiquelValironLICS19}.
Also, the language $\mathcal L^{\mathcal S}$ can be extended with a non-deterministic connective $\odot$~\cite{odot}, from which a quantum measurement operator can be encoded. We did not add such a connective to our presentation, to stay in a pure linear logic setting, however, the extension is straightforward.
Another possible future work is to extend the categorical model
of the $\mathcal L^{\mathcal S}$-calculus given in~\cite{DiazcaroMalherbe22}.
To accommodate the $\OC$-calculus, we would need to use hyperdoctrines~\cite{Crole},
following the approach of~\cite{Maneggia}, a direction we are willing to pursue.

\bibliographystyle{abbrv}
\bibliography{polymorphicLScalculus}

\appendix
\section{Proof of Section~\ref{sec:SR}}\label{proof:SR}
We need two substitution lemmas to prove subject reduction (Theorem~\ref{thm:SR}).

\begin{lemma}[Substitution of propositions]
  \label{lem:substitutiontypesonterms}
  If $\Xi; \Gamma \vdash t:A$ and $B$ is a proposition, then
  \(
    (B/X)\Xi; (B/X)\Gamma \vdash (B/X)t:(B/X)A
  \).
\end{lemma}
\begin{proof}
  By induction on $t$.
  \begin{itemize}
    \item If $t = x$, then either $\Gamma = \{x^A\}$ or $\Gamma$ is empty and $x^A \in \Xi$.
      \begin{itemize}
	\item In the first case, by rule lin-ax, $(B/X)\Xi; x:(B/X)A \vdash x:(B/X)A$.
	\item In the second case, by rule ax, $(B/X)\Xi; \varnothing \vdash x:(B/X)A$.
      \end{itemize}

    \item If $t = u \plus v$, then $\Xi;\Gamma \vdash u:A$ and $\Xi;\Gamma \vdash v:A$. By the induction hypothesis, $(B/X)\Xi;(B/X)\Gamma \vdash (B/X)u:(B/X)A$ and $(B/X)\Xi;(B/X)\Gamma \vdash (B/X)v:(B/X)A$. Therefore, by rule sum, $(B/X)\Xi;(B/X)\Gamma \vdash (B/X)(u \plus v):(B/X)A$.

    \item If $t = a \bullet u$, then $\Xi;\Gamma \vdash u:A$. By the induction hypothesis, we have that $(B/X)\Xi;(B/X)\Gamma \vdash (B/X)u:(B/X)A$. Therefore, by rule prod($a$), $(B/X)\Xi;(B/X)\Gamma \vdash (B/X)(a \bullet u):(B/X)A$.

    \item If $t = a.\star$, then $A = \one$ and $\Gamma$ is empty. By rule $\one_i$($a$), $(B/X)\Xi; \varnothing \vdash a.\star:\one$.

    \item If $t = \elimone(u,v)$, then there are $\Gamma_1$ and $\Gamma_2$ such that $\Gamma = \Gamma_1, \Gamma_2$, $\Xi;\Gamma_1 \vdash u:\one$ and $\Xi;\Gamma_2 \vdash v:A$. By the induction hypothesis, $(B/X)\Xi;(B/X)\Gamma_1 \vdash (B/X)u:\one$ and $(B/X)\Xi;(B/X)\Gamma_2 \vdash (B/X)v:(B/X)A$. Therefore, by rule $\one_e$, $(B/X)\Xi;(B/X)\Gamma \vdash (B/X)\elimone(u, v):(B/X)A$.

    \item If $t = \lambda x^C.u$, then $A = C \multimap D$ and $\Xi;\Gamma, x^C \vdash u:D$. By the induction hypothesis, $(B/X)\Xi;(B/X)\Gamma, x^{(B/X)C} \vdash (B/X)u:(B/X)D$. Therefore, by rule $\multimap_i$, $(B/X)\Xi;(B/X)\Gamma \vdash (B/X)\lambda x^C.u:(B/X)(C \multimap D)$.

    \item If $t = u~v$, then there are $\Gamma_1$ and $\Gamma_2$ such that $\Gamma = \Gamma_1, \Gamma_2$, $\Xi;\Gamma_1 \vdash u:C \multimap A$ and $\Xi;\Gamma_2 \vdash v:C$. 
      By the induction hypothesis, $(B/X)\Xi;(B/X)\Gamma_1 \vdash (B/X)u:(B/X)C\!\multimap\!(B/X)A$ and $(B/X)\Xi;(B/X)\Gamma_2\! \vdash\! (B/X)v:(B/X)C$.
      Therefore, by rule $\multimap_e$, $(B/X)\Xi;(B/X)\Gamma \vdash (B/X)(u~v):(B/X)A$.

    \item If $t = u \otimes v$, then $A = C \otimes D$ and there are $\Gamma_1$ and $\Gamma_2$ such that $\Gamma = \Gamma_1, \Gamma_2$, $\Xi;\Gamma_1 \vdash u:C$ and $\Xi;\Gamma_2 \vdash v:D$. By the induction hypothesis, $(B/X)\Xi;(B/X)\Gamma_1 \vdash (B/X)u:(B/X)C$ and $(B/X)\Xi;(B/X)\Gamma_2 \vdash (B/X)v:(B/X)D$. Therefore, by rule $\otimes_i$, $(B/X)\Xi;(B/X)\Gamma \vdash (B/X)(u \otimes v):(B/X)(C \otimes D)$.

    \item If $t = \elimtens(u, x^C y^D.v)$, then there are $\Gamma_1$ and $\Gamma_2$ such that $\Gamma = \Gamma_1, \Gamma_2$, $\Xi;\Gamma_1 \vdash u:C \otimes D$ and $\Xi;\Gamma_2, x^C, y^D \vdash v:A$. By the induction hypothesis, $(B/X)\Xi;(B/X)\Gamma_1 \vdash (B/X)u:(B/X)C \otimes (B/X)D$ and\\ $(B/X)\Xi;(B/X)\Gamma_2, x^{(B/X)C}, y^{(B/X)D} \vdash (B/D)v:(B/X)A$. Therefore, by rule $\otimes_e$, $(B/X)\Xi;(B/X)\Gamma \vdash (B/X)\elimtens(u, x^C y^D.v):(B/X)A$.

    \item If $t = \langle\rangle$, then $A = \top$. By rule $\top_i$, $(B/X)\Xi;(B/X)\Gamma \vdash \langle\rangle:\top$.

    \item If $t = \elimzero(u)$, then there are $\Gamma_1$ and $\Gamma_2$ such that $\Gamma = \Gamma_1, \Gamma_2$ and $\Xi;\Gamma_1 \vdash u:\zero$. By the induction hypothesis, $(B/X)\Xi;(B/X)\Gamma_1 \vdash (B/X)u:\zero$. By rule $\zero_e$, $(B/X)\Xi;(B/X)\Gamma \vdash (B/X)\elimzero(u):(B/X)A$.

    \item If $t = \langle u,v \rangle$, then $A = C \with D$, $\Xi;\Gamma \vdash u:C$ and $\Xi;\Gamma \vdash v:D$. By the induction hypothesis, we have $(B/X)\Xi;(B/X)\Gamma \vdash (B/X)u:(B/X)C$ and $(B/X)\Xi;(B/X)\Gamma \vdash (B/X)v:(B/X)D$. Therefore, by rule $\with_i$, we have $(B/X)\Xi;(B/X)\Gamma \vdash (B/X)\langle u, v \rangle:(B/X)(C \with D)$.

    \item If $t = \elimwith^1(u, x^C.v)$, then there are $\Gamma_1$ and $\Gamma_2$ such that $\Gamma = \Gamma_1, \Gamma_2$, $\Xi;\Gamma_1 \vdash u:C \with D$ and $\Xi;\Gamma_2, x^C \vdash v:A$. By the induction hypothesis, we have $(B/X)\Xi;(B/X)\Gamma_1 \vdash (B/X)u:(B/X)C \with (B/X)D$ and\\ $(B/X)\Xi;(B/X)\Gamma_2, x^{(B/X)C} \vdash (B/X)v:(B/X)A$. Therefore, by rule $\with_{e1}$, $(B/X)\Xi;(B/X)\Gamma \vdash (B/X)\elimwith^1(u, x^C.v):(B/X)A$.

    \item If $t = \elimwith^2(u, x^C.v)$, then there are $\Gamma_1$ and $\Gamma_2$ such that $\Gamma = \Gamma_1, \Gamma_2$, $\Xi;\Gamma_1 \vdash u:D \with C$ and $\Xi;\Gamma_2, x^C \vdash v:A$. By the induction hypothesis, we have $(B/X)\Xi;(B/X)\Gamma_1 \vdash (B/X)u:(B/X)D \with (B/X)C$ and\\ $(B/X)\Xi;(B/X)\Gamma_2, x^{(B/X)C} \vdash (B/X)v:(B/X)A$. Therefore, by rule $\with_{e2}$, $(B/X)\Xi;(B/X)\Gamma \vdash (B/X)\elimwith^2(u, x^C.v):(B/X)A$.

    \item If $t = \inl(u)$, then $A = C \oplus D$ and $\Xi;\Gamma \vdash u:C$. By the induction hypothesis, $(B/X)\Xi;(B/X)\Gamma \vdash (B/X)u:(B/X)C$. By rule $\oplus_{i1}$, $(B/X)\Xi;(B/X)\Gamma \vdash (B/X)\inl(u):(B/X)(C \oplus D)$.

    \item If $t = \inr(u)$, then $A = C \oplus D$ and $\Xi;\Gamma \vdash u:D$. By the induction hypothesis, $(B/X)\Xi;(B/X)\Gamma \vdash (B/X)u:(B/X)D$. By rule $\oplus_{i2}$, $(B/X)\Xi;(B/X)\Gamma \vdash (B/X)\inr(u):(B/X)(C \oplus D)$.

    \item If $t = \elimplus(u, x^C.v, y^D.w)$, then there are $\Gamma_1$ and $\Gamma_2$ such that $\Gamma = \Gamma_1, \Gamma_2$, $\Xi;\Gamma_1 \vdash u:C \oplus D$, $\Xi;\Gamma_2, x^C \vdash v:A$ and $\Xi;\Gamma_2, y^D \vdash w:A$. By the induction hypothesis, we have $(B/X)\Xi;(B/X)\Gamma_1 \vdash (B/X)u:(B/X)C \oplus (B/X)D$, $(B/X)\Xi;(B/X)\Gamma_2, x^{(B/X)C} \vdash (B/X)v:(B/X)A$ and\\ $(B/X)\Xi;(B/X)\Gamma_2, y^{(B/X)D} \vdash (B/X)w:(B/X)A$. Therefore, by rule $\oplus_e$, $(B/X)\Xi;(B/X)\Gamma \vdash (B/X)\elimplus(u, x^C.v, y^D.w):(B/X)A$.

    \item If $t = \oc u$, then $A = \oc C$, $\Gamma$ is empty and $\Xi; \varnothing \vdash u:C$. By the induction hypothesis, $(B/X)\Xi; \varnothing \vdash (B/X)u:(B/X)C$. Therefore, by rule $\oc_i$, $(B/X)\Xi; \varnothing \vdash (B/X)\oc u:(B/X)\oc C$.

    \item If $t = \elimbang(u, x^C.v)$, then there are $\Gamma_1$ and $\Gamma_2$ such that\\ $\Gamma = \Gamma_1, \Gamma_2$, $\Xi; \Gamma_1 \vdash u:\oc C$ and $\Xi, x^C;\Gamma_2 \vdash v:A$. By the induction hypothesis, we have $(B/X)\Xi; (B/X)\Gamma_1 \vdash (B/X)u:(B/X)\oc C$ and\\ $(B/X)\Xi, x^{(B/X)C}; (B/X)\Gamma_2 \vdash (B/X)v:(B/X)A$. Therefore, by rule $\oc_e$, we have $(B/X)\Xi; (B/X)\Gamma \vdash (B/X)\elimbang(u, x^C.v):(B/X)A$.

    \item If $t = \Lambda Y.u$, then $A = \forall Y.C$, $\Xi;\Gamma \vdash u:C$ and $Y \notin \FV(\Xi,\Gamma)$. By the induction hypothesis, $(B/X)\Xi;(B/X)\Gamma \vdash (B/X)u:(B/X)C$. Since $Y \notin \FV((B/X)\Xi, (B/X)\Gamma)$, by rule $\forall_i$ $(B/X)\Xi;(B/X)\Gamma \vdash (B/X)\Lambda Y.u: (B/X)\forall Y.C$.

    \item If $t = u~C$, then $A = (C/Y)D$ and $\Xi;\Gamma \vdash u:\forall Y.D$. By the induction hypothesis, $(B/X)\Xi;(B/X)\Gamma \vdash (B/X)u:\forall Y.(B/X)D$. By rule $\forall_e$, $(B/X)\Xi;(B/X)\Gamma \vdash (B/X)(u~C):(C/Y)(B/X)D$, since $Y\notin \FV(B)$, $(C/Y)(B/X)D=(B/X)(C/Y)D$.
      \qed
  \end{itemize}
\end{proof}

\begin{lemma}[Substitution of proof-terms]
  \label{lem:polysubstitution}
  ~
  \begin{enumerate}
    \item If $\Xi;\Gamma, x^B \vdash t:A$ and $\Xi;\Delta \vdash u:B$, then $\Xi;\Gamma, \Delta \vdash (u/x)t:A$.
    \item If $\Xi, x^B;\Gamma \vdash t:A$ and $\Xi;\varnothing\vdash u:B$, then $\Xi;\Gamma\vdash (u/x)t:A$.
  \end{enumerate}
\end{lemma}
\begin{proof}
  ~
  \begin{enumerate}
    \item By induction on $t$.
      \begin{itemize}
	\item If $t = x$, then $\Gamma$ is empty and $A = B$. Thus, $\Xi;\Gamma, \Delta \vdash (u/x)t:A$ is the same as $\Xi;\Delta \vdash u:B$ and this is valid by hypothesis.

	\item If $t = v \plus w$, then $\Xi;\Gamma, x^B \vdash v:A$ and $\Xi;\Gamma, x^B \vdash w:A$. By the induction hypothesis, $\Xi;\Gamma, \Delta \vdash (u/x)v:A$ and $\Xi;\Gamma, \Delta \vdash (u/x)w:A$. Therefore, by rule sum, $\Xi;\Gamma, \Delta \vdash (u/x)(v \plus w):A$.

	\item If $t = a \bullet v$, then $\Xi;\Gamma, x^B \vdash v:A$. By the induction hypothesis, $\Xi;\Gamma, \Delta \vdash (u/x)v:A$. Therefore, by rule prod($a$), $\Xi;\Gamma, \Delta \vdash (u/x)(a \bullet v):A$.

	\item If $t = \elimone(v, w)$, then there are $\Gamma_1, \Gamma_2$ such that $\Gamma = \Gamma_1, \Gamma_2$ and there are two cases.
	  \begin{itemize}
	    \item If $\Xi;\Gamma_1, x^B \vdash v:\one$ and $\Xi;\Gamma_2 \vdash w:A$, by the induction hypothesis $\Xi;\Gamma_1, \Delta \vdash (u/x)v:\one$. By rule $\one_e$, $\Xi;\Gamma, \Delta \vdash \elimone((u/x)v, w):A$.
	    \item If $\Xi;\Gamma_1 \vdash v:\one$ and $\Xi;\Gamma_2, x^B \vdash w:A$, by the induction hypothesis $\Xi;\Gamma_2, \Delta \vdash (u/x)w:A$. By rule $\one_e$, $\Xi;\Gamma, \Delta \vdash \elimone(v,(u/x)w):A$.
	  \end{itemize}
	  Therefore, $\Xi;\Gamma, \Delta \vdash (u/x)\elimone(v,w):A$.

	\item If $t = \lambda y^C.v$, then $A = C \multimap D$ and $\Xi;\Gamma, x^B, y^C \vdash v:D$. By the induction hypothesis, $\Xi;\Gamma, y^C, \Delta \vdash (u/x)v:D$. Therefore, by rule $\multimap_i$, $\Xi;\Gamma, \Delta \vdash (u/x)\lambda y^C.v:C \multimap D$.

	\item If $t = v~w$, then there are $\Gamma_1$ and $\Gamma_2$ such that $\Gamma = \Gamma_1, \Gamma_2$ and there are two cases.
	  \begin{itemize}
	    \item If $\Xi;\Gamma_1, x^B \vdash v:C \multimap A$ and $\Xi;\Gamma_2 \vdash w:C$, by the induction hypothesis $\Xi;\Gamma_1, \Delta \vdash (u/x)v:C \multimap A$. By rule $\multimap_e$, $\Xi;\Gamma, \Delta \vdash (u/x)v~w:A$.
	    \item If $\Xi;\Gamma_1 \vdash v:C \multimap A$ and $\Xi;\Gamma_2, x^B \vdash w:C$, by the induction hypothesis $\Xi;\Gamma_2, \Delta \vdash (u/x)w:C$. By rule $\multimap_e$, $\Xi;\Gamma, \Delta \vdash v~(u/x)w:A$.
	  \end{itemize}
	  Therefore, $\Xi;\Gamma, \Delta \vdash (u/x)(v~w):A$.

	\item If $t = v \otimes w$, then $A = C \otimes D$ and there are $\Gamma_1$ and $\Gamma_2$ such that $\Gamma = \Gamma_1, \Gamma_2$ and there are two cases.
	  \begin{itemize}
	    \item If $\Xi;\Gamma_1, x^B \vdash v:C$ and $\Xi;\Gamma_2 \vdash w:D$, by the induction hypothesis $\Xi;\Gamma_1, \Delta \vdash (u/x)v:C$. By rule $\otimes_i$, $\Xi;\Gamma, \Delta \vdash ((u/x)v) \otimes w:C \otimes D$.
	    \item If $\Xi;\Gamma_1 \vdash v:C$ and $\Xi;\Gamma_2, x^B \vdash w:D$, by the induction hypothesis $\Xi;\Gamma_2, \Delta \vdash (u/x)w:D$. By rule $\otimes_i$, $\Xi;\Gamma, \Delta \vdash v \otimes ((u/x)w):C \otimes D$.
	  \end{itemize}
	  Therefore, $\Xi;\Gamma, \Delta \vdash (u/x)(v \otimes w):C \otimes D$.

	\item If $t = \elimtens(v, y^C z^D.w)$, then there are $\Gamma_1$ and $\Gamma_2$ such that $\Gamma = \Gamma_1, \Gamma_2$ and there are two cases.
	  \begin{itemize}
	    \item If $\Xi;\Gamma_1, x^B \vdash v:C \otimes D$ and $\Xi;\Gamma_2, y^C, z^D \vdash w:A$, by the induction hypothesis $\Xi;\Gamma_1, \Delta \vdash (u/x)v:C \otimes D$. By rule $\otimes_e$, $\Xi;\Gamma, \Delta \vdash \elimtens((u/x)v, y^C z^D.w):A$.
	    \item If $\Xi;\Gamma_1 \vdash v:C \otimes D$ and $\Xi;\Gamma_2, x^B, y^C, z^D \vdash w:A$, by the induction hypothesis $\Xi;\Gamma_2, y^C, z^D, \Delta \vdash (u/x)w:A$. By rule $\otimes_e$, $\Xi;\Gamma, \Delta \vdash \elimtens(v, y^C z^D.(u/x)w):A$.
	  \end{itemize}
	  Therefore, $\Xi;\Gamma, \Delta \vdash (u/x)\elimtens(v, y^C z^D.w):A$.

	\item If $t = \langle\rangle$, then $A = \top$. By rule $\top_i$ $\Xi;\Gamma, \Delta \vdash \langle\rangle:\top$.

	\item If $t = \elimzero(v)$, then there are $\Gamma_1$ and $\Gamma_2$ such that $\Gamma = \Gamma_1, \Gamma_2$ and there are two cases.
	  \begin{itemize}
	    \item If $\Xi;\Gamma_1, x^B \vdash v:\zero$, by the induction hypothesis $\Xi;\Gamma_1, \Delta \vdash (u/x)v:\zero$. By rule $\zero_e$, $\Xi;\Gamma, \Delta \vdash \elimzero((u/x)v):A$.
	    \item If $\Xi;\Gamma_1 \vdash v:\zero$ and $x \notin \fv(v)$, by rule $\zero_e$ $\Xi;\Gamma, \Delta \vdash \elimzero(v):A$.
	  \end{itemize}
	  Therefore, $\Xi;\Gamma, \Delta \vdash (u/x)\elimzero(v):A$.

	\item If $t = \langle v, w \rangle$, then $A = C \with D$, $\Xi;\Gamma, x^B \vdash v:C$ and $\Xi;\Gamma, x^B \vdash w:D$. By the induction hypothesis, $\Xi;\Gamma, \Delta \vdash (u/x)v:C$ and $\Xi;\Gamma, \Delta \vdash (u/x)w:D$. By rule $\with_i$, $\Xi;\Gamma, \Delta \vdash (u/x)\langle v,w \rangle:C \with D$.

	\item If $t = \elimwith^1(v, y^C.w)$, then there are $\Gamma_1$ and $\Gamma_2$ such that $\Gamma = \Gamma_1, \Gamma_2$ and there are two cases.
	  \begin{itemize}
	    \item If $\Xi;\Gamma_1, x^B \vdash v:C \with D$ and $\Xi;\Gamma_2, y^C \vdash w:A$, by the induction hypothesis $\Xi;\Gamma_1, \Delta \vdash (u/x)v:C \with D$. By rule $\with_{e1}$, $\Xi;\Gamma, \Delta \vdash \elimwith^1((u/x)v, y^C.w):A$.
	    \item If $\Xi;\Gamma_1 \vdash v:C \with D$ and $\Xi;\Gamma_2, x^B, y^C \vdash w:A$, by the induction hypothesis $\Xi;\Gamma_2, y^C, \Delta \vdash (u/x)w:A$. By rule $\with_{e1}$, $\Xi;\Gamma, \Delta \vdash \elimwith^1(v, y^C.(u/x)w):A$.
	  \end{itemize}
	  Therefore, $\Xi;\Gamma, \Delta \vdash (u/x)\elimwith^1(v, y^C.w):A$.

	\item If $t = \elimwith^2(v, y^D.w)$, then there are $\Gamma_1$ and $\Gamma_2$ such that $\Gamma = \Gamma_1, \Gamma_2$ and there are two cases.
	  \begin{itemize}
	    \item If $\Xi;\Gamma_1, x^B \vdash v:C \with D$ and $\Xi;\Gamma_2, y^D \vdash w:A$, by the induction hypothesis $\Xi;\Gamma_1, \Delta \vdash (u/x)v:C \with D$. By rule $\with_{e2}$, $\Xi;\Gamma, \Delta \vdash \elimwith^2((u/x)v, y^D.w):A$.
	    \item If $\Xi;\Gamma_1 \vdash v:C \with D$ and $\Xi;\Gamma_2, x^B, y^D \vdash w:A$, by the induction hypothesis $\Xi;\Gamma_2, y^D, \Delta \vdash (u/x)w:A$. By rule $\with_{e2}$, $\Xi;\Gamma, \Delta \vdash \elimwith^2(v, y^D.(u/x)w):A$.
	  \end{itemize}
	  Therefore, $\Xi;\Gamma, \Delta \vdash (u/x)\elimwith^2(v, y^D.w):A$.

	\item If $t = \inl(v)$, then $A = C \oplus D$ and $\Xi;\Gamma, x^B \vdash v:C$. By the induction hypothesis, $\Xi;\Gamma, \Delta \vdash (u/x)v:C$. Therefore, by rule $\oplus_{i1}$, $\Xi;\Gamma, \Delta \vdash (u/x)\inl(v):C \oplus D$.

	\item If $t = \inr(v)$, then $A = C \oplus D$ and $\Xi;\Gamma, x^B \vdash v:D$. By the induction hypothesis, $\Xi;\Gamma, \Delta \vdash (u/x)v:D$. Therefore, by rule $\oplus_{i2}$, $\Xi;\Gamma, \Delta \vdash (u/x)\inr(v):C \oplus D$.

	\item If $t = \elimplus(v, y^C.w, z^D.s)$, then there are $\Gamma_1$ and $\Gamma_2$ such that $\Gamma = \Gamma_1, \Gamma_2$ and there are two cases.
	  \begin{itemize}
	    \item If $\Xi;\Gamma_1, x^B \vdash v:C \oplus D$, $\Xi;\Gamma_2, y^C \vdash w:A$ and $\Xi;\Gamma_2, z^D \vdash s:A$, by the induction hypothesis $\Xi;\Gamma_1, \Delta \vdash (u/x)v:C \oplus D$. By rule $\oplus_e$, $\Xi;\Gamma, \Delta \vdash \elimplus((u/x)v, y^C.w, z^D.s):A$.
	    \item If $\Xi;\Gamma_1 \vdash v:C \oplus D$, $\Xi;\Gamma_2, x^B, y^C \vdash w:A$ and $\Xi;\Gamma_2, x^B, z^D \vdash s:A$, by the induction hypothesis $\Xi;\Gamma_2, y:C \vdash (u/x)w:A$ and $\Xi;\Gamma_2, z:D \vdash (u/x)s:A$. By rule $\oplus_e$,\\ $\Xi;\Gamma, \Delta \vdash \elimplus(v, y^C.(u/x)w, z^D.(u/x)s):A$.
	  \end{itemize}
	  Therefore, $\Xi;\Gamma, \Delta \vdash (u/x)\elimplus(v, y^C.w, z^D.s):A$.

	\item If $t = \oc v$, this is not possible since the linear context should be empty.

	\item If $t = \elimbang(v, y^C.w)$, then there are $\Gamma_1$ and $\Gamma_2$ such that $\Gamma = \Gamma_1, \Gamma_2$ and there are two cases.
	  \begin{itemize}
	    \item If $\Xi; \Gamma_1, x^B \vdash v:\oc C$ and $\Xi, y^C; \Gamma_2 \vdash w:A$, by the induction hypothesis $\Xi; \Gamma_1, \Delta \vdash (u/x) v:\oc C$. By rule $\oc_e$, $\Xi; \Gamma, \Delta \vdash \elimbang((u/x)v, y^C.w):A$.
	    \item If $\Xi; \Gamma_1 \vdash v:\oc C$ and $\Xi, y^C; \Gamma_2, x^B \vdash w:A$, by the induction hypothesis $\Xi, y^C; \Gamma_2, \Delta \vdash (u/x)w:A$. By rule $\oc_e$,\\ $\Xi; \Gamma, \Delta \vdash \elimbang(v, y^C.(u/x)w):A$.
	  \end{itemize}
	  Therefore, $\Xi; \Gamma, \Delta \vdash (u/x)\elimbang(v, y^C.w):A$.

	\item If $t = \Lambda X.v$, then $A = \forall X.C$, $\Xi;\Gamma, x^B \vdash v:C$ and $X \notin \FV(\Xi, \Gamma, B)$. By the induction hypothesis, $\Xi;\Gamma \vdash (u/x)v:C$. Since $X \notin \FV(\Xi, \Gamma)$, by rule $\forall_i$ $\Xi;\Gamma \vdash (u/x) \Lambda X.v:\forall X.C$.

	\item If $t = v~C$, then $A = (C/X)D$ and $\Xi;\Gamma, x^B \vdash v:\forall X.D$. By the induction hypothesis, $\Xi;\Gamma, \Delta \vdash (u/x)v:\forall X.D$. Therefore, by rule $\forall_e$, $\Xi;\Gamma, \Delta \vdash (u/x)v:(C/X)D$.
      \end{itemize}
    \item By induction on $t$.
      \begin{itemize}
	\item If $t = x$, then $\Gamma$ is empty and $A = B$. Thus, $\Xi;\Gamma \vdash (u/x)t:A$ is the same as $\Xi;\varnothing \vdash u:B$ and this is valid by hypothesis.

	\item If $t = y \neq x$, then either $\Gamma = \{y^A\}$ or $\Gamma$ is empty and $y^A \in \Xi$.
	  \begin{itemize}
	    \item In the first case, $\Xi; y^A \vdash y:A$ by rule lin-ax.
	    \item In the second case, $\Xi; \varnothing \vdash y:A$ by rule ax.
	  \end{itemize}
	  Therefore, $\Xi; \Gamma \vdash y:A$.

	\item If $t = v \plus w$, then $\Xi, x^B;\Gamma \vdash v:A$ and $\Xi, x^B;\Gamma \vdash w:A$. By the induction hypothesis, $\Xi;\Gamma \vdash (u/x)v:A$ and $\Xi;\Gamma \vdash (u/x)w:A$. Therefore, by rule sum, $\Xi;\Gamma \vdash (u/x)(v \plus w):A$.

	\item If $t = a \bullet v$, then $\Xi, x^B;\Gamma \vdash v:A$. By the induction hypothesis, $\Xi;\Gamma \vdash (u/x)v:A$. Therefore, by rule prod($a$), $\Xi;\Gamma \vdash (u/x)(a \bullet v):A$.

	\item If $t = \elimone(v, w)$, then there are $\Gamma_1, \Gamma_2$ such that $\Gamma = \Gamma_1, \Gamma_2$, $\Xi, x^B; \Gamma_1 \vdash v:\one$ and $\Xi, x^B; \Gamma_2 \vdash w:A$. By the induction hypothesis, $\Xi; \Gamma_1 \vdash (u/x)v:\one$ and $\Xi; \Gamma_2 \vdash (u/x)w:A$. By rule $\one_e$, $\Xi; \Gamma \vdash (u/x)\elimone(v, w):A$.

	\item If $t = \lambda y^C.v$, then $A = C \multimap D$ and $\Xi, x^B;\Gamma, y^C \vdash v:D$. By the induction hypothesis, $\Xi;\Gamma, y^C \vdash (u/x)v:D$. Therefore, by rule $\multimap_i$, $\Xi;\Gamma \vdash (u/x)\lambda y^C.v:C \multimap D$.

	\item If $t = v~w$, then there are $\Gamma_1$ and $\Gamma_2$ such that $\Gamma = \Gamma_1, \Gamma_2$, $\Xi, x^B; \Gamma_1 \vdash v:C \multimap A$ and $\Xi, x^B; \Gamma_2 \vdash w:C$. By the induction hypothesis, $\Xi; \Gamma_1 \vdash (u/x)v:C \multimap A$ and $\Xi; \Gamma_2 \vdash (u/x)w:C$. By rule $\multimap_e$, $\Xi;\Gamma \vdash (u/x)(v~w):A$.

	\item If $t = v \otimes w$, then $A = C \otimes D$ and there are $\Gamma_1$ and $\Gamma_2$ such that $\Gamma = \Gamma_1, \Gamma_2$, $\Xi, x^B; \Gamma_1 \vdash v:C$ and $\Xi, x^B; \Gamma_2 \vdash w:D$. By the induction hypothesis, $\Xi; \Gamma_1 \vdash (u/x)v:C$ and $\Xi; \Gamma_2 \vdash (u/x)w:D$. By rule $\otimes_i$, $\Xi;\Gamma \vdash (u/x)(v \otimes w):C \otimes D$.

	\item If $t = \elimtens(v, y^C z^D.w)$, then there are $\Gamma_1$ and $\Gamma_2$ such that $\Gamma = \Gamma_1, \Gamma_2$, $\Xi, x^B; \Gamma_1 \vdash v:C \otimes D$ and $\Xi, x^B; \Gamma_2, y^C, z^D \vdash w:A$. By the induction hypothesis, $\Xi; \Gamma_1 \vdash (u/x)v:C \otimes D$ and $\Xi; \Gamma_2, y^C, z^D \vdash (u/x)w:A$. By rule $\otimes_e$, $\Xi;\Gamma \vdash (u/x)\elimtens(v, y^C z^D.w):A$.

	\item If $t = \langle\rangle$, then $A = \top$. By rule $\top_i$ $\Xi;\Gamma \vdash \langle\rangle:\top$.

	\item If $t = \elimzero(v)$, then there are $\Gamma_1$ and $\Gamma_2$ such that $\Gamma = \Gamma_1, \Gamma_2$ and $\Xi, x^B; \Gamma_1 \vdash v:\zero$. By the induction hypothesis $\Xi; \Gamma_1 \vdash (u/x)v:\zero$. By rule $\zero_e$, $\Xi;\Gamma \vdash (u/x)\elimzero(v):A$.

	\item If $t = \langle v, w \rangle$, then $A = C \with D$, $\Xi, x^B;\Gamma \vdash v:C$ and $\Xi, x^B;\Gamma \vdash w:D$. By the induction hypothesis, $\Xi;\Gamma, \vdash (u/x)v:C$ and $\Xi;\Gamma, \vdash (u/x)w:D$. By rule $\with_i$, $\Xi;\Gamma \vdash (u/x)\langle v,w \rangle:C \with D$.

	\item If $t = \elimwith^1(v, y^C.w)$, then there are $\Gamma_1$ and $\Gamma_2$ such that $\Gamma = \Gamma_1, \Gamma_2$, $\Xi, x^B; \Gamma_1 \vdash v:C \with D$ and $\Xi, x^B; \Gamma_2, y^C \vdash w:A$. By the induction hypothesis, $\Xi; \Gamma_1 \vdash (u/x)v:C \with D$ and $\Xi; \Gamma_2, y^C \vdash (u/x)w:A$. By rule $\with_{e1}$, $\Xi;\Gamma \vdash (u/x)\elimwith^1(v, y^C.w):A$.

	\item If $t = \elimwith^2(v, y^D.w)$, then there are $\Gamma_1$ and $\Gamma_2$ such that $\Gamma = \Gamma_1, \Gamma_2$, $\Xi, x^B; \Gamma_1 \vdash v:C \with D$ and $\Xi, x^B; \Gamma_2, y^C \vdash w:A$. By the induction hypothesis, $\Xi; \Gamma_1 \vdash (u/x)v:C \with D$ and $\Xi; \Gamma_2, y^C \vdash (u/x)w:A$. By rule $\with_{e2}$, $\Xi;\Gamma \vdash (u/x)\elimwith^2(v, y^D.w):A$.

	\item If $t = \inl(v)$, then $A = C \oplus D$ and $\Xi, x^B;\Gamma \vdash v:C$. By the induction hypothesis, $\Xi;\Gamma \vdash (u/x)v:C$. Therefore, by rule $\oplus_{i1}$, $\Xi;\Gamma \vdash (u/x)\inl(v):C \oplus D$.

	\item If $t = \inr(v)$, then $A = C \oplus D$ and $\Xi, x^B;\Gamma \vdash v:D$. By the induction hypothesis, $\Xi;\Gamma \vdash (u/x)v:D$. Therefore, by rule $\oplus_{i2}$, $\Xi;\Gamma \vdash (u/x)\inr(v):C \oplus D$.

	\item If $t = \elimplus(v, y^C.w, z^D.s)$, then there are $\Gamma_1$ and $\Gamma_2$ such that $\Gamma = \Gamma_1, \Gamma_2$, $\Xi, x^B; \Gamma_1 \vdash v:C \oplus D$, $\Xi, x^B; \Gamma_2, y^C \vdash w:A$ and $\Xi, x^B; \Gamma_2, z^D \vdash s:A$. By the induction hypothesis, $\Xi; \Gamma_1 \vdash (u/x)v:C \oplus D$, $\Xi; \Gamma_2, y^C \vdash (u/x)w:A$ and $\Xi; \Gamma_2, z^D \vdash (u/x)s:A$. By rule $\oplus_e$,\\ $\Xi;\Gamma \vdash (u/x)\elimplus(v, y^C.w, z^D.s):A$.

	\item If $t = \oc v$, then $A = \oc C$, $\Gamma$ is empty and $\Xi, x^B; \varnothing \vdash v:C$. By the induction hypothesis $\Xi; \varnothing \vdash (u/x)v:C$, and by rule $\oc_i$ $\Xi; \varnothing \vdash (u/x)\oc v: \oc C$.

	\item If $t = \elimbang(v, y^C.w)$, then there are $\Gamma_1$ and $\Gamma_2$ such that $\Gamma = \Gamma_1, \Gamma_2$, $\Xi, x^B; \Gamma_1 \vdash v:\oc C$ and $\Xi, x^B; \Gamma_2, y^C \vdash w:A$. By the induction hypothesis, $\Xi; \Gamma_1 \vdash (u/x)v:\oc C$ and $\Xi; \Gamma_2, y^C \vdash (u/x)w:A$. By rule $\oc_e$, $\Xi; \Gamma \vdash (u/x)\elimbang(v, y^C.w):A$.

	\item If $t = \Lambda X.v$, then $A = \forall X.C$, $\Xi, x^B;\Gamma \vdash v:C$ and $X \notin \FV(\Xi, \Gamma, B)$. By the induction hypothesis, $\Xi;\Gamma \vdash (u/x)v:C$. Since $X \notin \FV(\Xi, \Gamma)$, by rule $\forall_i$ $\Xi;\Gamma \vdash (u/x) \Lambda X.v:\forall X.C$.

	\item If $t = v~C$, then $A = (C/X)D$ and $\Xi, x^B;\Gamma \vdash v:\forall X.D$. By the induction hypothesis, $\Xi;\Gamma \vdash (u/x)v:\forall X.D$. Therefore, by rule $\forall_e$, $\Xi;\Gamma \vdash (u/x)v:(C/X)D$. 
  \qed
      \end{itemize}
  \end{enumerate}
\end{proof}

\SR*
\begin{proof}
  By induction on the relation $\lra$. 
  \begin{itemize}
    \item If $t = \elimone(a.\star, v)$ and $u = a \bullet v$, then $\Xi;\Gamma \vdash v:A$. Therefore, by rule prod($a$), $\Xi;\Gamma \vdash a \bullet v:A$.

    \item If $t = (\lambda x^B.v_1)~v_2$ and $u = (v_2/x)v_1$, then there are $\Gamma_1$ and $\Gamma_2$ such that $\Gamma = \Gamma_1, \Gamma_2$, $\Xi;\Gamma_1, x^B \vdash v_1:A$ and $\Xi;\Gamma_2 \vdash v_2:B$. By Lemma~\ref{lem:polysubstitution}, $\Xi;\Gamma \vdash (v_2/x)v_1:A$.

    \item If $t = \elimtens(v_1 \otimes v_2, x^B y^C.v_3)$ and $u = (v_1/x, v_2/y)v_3$, then there are $\Gamma_1$, $\Gamma_2$ and $\Gamma_3$ such that $\Gamma = \Gamma_1, \Gamma_2, \Gamma_3$, $\Xi;\Gamma_1 \vdash v_1:B$, $\Xi;\Gamma_2 \vdash v_2:C$ and $\Xi;\Gamma_3, x^B, y^C \vdash v_3:A$. By Lemma~\ref{lem:polysubstitution} twice, $\Xi;\Gamma \vdash (v_1/x, v_2/y)v_3:A$.

    \item If $t = \elimwith^1(\langle v_1, v_2 \rangle, x^B.v_3)$ and $u = (v_1/x)v_3$, then there are $\Gamma_1$ and $\Gamma_2$ such that $\Gamma = \Gamma_1, \Gamma_2$, $\Xi;\Gamma_1 \vdash v_1:B$, $\Xi;\Gamma_1 \vdash v_2:C$ and $\Xi;\Gamma_2, x^B \vdash v_3:A$. By Lemma~\ref{lem:polysubstitution}, $\Xi;\Gamma \vdash (v_1/x)v_3:A$.

    \item If $t = \elimwith^2(\langle v_1, v_2 \rangle, x^B.v_3)$ and $u = (v_2/x)v_3$, then there are $\Gamma_1$ and $\Gamma_2$ such that $\Gamma = \Gamma_1, \Gamma_2$, $\Xi;\Gamma_1 \vdash v_1:C$, $\Xi;\Gamma_1 \vdash v_2:B$ and $\Xi;\Gamma_2, x^B \vdash v_3:A$. By Lemma~\ref{lem:polysubstitution}, $\Xi;\Gamma \vdash (v_2/x)v_3:A$.

    \item If $t = \elimplus(\inl(v_1), x^B.v_2, y^C.v_3)$ and $u = (v_1/x)v_2$, then there are $\Gamma_1$ and $\Gamma_2$ such that $\Gamma = \Gamma_1, \Gamma_2$, $\Xi;\Gamma_1 \vdash v_1:B$, $\Xi;\Gamma_2, x^B \vdash v_2:A$ and $\Xi;\Gamma_2, y^C \vdash v_3:A$. By Lemma~\ref{lem:polysubstitution}, $\Xi;\Gamma \vdash (v_1/x)v_2:A$.

    \item If $t = \elimplus(\inr(v_1), x^B.v_2, y^C.v_3)$ and $u = (v_1/y)v_3$, then there are $\Gamma_1$ and $\Gamma_2$ such that $\Gamma = \Gamma_1, \Gamma_2$, $\Xi;\Gamma_1 \vdash v_1:C$, $\Xi;\Gamma_2, x^B \vdash v_2:A$ and $\Xi;\Gamma_2, y^C \vdash v_3:A$. By Lemma~\ref{lem:polysubstitution}, $\Xi;\Gamma \vdash (v_1/y)v_3:A$.

    \item If $t = (\Lambda X.v)~B$ and $u = (B/X)v$, then $A = (B/X)C$, $\Xi;\Gamma \vdash v:C$ and $X \notin \FV(\Xi, \Gamma)$. By Lemma~\ref{lem:substitutiontypesonterms}, $(B/X)\Xi;(B/X)\Gamma \vdash (B/X)v:(B/X)C$. Since $X \notin \FV(\Xi, \Gamma)$, then $\Xi;\Gamma \vdash (B/X)v:(B/X)C$.

    \item If $t = \elimbang(\oc v_1, x^B.v_2)$ and $u = (v_1/x)v_2$, then $\Xi; \varnothing \vdash v_1:B$ and $\Xi, x^B; \Gamma \vdash v_2:A$. By Lemma~\ref{lem:polysubstitution}, $\Xi; \Gamma \vdash (v_1/x)v_2:A$.

    \item If $t = {a.\star} \plus b.\star$ and $u = (a+b).\star$, then $A = \one$ and $\Gamma$ is empty. By rule $\one_i$($a+b$), $\Xi; \varnothing \vdash (a+b).\star:\one$.

    \item If $t = (\lambda x^B.v_1) \plus (\lambda x^B.v_2)$ and $u = \lambda x^B.(v_1 \plus v_2)$, then $A = B \multimap C$, $\Xi; \Gamma, x^B \vdash v_1:C$ and $\Xi; \Gamma, x^B \vdash v_2:C$. By rule sum, $\Xi; \Gamma, x^B \vdash v_1 \plus v_2:C$. Therefore, by rule $\multimap_i$, $\Xi; \Gamma \vdash \lambda x^B.(v_1 \plus v_2): B \multimap C$.

    \item If $t = \elimtens(v_1 \plus v_2, x^B y^C.v_3)$ and $u = \elimtens(v_1, x^B y^C.v_3) \plus \elimtens(v_2, x^B y^C.v_3)$, then there are $\Gamma_1$ and $\Gamma_2$ such that $\Gamma = \Gamma_1, \Gamma_2$, $\Xi; \Gamma_1 \vdash v_1:B \otimes C$, $\Xi; \Gamma_1 \vdash v_2:B \otimes C$ and $\Xi; \Gamma_2, x^B, y^C \vdash v_3:A$. By rule $\otimes_e$, $\Xi; \Gamma \vdash \elimtens(v_1, x^B y^C.v_3):A$ and $\Xi; \Gamma \vdash \elimtens(v_2, x^B y^C.v_3):A$. Therefore, by rule sum, $\Xi; \Gamma \vdash \elimtens(v_1, x^B y^C.v_3) \plus \elimtens(v_2, x^B y^C.v_3):A$.

    \item If $t = \langle\rangle \plus \langle\rangle$ and $u = \langle\rangle$, then $A = \top$. By rule $\top_i$, $\Xi;\Gamma \vdash \langle\rangle:\top$.

    \item If $t = \langle v_1, v_2 \rangle \plus \langle v_3, v_4 \rangle$ and $u = \langle v_1 \plus v_3, v_2 \plus v_4 \rangle$, then $A = B \with C$, $\Xi; \Gamma \vdash v_1:B$, $\Xi; \Gamma \vdash v_2:C$, $\Xi; \Gamma \vdash v_3:B$ and $\Xi; \Gamma \vdash v_4:C$. By rule sum, $\Xi; \Gamma \vdash v_1 \plus v_3:B$ and $\Xi; \Gamma \vdash v_2 \plus v_4:C$. Therefore, by rule $\with_i$, $\Xi; \Gamma \vdash \langle v_1 \plus v_3, v_2 \plus v_4 \rangle:B \with C$.

    \item If $t = \elimplus(v_1 \plus v_2, x^B.v_3, y^C.v_4)$ and\\ $u = \elimplus(v_1, x^B.v_3, y^C.v_4) \plus \elimplus(v_2, x^B.v_3, y^C.v_4)$, then there are $\Gamma_1$ and $\Gamma_2$ such that $\Gamma = \Gamma_1, \Gamma_2$, $\Xi; \Gamma_1 \vdash v_1:B \oplus C$, $\Xi; \Gamma_1 \vdash v_2:B \oplus C$, $\Xi; \Gamma_2, x^B \vdash v_3:A$ and $\Xi; \Gamma_2, y^C \vdash v_4:A$. By rule $\oplus_e$, $\Xi; \Gamma \vdash \elimplus(v_1, x^B.v_3, y^C.v_4):A$ and $\Xi; \Gamma \vdash \elimplus(v_2, x^B.v_3, y^C.v_4):A$. Therefore, by rule sum, $\Xi; \Gamma \vdash \elimplus(v_1, x^B.v_3, y^C.v_4) \plus \elimplus(v_2, x^B.v_3, y^C.v_4):A$.

    \item If $t = (\Lambda X.v_1) \plus (\Lambda X.v_2)$ and $u = \Lambda X.(v_1 \plus v_2)$, then $A = \forall X.B$, $\Xi; \Gamma \vdash v_1:B$, $\Xi; \Gamma \vdash v_2:B$ and $X \notin \FV(\Xi, \Gamma)$. By rule sum, $\Xi; \Gamma \vdash v_1 \plus v_2:B$. Therefore, by rule $\forall_i$, $\Xi; \Gamma \vdash \Lambda X.(v_1 \plus v_2):\forall X.B$.

    \item If $t = \oc v_1 \plus \oc v_2$ and $u = \oc (v_1 \plus v_2)$, then $A = \oc B$, $\Xi; \varnothing \vdash v_1:B$ and $\Xi; \varnothing \vdash v_2:B$. By rule sum, $\Xi; \varnothing \vdash v_1 \plus v_2:B$. By rule $\oc_i$, $\Xi; \varnothing \vdash \oc(v_1 \plus v_2):\oc B$.

    \item If $t = a \bullet b.\star$ and $u = (a \times b).\star$, then $A = \one$ and $\Gamma$ is empty. By rule $\one_i$($a \times b$), $\Xi; \varnothing \vdash (a \times b).\star:\one$.

    \item If $t = a \bullet \lambda x^B.v$ and $u = \lambda x^B.a \bullet v$, then $A = B \multimap C$ and $\Xi; \Gamma, x^B \vdash v:C$. By rule prod($a$), $\Xi; \Gamma, x^B \vdash a \bullet v:C$. Therefore, by rule $\multimap_i$, $\Xi; \Gamma \vdash \lambda x^B.a \bullet v:B \multimap C$.

    \item If $t = \elimtens(a \bullet v_1, x^B y^C.v_2)$ and $u = a \bullet \elimtens(v_1, x^B y^C.v_2)$, then there are $\Gamma_1$ and $\Gamma_2$ such that $\Gamma = \Gamma_1, \Gamma_2$, $\Xi; \Gamma_1 \vdash v_1:B \otimes C$ and $\Xi; \Gamma_2, x^B, y^C \vdash v_2:A$. By rule $\otimes_e$, $\Xi; \Gamma \vdash \elimtens(v_1, x^B y^C.v_2):A$. Therefore, by rule prod($a$), $\Xi; \Gamma \vdash a \bullet \elimtens(v_1, x^B y^C.v_2):A$.

    \item If $t = a \bullet \langle\rangle$ and $u = \langle\rangle$, then $A = \top$. Therefore, by rule $\top_i$, $\Xi; \Gamma \vdash \langle\rangle:\top$.

    \item If $t = a \bullet \langle v_1, v_2 \rangle$ and $u = \langle a \bullet v_1, a \bullet v_2 \rangle$, then $A = B \with C$, $\Xi; \Gamma \vdash v_1:B$ and $\Xi; \Gamma \vdash v_2:C$. By rule prod($a$), $\Xi; \Gamma \vdash a \bullet v_1:B$ and $\Xi; \Gamma \vdash a \bullet v_2:C$. Therefore, by rule $\with_i$, $\Xi; \Gamma \vdash \langle a \bullet v_1, a \bullet v_2 \rangle:B \with C$.

    \item If $t = \elimplus(a \bullet v_1, x^B.v_2, y^C.v_3)$ and $u = a \bullet \elimplus(v_1, x^B.v_2, y^C.v_3)$, then there are $\Gamma_1$ and $\Gamma_2$ such that $\Gamma = \Gamma_1, \Gamma_2$, $\Xi; \Gamma_1 \vdash v_1:B \oplus C$, $\Xi; \Gamma_2, x^B \vdash v_2:A$ and $\Xi; \Gamma_2, y^C \vdash v_3:A$. By rule $\oplus_e$, $\Xi; \Gamma \vdash \elimplus(v_1, x^B.v_2, y^C.v_3):A$. Therefore, by rule prod($a$), $\Xi; \Gamma \vdash a \bullet \elimplus(v_1, x^B.v_2, y^C.v_3):A$.

    \item If $t = a \bullet \Lambda X.v$ and $u = \Lambda X.a \bullet v$, then $A = \forall X.B$, $\Xi; \Gamma \vdash v:B$ and $X \notin \FV(\Xi, \Gamma)$. By rule prod($a$), $\Xi; \Gamma \vdash a \bullet v:B$. Therefore, by rule $\forall_i$, $\Xi; \Gamma \vdash \Lambda X.v:\forall X.B$.

    \item If $t = a \bullet \oc v$ and $u = \oc(a \bullet v)$, then $A = \oc B$ and $\Xi; \varnothing \vdash v:B$. By rule prod($a$), $\Xi; \varnothing \vdash a \bullet v: B$. By rule $\oc_i$, $\Xi; \varnothing \vdash \oc(a \bullet v): \oc B$.\qed
  \end{itemize}
\end{proof}

\section{Proof of Section~\ref{sec:ST}}\label{proof:ST}
\typeinterpretationsarerc*
\begin{proof}
  By induction on $A$.
  \begin{itemize}
    \item If $A = X$, $\rho(X) \in \mathcal R$ since $\rho$ is a valuation.
    \item If $A = \one$, $\SN$ has the properties CR1, CR2, CR3 and CR4.
    \item If $A = B \multimap C$:
      \begin{itemize}
        \item Let $t \in \llbracket B \rrbracket_\rho \hatmultimap \llbracket C \rrbracket_\rho$, then $t \in \SN$.
        \item Let $t \in \llbracket B \rrbracket_\rho \hatmultimap \llbracket C \rrbracket_\rho$ such that $t \lra t'$. Then $t' \in \SN$, and if $t' \lras \lambda x^D.u$, $t \lras \lambda x^D.u$. Therefore, for all $v \in \llbracket B \rrbracket_\rho$, $(v/x)u \in \llbracket C \rrbracket_\rho$.
        \item Let $t$ be a proof-term that is not an introduction such that $Red(t) \subseteq \llbracket B \rrbracket_\rho \hatmultimap \llbracket C \rrbracket_\rho$. Since $Red(t) \subseteq \SN$, $t \in \SN$. If $t \lras \lambda x^D.u$, the rewrite sequence has at least one step because $t$ is not an introduction. Then, there is a proof-term $t' \in Red(t)$ such that $t' \lras \lambda x^D.u$. Therefore, for all $v \in \llbracket B \rrbracket_\rho$, $(v/x)u \in \llbracket C \rrbracket_\rho$.
        \item Let $t \in \llbracket B \rrbracket_\rho \hatmultimap \llbracket C \rrbracket_\rho$ such that $(D/X)t \lras \lambda x^E.u'$. Then, there is a proof-term $u$ such that $t \lras \lambda x^E.u$ and $u' = (D/X)u$. Let $v \in \llbracket B \rrbracket_\rho$, we have that $(v/x)(D/X)u = (X/Y)(D/X)((Y/X)v/x)u$, where $Y$ is a fresh variable. By the induction hypothesis, $(Y/X)v \in \llbracket B \rrbracket_\rho$. Then, $((Y/X)v/x)u \in \llbracket C \rrbracket_\rho$. Therefore, by the induction hypothesis twice, $(X/Y)(D/X)((Y/X)v/x)u \in \llbracket C \rrbracket_\rho$.
      \end{itemize}
    \item If $A = B \otimes C$:
      \begin{itemize}
        \item Let $t \in \llbracket B \rrbracket_\rho \hatotimes \llbracket C \rrbracket_\rho$, then $t \in \SN$.
        \item Let $t \in \llbracket B \rrbracket_\rho \hatotimes \llbracket C \rrbracket_\rho$ such that $t \lra t'$. Then $t' \in \SN$, and if $t' \lras u \otimes v$, $t \lras u \otimes v$. Therefore, $u \in \llbracket B \rrbracket_\rho$ and $v \in \llbracket C \rrbracket_\rho$.
        \item Let $t$ be a proof-term that is not an introduction such that $Red(t) \subseteq \llbracket B \rrbracket_\rho \hatotimes \llbracket C \rrbracket_\rho$. Since $Red(t) \subseteq \SN$, $t \in \SN$. If $t \lras u \otimes v$, the rewrite sequence has at least one step because $t$ is not an introduction. Then, there is a proof-term $t' \in Red(t)$ such that $t' \lras u \otimes v$. Therefore, $u \in \llbracket B \rrbracket_\rho$ and $v \in \llbracket C \rrbracket_\rho$.
        \item Let $t \in \llbracket B \rrbracket_\rho \hatotimes \llbracket C \rrbracket_\rho$ such that $(D/X)t \lras u' \otimes v'$. Then, there are proof-terms $u$ and $v$ such that $t \lras u \otimes v$ with $u' = (D/X)u$ and $v' = (D/X)v$. Then, $u \in \llbracket B \rrbracket_\rho$ and $v \in \llbracket C \rrbracket_\rho$. By the induction hypothesis, $(D/X)u \in \llbracket B \rrbracket_\rho$ and $(D/X)v \in \llbracket C \rrbracket_\rho$.
      \end{itemize}
    \item If $A = \top$, $\SN$ has the properties CR1, CR2, CR3 and CR4.
    \item If $A = \zero$, $\SN$ has the properties CR1, CR2, CR3 and CR4.
    \item If $A = B \with C$:
    \begin{itemize}
      \item Let $t \in \llbracket B \rrbracket_\rho \hatand \llbracket C \rrbracket_\rho$, then $t \in \SN$.
      \item Let $t \in \llbracket B \rrbracket_\rho \hatand \llbracket C \rrbracket_\rho$ such that $t \lra t'$. Then $t' \in \SN$, and if $t' \lras \langle u,v \rangle$, $t \lras \langle u,v \rangle$. Therefore, $u \in \llbracket B \rrbracket_\rho$ and $v \in \llbracket C \rrbracket_\rho$.
      \item Let $t$ be a proof-term that is not an introduction such that $Red(t) \subseteq \llbracket B \rrbracket_\rho \hatand \llbracket C \rrbracket_\rho$. Since $Red(t) \subseteq \SN$, $t \in \SN$. If $t \lras \langle u,v \rangle$, the rewrite sequence has at least one step because $t$ is not an introduction. Then, there is a proof-term $t' \in Red(t)$ such that $t' \lras \langle u,v \rangle$. Therefore, $u \in \llbracket B \rrbracket_\rho$ and $v \in \llbracket C \rrbracket_\rho$.
      \item Let $t \in \llbracket B \rrbracket_\rho \hatand \llbracket C \rrbracket_\rho$ such that $(D/X)t \lras \langle u', v' \rangle$. Then, there are proof-terms $u$ and $v$ such that $t \lras \langle u,v \rangle$ with $u' = (D/X)u$ and $v' = (D/X)v$. Then, $u \in \llbracket B \rrbracket_\rho$ and $v \in \llbracket C \rrbracket_\rho$. By the induction hypothesis, $(D/X)u \in \llbracket B \rrbracket_\rho$ and $(D/X)v \in \llbracket C \rrbracket_\rho$.
    \end{itemize}
    \item If $A = B \oplus C$:
      \begin{itemize}
        \item Let $t \in \llbracket B \rrbracket_\rho \hatoplus \llbracket C \rrbracket_\rho$, then $t \in \SN$.
        \item Let $t \in \llbracket B \rrbracket_\rho \hatoplus \llbracket C \rrbracket_\rho$ such that $t \lra t'$. Then $t' \in \SN$, and if $t' \lras \inl(u)$, $t \lras \inl(u)$. Therefore, $u \in \llbracket B \rrbracket_\rho$. If $t' \lras \inr(v)$, $t \lras \inr(v)$. Therefore, $v \in \llbracket C \rrbracket_\rho$.
        \item Let $t$ be a proof-term that is not an introduction such that $Red(t) \subseteq \llbracket B \rrbracket_\rho \hatoplus \llbracket C \rrbracket_\rho$. Since $Red(t) \subseteq \SN$, $t \in \SN$. If $t \lras \inl(u)$, the rewrite sequence has at least one step because $t$ is not an introduction. Then, there is a proof-term $t' \in Red(t)$ such that $t' \lras \inl(u)$. Therefore, $u \in \llbracket B \rrbracket_\rho$. If $t \lras \inr(v)$, the rewrite sequence has at least one step because $t$ is not an introduction. Then, there is a proof-term $t' \in Red(t)$ such that $t' \lras \inr(v)$. Therefore, $v \in \llbracket C \rrbracket_\rho$.
        \item Let $t \in \llbracket B \rrbracket_\rho \hatoplus \llbracket C \rrbracket_\rho$ such that $(D/X)t \lras \inl(u')$. Then, there is a proof-term $u$ such that $t \lras \inl(u)$ with $u' = (D/X)u$. Then, $u \in \llbracket B \rrbracket_\rho$. By the induction hypothesis, $(D/X)u \in \llbracket B \rrbracket_\rho$. Let $t' \in \llbracket B \rrbracket_\rho \hatoplus \llbracket C \rrbracket_\rho$ such that $(D/X)t' \lras \inr(v')$. Then, there is a proof-term $v$ such that $t' \lras \inr(v)$ with $v' = (D/X)v$. Then, $v \in \llbracket C \rrbracket_\rho$. By the induction hypothesis, $(D/X)v \in \llbracket C \rrbracket_\rho$.
      \end{itemize}
    \item If $A = \oc B$:
      \begin{itemize}
	\item Let $t \in \hatbang\llbracket B\rrbracket_\rho$, then $t\in \SN$.
	\item Let $t \in \hatbang\llbracket B \rrbracket_\rho$ such that $t \lra t'$. Then $t' \in \SN$, and if $t' \lras \oc u$, $t \lras\oc u$. Therefore, $u \in \llbracket B \rrbracket_\rho$. 
        \item Let $t$ be a proof-term that is not an introduction such that $Red(t) \subseteq \hatbang\llbracket B \rrbracket_\rho$. Since $Red(t) \subseteq \SN$, $t \in \SN$. If $t \lras \oc u$, the rewrite sequence has at least one step because $t$ is not an introduction. Then, there is a proof-term $t' \in Red(t)$ such that $t' \lras \oc u$. Therefore, $u \in \llbracket B \rrbracket_\rho$. 
	\item Let $t \in \hatbang \llbracket B \rrbracket_\rho$ such that $(C/X)t \lras \oc u'$. Then, there is a proof-term $u$ such that $t \lras \oc u$ with $u' = (C/X)u$. Then, $u \in \llbracket B \rrbracket_\rho$. By the induction hypothesis, $(C/X)u \in \llbracket B \rrbracket_\rho$. 
      \end{itemize}
    \item If $A = \forall X.B$:
      \begin{itemize}
        \item Let $t \in \llbracket \forall X.B \rrbracket_\rho$ such that $t \lras \Lambda X.u$, then for every proposition $C$ and every $E \in \mathcal R$, $(C/X)u \in \llbracket B \rrbracket_{\rho, E/X}$. By the induction hypothesis, $\llbracket B \rrbracket_{\rho, E/X} \subseteq \SN$. Then, $(C/X)u \in \SN$ and therefore $t \in \SN$.
        \item Let $t \in \llbracket \forall X.B \rrbracket_\rho$ such that $t \lras t'$. Then $t' \in \SN$, and if $t' \lras \Lambda X.u$, $t \lras \Lambda X.u$. Therefore, for every proposition $C$ and every $E \in \mathcal R$, $(C/X)u \in \llbracket B \rrbracket_{\rho, E/X}$.
        \item Let $t$ be a proof-term that is not an introduction such that $Red(t) \subseteq \llbracket \forall X.B \rrbracket_\rho$. Since $Red(t) \subseteq \SN$, $t \in \SN$. If $t \lras \Lambda X.u$, the rewrite sequence has at least one step because $t$ is not an introduction. Then, there is a proof-term $t' \in Red(t)$ such that $t' \lras \Lambda X.u$. Therefore, for every $E \in \mathcal R$ and every proposition $C$, $(C/X)u \in \llbracket B \rrbracket_{\rho, E/X}$.
        \item Let $t \in \llbracket \forall X.B \rrbracket_\rho$. Let $C$ be a proposition and $E \in \mathcal R$, we have that $(D/X)((C/Y)t) = (Y/Z)((C/Y)((Z/Y)D/X)t)$, where $Z$ is a fresh variable. Since $(Z/Y)D$ is a proposition, $((Z/Y)D/X)t \in \llbracket B \rrbracket_{\rho, E/X}$. Therefore, by the induction hypothesis twice, $(Y/Z)((C/Y)((Z/Y)D/X)t) \in \llbracket B \rrbracket_{\rho, E/X}$.
      \end{itemize}
  \end{itemize}
\end{proof}

\subsection{Proof of Adequacy}
\label{proof:Adequacy}

The following auxiliary lemma is needed to prove the adequacy of proposition application.

\begin{lemma}
  \label{lem:substRho}
  For any $A, B$, and valuation $\rho$, $\llbracket (B/X)A \rrbracket_\rho = \llbracket A \rrbracket_{\rho, \llbracket B \rrbracket_\rho/X}$.
\end{lemma}
\begin{proof}
  By induction on $A$. Let $\rho' = \rho, \llbracket B \rrbracket_\rho/X$.
  \begin{itemize}
    \item If $A = X$, we have $\llbracket X \rrbracket_{\rho'} = \rho'(X) = \llbracket B \rrbracket_\rho = \llbracket (B/X)X \rrbracket_\rho$.
    \item If $A = Y \neq X$, we have $\llbracket Y \rrbracket_{\rho'} = \rho'(Y) = \rho(Y) = \llbracket Y \rrbracket_\rho = \llbracket (B/X)Y \rrbracket_\rho$.
    \item If $A = \one$, we have $\llbracket (B/X)\one \rrbracket_\rho = \SN = \llbracket \one \rrbracket_{\rho'}$.
    \item If $A = C \multimap D$, we have $\llbracket (B/X)(C \multimap D) \rrbracket_\rho = \llbracket (B/X)C \rrbracket_\rho \hatmultimap \llbracket (B/X)D \rrbracket_\rho$. By the induction hypothesis, $\llbracket (B/X)C \rrbracket_\rho = \llbracket C \rrbracket_{\rho'}$ and $\llbracket (B/X)D \rrbracket_\rho = \llbracket D \rrbracket_{\rho'}$. Therefore, $\llbracket (B/X)(C \multimap D) \rrbracket_\rho = \llbracket C \multimap D \rrbracket_{\rho'}$.
    \item If $A = C \otimes D$, we have $\llbracket (B/X)(C \otimes D) \rrbracket_\rho = \llbracket (B/X)C \rrbracket_\rho \hatotimes \llbracket (B/X)D \rrbracket_\rho$. By the induction hypothesis, $\llbracket (B/X)C \rrbracket_\rho = \llbracket C \rrbracket_{\rho'}$ and $\llbracket (B/X)D \rrbracket_\rho = \llbracket D \rrbracket_{\rho'}$. Therefore, $\llbracket (B/X)(C \otimes D) \rrbracket_\rho = \llbracket C \otimes D \rrbracket_{\rho'}$.
    \item If $A = \top$, we have $\llbracket (B/X)\top \rrbracket_\rho = \SN = \llbracket \top \rrbracket_{\rho'}$.
    \item If $A = \zero$, we have $\llbracket (B/X)\zero \rrbracket_\rho = \SN = \llbracket \zero \rrbracket_{\rho'}$.
    \item If $A = C \with D$, we have $\llbracket (B/X)(C \with D) \rrbracket_\rho = \llbracket (B/X)C \rrbracket_\rho \hatand \llbracket (B/X)D \rrbracket_\rho$. By the induction hypothesis, $\llbracket (B/X)C \rrbracket_\rho = \llbracket C \rrbracket_{\rho'}$ and $\llbracket (B/X)D \rrbracket_\rho = \llbracket D \rrbracket_{\rho'}$. Therefore, $\llbracket (B/X)(C \with D) \rrbracket_\rho = \llbracket C \with D \rrbracket_{\rho'}$.
    \item If $A = C \oplus D$, we have $\llbracket (B/X)(C \oplus D) \rrbracket_\rho = \llbracket (B/X)C \rrbracket_\rho \hatoplus \llbracket (B/X)D \rrbracket_\rho$. By the induction hypothesis, $\llbracket (B/X)C \rrbracket_\rho = \llbracket C \rrbracket_{\rho'}$ and $\llbracket (B/X)D \rrbracket_\rho = \llbracket D \rrbracket_{\rho'}$. Therefore, $\llbracket (B/X)(C \oplus D) \rrbracket_\rho = \llbracket C \oplus D \rrbracket_{\rho'}$.
    \item If $A = \oc C$, we have $\llbracket (B/X)(\oc C) \rrbracket_\rho = \hatbang \llbracket (B/X)C \rrbracket_\rho$. By the induction hypothesis, $\llbracket (B/X)C \rrbracket_{\rho} = \llbracket C \rrbracket_{\rho'}$. Therefore, $\llbracket (B/X)(\oc C) \rrbracket_{\rho} = \llbracket \oc C \rrbracket_{\rho'}$.
    \item If $A = \forall Y.C$, we have $\llbracket (B/X) \forall Y.C \rrbracket_\rho = \llbracket \forall Y.(B/X)C \rrbracket_\rho$. By definition, $t \in \llbracket \forall Y.(B/X)C \rrbracket_\rho$ if and only if $t \in \SN$ and if $t \lras \Lambda Y.u$, then for every proposition $D$ and every $E \in \mathcal R$, $(D/Y)u \in \llbracket (B/X)C \rrbracket_{\rho, E/Y}$. By the induction hypothesis, $\llbracket (B/X)C \rrbracket_{\rho, E/Y} = \llbracket C \rrbracket_{\rho', E/Y}$. Therefore, $\llbracket \forall Y.(B/X)C \rrbracket_\rho = \llbracket \forall Y.C \rrbracket_{\rho'}$.
      \qed
  \end{itemize}
\end{proof}

In Lemmas~\ref{lem:sum} to~{\ref{lem:typeapplication}}, we prove the adequacy of each proof-term constructor. If $t$ is a strongly normalising proof-term, we write $|t|$ for the maximum length of a reduction sequence issued from $t$.

\begin{lemma}[Normalisation of a sum]
    \label{lem:terminationsum}
    If $t$ and $u$ strongly normalise, then so does $t \plus u$. 
  \end{lemma}
  \begin{proof}
    We prove that all the one-step reducts of 
    $t \plus u$ strongly normalise, by induction first on 
    $|t| + |u|$ and then on the size of $t$. 
  
    If the reduction takes place in $t$ or in $u$ we apply the induction
    hypothesis.
    Otherwise, the reduction occurs at the root and the rule used is either
    \begin{align*}
      {a.\star} \plus {b.\star} &\lra  (a + b).\star\\
      (\lambda x^A.t') \plus (\lambda x^A.u')
      &\lra  \lambda x^A.(t' \plus u')\\
      \langle \rangle \plus \langle \rangle &\lra \langle \rangle\\
      \pair{t'_1}{t'_2} \plus \pair{u'_1}{u'_2}
      &\lra  \pair{t'_1 \plus u'_1}{t'_2 \plus u'_2}\\
      (\Lambda X.t') \plus (\Lambda X.u') &\lra
      \Lambda X.(t' \plus u')\\
      \oc t' \plus \oc u' &\lra \oc(t' \plus u')\\
      t \plus u &\lra t\\
      t \plus u &\lra u
    \end{align*}
  
    In the first case, the proof-term $(a + b).\star$ is irreducible, hence it
    strongly normalises. In the second, {and fifth, and sixth}, 
    by induction hypothesis, the proof-term $t' \plus u'$
    strongly normalises, thus so do the proof-terms
    $\lambda x^A.(t' \plus u')$, { $\Lambda X.(t' \plus u')$, and $\oc (t'\plus u')$}.
    In the third, the proof-term $\langle \rangle$ is irreducible, hence it
    strongly normalises. 
    In the fourth, 
    by induction hypothesis, the proof-terms
    $t'_1 \plus u'_1$ and $t'_2 \plus u'_2$
    strongly normalise, hence so does the proof-term
    $\pair{t'_1 \plus u'_1}{t'_2 \plus u'_2}$.
    In the seventh and eighth, the proof-terms $t$ and $u$ strongly normalise. 
    \qed
  \end{proof}
  
  \begin{lemma}[Normalisation of a product]
  \label{lem:terminationprod}
  If $t$ strongly normalises, then so does $a \bullet t$. 
  \end{lemma}
  \begin{proof}
  We prove that all the one-step reducts of 
  $a \bullet t$ strongly normalise, by induction first on 
  $|t|$ and then on the size of $t$. 
  
  If the reduction takes place in $t$, we apply the induction
  hypothesis.
  Otherwise, the reduction occurs at the root, and the rule used is either
  \begin{align*}
  a \bullet b.\star &\lra  (a \times b).\star\\
  a \bullet (\lambda x^A.t') 
  &\lra  \lambda x^A. a \bullet t'\\
  a \bullet \langle \rangle &\lra  \langle \rangle\\
  a \bullet \pair{t'_1}{t'_2} 
      &\lra  \pair{a \bullet t'_1}{a \bullet t'_2}\\
  a \bullet (\Lambda X.t')
      &\lra \Lambda X.a \bullet t'\\
  a \bullet \oc t'
      &\lra \oc (a \bullet t')\\
  a \bullet t &\lra t
  \end{align*}
  In the first case, the proof-term $(a \times b).\star$ is irreducible,
  hence it strongly normalises. In the second, { fifth, and sixth}, by induction hypothesis,
  the proof-term $a \bullet t'$ strongly normalises, thus so do the proof-terms
  $\lambda x^A.a \bullet t'$, { $\Lambda X. a \bullet t'$, and $\oc (a \bullet t')$}.  In the third, the proof-term $\langle
  \rangle$ is irreducible, hence it strongly normalises.  In the
  fourth, by induction hypothesis, the proof-terms $a \bullet t'_1$ and
  $a \bullet t'_2$ strongly normalise, hence so does the proof-term
  $\pair{a \bullet t'_1}{a \bullet t'_2}$. In the seventh, the proof-term $t$
  strongly normalises. \qed
  \end{proof}

\begin{lemma}[Adequacy of $\plus$]
\label{lem:sum}
{For every valuation $\rho$, if $\Xi;\Gamma \vdash t_1:A$, $\Xi;\Gamma \vdash t_2:A$,} $t_1 \in \llbracket A \rrbracket_{\rho}$ and $t_2 \in \llbracket A
\rrbracket_{\rho}$, then $t_1 \plus t_2 \in \llbracket A \rrbracket_{\rho}$.
\end{lemma}
\begin{proof}
  By induction on $A$.
  The proof-terms $t_1$ and $t_2$ strongly normalise.  Thus, by
  Lemma~\ref{lem:terminationsum}, the proof-term $t_1 \plus t_2$ strongly
  normalises.
  Furthermore:
  \begin{itemize}
    \item If the proposition $A$ has the form $X$, then $t_1, t_2 \in \rho(X) \in \mathcal R$. Using CR3, we need to prove that each of the one step reducts of $t_1 \plus t_2$ is in $\rho(X)$. Since $t_1, t_2 \in \rho(X)$, we have that $t_1, t_2 \in \SN$. We proceed by induction on $|t_1| + |t_2|$.
      \begin{itemize}
	\item If $t_1 \lra t_1'$, then $t_1 \plus t_2 \lra t_1' \plus t_2$. By CR2, $t_1' \in \rho(X)$. Since $|t_1'| < |t_1|$, by the induction hypothesis $t_1' \plus t_2 \in \rho(X)$.
	\item  If $t_2 \lra t_2'$, then $t_1 \plus t_2 \lra t_1 \plus t_2'$. By CR2, $t_2' \in \rho(X)$. Since $|t_2'| < |t_2|$, by the induction hypothesis $t_1 \plus t_2' \in \rho(X)$.
	\item By ultra-reduction, we have that $t_1 \plus t_2 \lra t_1$. By hypothesis, $t_1 \in \rho(X)$.
	\item By ultra-reduction, we have that $t_1 \plus t_2 \lra t_2$. By hypothesis, $t_2 \in \rho(X)$.
	\item There are no more cases since $t_1$ and $t_2$ are proof-terms of $X$.
      \end{itemize}
    \item If the proposition $A$ has the form $B \multimap C$, and $t_1
      \plus t_2 \lra^* \lambda x^B. v$ then either $t_1 \lra^*
      \lambda x^B. u_1$, $t_2 \lra^* \lambda x^B. u_2$, and $u_1
      \plus u_2 \lra^* v$, or $t_1 \lra^* \lambda x^B. v$, or $t_2
      \lra^* \lambda x^B. v$.

      In the first case, as $t_1$ and $t_2$ are in $\llbracket A
      \rrbracket_{\rho}$, for every $w$ in $\llbracket B \rrbracket_{\rho}$, $(w/x)u_1
      \in \llbracket C \rrbracket_{\rho}$ and $(w/x)u_2 \in \llbracket C
      \rrbracket_{\rho}$.  By induction hypothesis, $(w/x)(u_1 \plus u_2) =
      (w/x)u_1 \plus (w/x)u_2 \in \llbracket C \rrbracket_{\rho}$ and by
      {CR2}, $(w/x)v \in \llbracket C \rrbracket_{\rho}$.

      In the second and the third, as $t_1$ and $t_2$ are in $\llbracket A
      \rrbracket_{\rho}$, for every $w$ in $\llbracket B \rrbracket_{\rho}$, $(w/x)v \in
      \llbracket C \rrbracket_{\rho}$.

    \item If the proposition $A$ has the form $B \otimes C$, and $t_1
      \plus t_2 \lra^* v \otimes v'$ then $t_1 \lra^* v \otimes v'$, or
      $t_2 \lra^* v \otimes v'$.  As $t_1$ and $t_2$ are in $\llbracket
      A \rrbracket_{\rho}$, $v \in \llbracket B \rrbracket_{\rho}$ and $v' \in
      \llbracket C \rrbracket_{\rho}$.

    \item If the proposition $A$ has the form $B \with C$, and $t_1 \plus t_2
      \lra^* \pair{v}{v'}$ then $t_1 \lra^* \pair{u_1}{u'_1}$, $t_2 \lra^*
      \pair{u_2}{u'_2}$, $u_1 \plus u_2 \lra^* v$, and $u'_1 \plus u'_2
      \lra^* v'$, or $t_1 \lra^* \pair{v}{v'}$, or $t_2 \lra^*
      \pair{v}{v'}$.

      In the first case, as $t_1$ and $t_2$ are in $\llbracket A
      \rrbracket_{\rho}$, $u_1$ and $u_2$ are in $\llbracket B \rrbracket_{\rho}$ and
      $u'_1$ and $u'_2$ are in $\llbracket C \rrbracket_{\rho}$.  By induction
      hypothesis, $u_1 \plus u_2 \in \llbracket B \rrbracket_{\rho}$ and $u'_1
      \plus u'_2 \in \llbracket C \rrbracket_{\rho}$ and by {CR2}, $v \in \llbracket B \rrbracket_{\rho}$ and $v'
      \in \llbracket C \rrbracket_{\rho}$.

      In the second and the third, as $t_1$ and $t_2$ are in $\llbracket A
      \rrbracket_{\rho}$, $v \in \llbracket B \rrbracket_{\rho}$ and $v' \in \llbracket
      C \rrbracket_{\rho}$.

    \item If the proposition $A$ has the form $B \oplus C$, and $t_1 \plus
      t_2 \lra^* \inl(v)$ then $t_1 \lra^* \inl(v)$ or $t_2 \lra^*
      \inl(v)$.  As $t_1$ and $t_2$ are in $\llbracket A \rrbracket_{\rho}$, $v
      \in \llbracket B \rrbracket_{\rho}$.

      The proof is similar if $t_1 \plus t_2 \lra^* \inr(v)$.
    \item
	If the proposition $A$ has the form $\oc B$, and $t_1 \plus t_2
	\lra^* \oc v$ then $t_1 \lra^*\oc u_1$, $t_2 \lra^*
	\oc u_2$, $u_1 \plus u_2 \lra^* v$, or $t_1 \lra^* \oc v$, or $t_2 \lra^* \oc v$.
	In the first case, as $t_1$ and $t_2$ are in $\llbracket A
	\rrbracket_{\rho}$, $u_1$ and $u_2$ are in $\llbracket B \rrbracket_{\rho}$.
	By induction hypothesis, $u_1 \plus u_2 \in \llbracket B \rrbracket_{\rho}$,
	and by {CR2}, $v \in \llbracket B \rrbracket_{\rho}$.
	In the second and the third, as $t_1$ and $t_2$ are in $\llbracket A
	\rrbracket_{\rho}$, $v \in \llbracket B \rrbracket_{\rho}$.
    \item
	If the proposition $A$ has the form $\forall X.B$, and $t_1 \plus t_2 \lras \Lambda X.v$ then either $t_1 \lras \Lambda X.u_1$, $t_2 \lras \Lambda X.u_2$ and $u_1 \plus u_2 \lras v$, or $t_1 \lras \Lambda X.v$, or $t_2 \lras \Lambda X.v$. Let $E \in \mathcal R$ and $C$ be a proposition.
	In the first case, as $t_1, t_2 \in \llbracket \forall X.B \rrbracket_\rho$, $(C/X)u_1, (C/X)u_2 \in \llbracket B \rrbracket_{\rho, E/X}$. By the induction hypothesis, $(C/X)u_1 \plus (C/X)u_2 = (C/X)(u_1 \plus u_2) \in \llbracket B \rrbracket_{\rho, E/X}$. By CR2, since $(C/X)(u_1 \plus u_2) \lras (C/X)v$, $(C/X)v \in \llbracket B \rrbracket_{\rho, E/X}$.
	In the second and third cases, as $t_1, t_2 \in \llbracket \forall X.B \rrbracket_\rho$, $(C/X)v \in \llbracket B \rrbracket_{\rho, E/X}$.\qed
  \end{itemize}
\end{proof}

\begin{lemma}[Adequacy of $\bullet$]
\label{lem:prod}
{For every valuation $\rho$, if $\Xi;\Gamma \vdash t:A$ and} $t \in \llbracket A \rrbracket_{\rho}$, then $a \bullet t \in \llbracket A
\rrbracket_{\rho}$.
\end{lemma}

\begin{proof}
  By induction on $A$.  The proof-term $t$ strongly normalises.  Thus, by
  Lemma~\ref{lem:terminationprod}, the proof-term $a \bullet t$ strongly
  normalises.  Furthermore:

  \begin{itemize}
      \item If the proposition $A$ has the form $X$, then $t \in \rho(X) \in \mathcal R$. Using CR3, we need to prove that each of the one step reducts of $a \bullet t$ is in $\rho(X)$. Since $t \in \rho(X)$, we have that $t \in \SN$. We proceed by induction on $|t|$.
	\begin{itemize}
	  \item If $t \lra t'$, $a \bullet t \lra a \bullet t'$. By CR2, $t' \in \rho(X)$. Since $|t'| < |t|$, by the induction hypothesis, $a \bullet t' \in \rho(X)$.
	  \item By ultra-reduction, we have that $a \bullet t \lra t$. By hypothesis, $t \in \rho(X)$.
	  \item There are no more cases since $t_1$ and $t_2$ are proof-terms of $X$.
	\end{itemize}
    \item If the proposition $A$ has the form $B \multimap C$, and $a
      \bullet t \lra^* \lambda x^B. v$ then either $t \lra^* \lambda
      x^B. u$ and $a \bullet u \lra^* v$, or $t \lra^* \lambda
      x^B. v$.

      In the first case, as $t$ is in $\llbracket A \rrbracket_{\rho}$, for
      every $w$ in $\llbracket B \rrbracket_{\rho}$, $(w/x)u \in \llbracket C
      \rrbracket_{\rho}$.  By induction hypothesis, $(w/x) (a \bullet u) = a
      \bullet (w/x)u \in \llbracket C \rrbracket_{\rho}$ and by
      {CR2}, $(w/x)v \in \llbracket C \rrbracket_{\rho}$.

      In the second, as $t$ is in $\llbracket A \rrbracket_{\rho}$, for every $w$
      in $\llbracket B \rrbracket_{\rho}$, $(w/x)v \in \llbracket C \rrbracket_{\rho}$.

    \item If the proposition $A$ has the form $B \otimes C$, and $a
      \bullet t \lra^* v \otimes v'$ then $t \lra^* v \otimes v'$.  As $t$
      is in $\llbracket A \rrbracket_{\rho}$, $v \in \llbracket B \rrbracket_{\rho}$ and
      $v' \in \llbracket C \rrbracket_{\rho}$.

    \item If the proposition $A$ has the form $B \with C$, and $a \bullet t
      \lra^* \pair{v}{v'}$ then $t \lra^* \pair{u}{u'}$, $a \bullet u
      \lra^* v$, and $a \bullet u' \lra^* v'$, or $t \lra^* \pair{v}{v'}$.

      In the first case, as $t$ is in $\llbracket A \rrbracket_{\rho}$, $u$ is in
      $\llbracket B \rrbracket_{\rho}$ and $u'$ is in $\llbracket C \rrbracket_{\rho}$.
      By induction hypothesis, $a \bullet u \in \llbracket B \rrbracket_{\rho}$
      and $a \bullet u' \in \llbracket C \rrbracket_{\rho}$ and by
      {CR2}, $v \in \llbracket B \rrbracket_{\rho}$ and $v'
      \in \llbracket C \rrbracket_{\rho}$.

      In the second, as $t$ is in $\llbracket A \rrbracket_{\rho}$, $v \in
      \llbracket B \rrbracket_{\rho}$ and $v' \in \llbracket C \rrbracket_{\rho}$.

    \item If the proposition $A$ has the form $B \oplus C$, and $a \bullet t
      \lra^* \inl(v)$ then $t \lra^* \inl(v)$.
      Then, by {CR2},
      $\inl(v) \in \llbracket A \rrbracket_{\rho}$ hence, $v \in \llbracket
      B \rrbracket_{\rho}$.

      The proof is similar if $a \bullet t \lra^* \inr(v)$.
    \item 
	If the proposition $A$ has the form $\oc B$, and $a \bullet t \lra^* \oc v$ then $t \lra^* \oc u$ and $a \bullet u \lra^* v$, or $t \lra^* \oc v$.
	In the first case, as $t$ is in $\llbracket A \rrbracket_{\rho}$, $u$ is in
	$\llbracket B \rrbracket_{\rho}$.
	By induction hypothesis, $a \bullet u \in \llbracket B \rrbracket_{\rho}$
	and by {CR2}, $v \in \llbracket B \rrbracket_{\rho}$.
	In the second, as $t$ is in $\llbracket A \rrbracket_{\rho}$, $v \in
	\llbracket B \rrbracket_{\rho}$.

    \item
	If the proposition $A$ has the form $\forall X.B$, and $a \bullet t \lras \Lambda X.v$ then either $t \lras \Lambda X.u$ and $a \bullet u \lras v$, or $t \lras \Lambda X.v$. Let $E \in \mathcal R$ and $C$ be a proposition.
	In the first case, as $t \in \llbracket \forall X.B \rrbracket_\rho$, $(C/X)u \in \llbracket B \rrbracket_{\rho, E/X}$. By the induction hypothesis, $a \bullet (C/X)u = (C/X) (a \bullet u) \in \llbracket B \rrbracket_{\rho, E/X}$. By CR2, since $(C/X)(a \bullet u) \lras (C/X)v$, $(C/X)v \in \llbracket B \rrbracket_{\rho, E/X}$.
	In the second case, as $t \in \llbracket \forall X.B \rrbracket_\rho$, $(C/X)v \in \llbracket B \rrbracket_{\rho, E/X}$.\qed
  \end{itemize}
\end{proof}

\begin{lemma}[Adequacy of $a.\star$]
  \label{lem:star}
  {For every valuation $\rho$,} we have $a.\star \in \llbracket \one \rrbracket_{\rho}$.
  \end{lemma}
  
  \begin{proof}
  As $a.\star$ is irreducible, it strongly normalises, hence
  $a.\star \in \llbracket \one \rrbracket_{\rho}$. \qed
  \end{proof}
  
  \begin{lemma}[Adequacy of $\lambda$]
  \label{lem:abstraction}
  {For every valuation $\rho$,} if, for all $u \in \llbracket A \rrbracket_{\rho}$, $(u/x)t \in \llbracket B
  \rrbracket_{\rho}$, then $\lambda x^C.t \in \llbracket A \multimap B
  \rrbracket_{\rho}$.
  \end{lemma}
  
  \begin{proof}
  By Lemma~\ref{lem:Var}, $x \in \llbracket A \rrbracket_{\rho}$, thus
  $t = (x/x)t \in \llbracket B \rrbracket_{\rho}$. Hence, $t$ strongly
  normalises.  Consider a reduction sequence issued from $\lambda
  x^C.t$.  This sequence can only reduce $t$ hence it is finite. Thus,
  $\lambda x^C.t$ strongly normalises.
  
  Furthermore, if $\lambda x^C.t \lras \lambda x^C.t'$, then
  $t \lra^* t'$.  Let $u \in \llbracket A \rrbracket_{\rho}$,
  $(u/x)t \lra^* (u/x)t'$.
  As $(u/x)t \in \llbracket B
  \rrbracket_{\rho}$, by {CR2}, $(u/x)t' \in
  \llbracket B \rrbracket_{\rho}$.\qed
  \end{proof}
  
  \begin{lemma}[Adequacy of $\otimes$]
  \label{lem:tensor}
  {For every valuation $\rho$,} if $t_1 \in \llbracket A \rrbracket_{\rho}$ and $t_2 \in \llbracket B
  \rrbracket_{\rho}$, then $t_1 \otimes t_2 \in \llbracket A \otimes B
  \rrbracket_{\rho}$.
  \end{lemma}
  
  \begin{proof}
  The proof-terms $t_1$ and $t_2$ strongly normalise. Consider a reduction
  sequence issued from $t_1 \otimes t_2$.  This sequence can only reduce
  $t_1$ and $t_2$, hence it is finite.  Thus, $t_1 \otimes t_2$ strongly
  normalises.
  
  Furthermore, if $t_1 \otimes t_2 \lras t'_1 \otimes t'_2$,
  then $t_1 \lra^* t'_1$ and $t_2 \lra^* t'_2$.  By {CR2},
  $t'_1 \in \llbracket A \rrbracket_{\rho}$ and $t'_2 \in
  \llbracket B \rrbracket_{\rho}$. \qed
  \end{proof}

  \begin{lemma}[Adequacy of $\langle \rangle$]
  \label{lem:unit}
  {For every valuation $\rho$,} we have $\langle \rangle \in \llbracket \top \rrbracket_{\rho}$.
  \end{lemma}
  
  \begin{proof}
  As $\langle \rangle$ is irreducible, it strongly normalises, hence
  $\langle \rangle \in \llbracket \top \rrbracket_{\rho}$. \qed
  \end{proof}
  
  \begin{lemma}[Adequacy of $\pair{}{}$]
  \label{lem:pair}
  {For every valuation $\rho$,} if $t_1 \in \llbracket A \rrbracket_{\rho}$ and $t_2 \in \llbracket B
  \rrbracket_{\rho}$, then $\pair{t_1}{t_2} \in \llbracket A \with B
  \rrbracket_{\rho}$.
  \end{lemma}
  
  \begin{proof}
  The proof-terms $t_1$ and $t_2$ strongly normalise. Consider a reduction
  sequence issued from $\pair{t_1}{t_2}$.  This sequence can only
  reduce $t_1$
  and $t_2$, hence it is finite.  Thus, $\pair{t_1}{t_2}$
  strongly normalises.
  
  Furthermore, if $\pair{t_1}{t_2} \lras \pair
  {t'_1}{t'_2}$, then $t_1 \lra^* t'_1$ and $t_2 \lra^* t'_2$.  By
  {CR2}, $t'_1 \in \llbracket A \rrbracket_{\rho}$ and
  $t'_2 \in \llbracket B \rrbracket_{\rho}$. \qed
  \end{proof}
  
  \begin{lemma}[Adequacy of $\inl$]
  \label{lem:inl}
  {For every valuation $\rho$,} if $t \in \llbracket A \rrbracket_{\rho}$, then $\inl(t) \in \llbracket A
  \oplus B \rrbracket_{\rho}$.
  \end{lemma}
  
  \begin{proof}
  The proof-term $t$ strongly normalises. Consider a reduction sequence
  issued from $\inl(t)$.  This sequence can only reduce $t$, hence it is
  finite.  Thus, $\inl(t)$ strongly normalises.
  
  Furthermore, if $\inl(t) \lras \inl(t')$, then $t \lra^* t'$.  By {CR2}, $t' \in \llbracket A
  \rrbracket_{\rho}$. And $\inl(t)$ never reduces to $\inr(t')$. \qed
  \end{proof}
  
  \begin{lemma}[Adequacy of $\inr$]
  \label{lem:inr}
  {For every valuation $\rho$,} if $t \in \llbracket B \rrbracket_{\rho}$, then $\inr(t) \in \llbracket A
  \oplus B \rrbracket_{\rho}$.
  \end{lemma}
  
  \begin{proof}
    Similar to the proof of Lemma~\ref{lem:inl}. \qed
  \end{proof}

  \begin{lemma}[Adequacy of $\oc$]
    \label{lem:bang}
    For every valuation $\rho$, if $t \in \llbracket A\rrbracket_\rho$, then $\oc t \in \llbracket\oc A\rrbracket_\rho$.
  \end{lemma}
  \begin{proof}
    The proof-term $t$ strongly normalises. Consider a reduction sequence
    issued from $\oc t$.  This sequence can only reduce $t$, hence it is
    finite.  Thus, $\oc t$ strongly normalises.
    Furthermore, if $\oc t \lras \oc t'$, then $t \lra^* t'$.  By CR2, $t' \in \llbracket A
    \rrbracket_\rho$. \qed
  \end{proof}
  
  \begin{lemma}[Adequacy of $\Lambda$]
    \label{lem:Lambda}
    If $t \in \llbracket A \rrbracket_{\rho, E/X}$ for every $E \in \mathcal R$, then $\Lambda X.t \in \llbracket \forall X.A \rrbracket_\rho$.
  \end{lemma}
  \begin{proof}
    Let $B$ be a proposition, and $F \in \mathcal R$. By Lemma~\ref{lem:typeinterpretationsarerc}, $\llbracket A \rrbracket_{\rho, F/X} \in \mathcal R$. Then, $t \in \SN$, and $\Lambda X.t \in \SN$. By CR4, $(B/X)t \in \llbracket A \rrbracket_{\rho, F/X}$. \qed
  \end{proof}

  \begin{lemma}[Adequacy of $\elimone$]
  \label{lem:elimone}
  {For every valuation $\rho$,} if $t_1 \in \llbracket \one \rrbracket_{\rho}$ and $t_2 \in \llbracket C \rrbracket_{\rho}$, 
  then $\elimone(t_1,t_2) \in \llbracket C \rrbracket_{\rho}$.
  \end{lemma}
  
  \begin{proof}
  The proof-terms $t_1$ and $t_2$ strongly normalise.  We prove, by
  induction on $|t_1| + |t_2|$, that $\elimone(t_1,t_2)
  \in \llbracket C \rrbracket_{\rho}$.  Using {CR3}, we only
  need to prove that every of its one step reducts is in $\llbracket C
  \rrbracket_{\rho}$.  If the reduction takes place in $t_1$ or $t_2$, then we
  apply {CR2} and the induction hypothesis.
  
  Otherwise, the proof-term $t_1$ is $a.\star$ and the
  reduct is $a \bullet t_2$. We conclude with Lemma~\ref{lem:prod}. \qed
  \end{proof}
  
  \begin{lemma}[Adequacy of application]
  \label{lem:application}
  {For every valuation $\rho$,} if $t_1 \in \llbracket A \multimap B \rrbracket_{\rho}$ and $t_2 \in
  \llbracket A \rrbracket_{\rho}$, then $t_1~t_2 \in \llbracket B
  \rrbracket_{\rho}$.
  \end{lemma}
  
  \begin{proof}
  The proof-terms $t_1$ and $t_2$ strongly normalise. We prove, by induction
  on $|t_1| + |t_2|$, that $t_1~t_2 \in \llbracket B \rrbracket_{\rho}$. Using
  {CR3}, we only need to prove that every of its one
  step reducts is in $\llbracket B \rrbracket_{\rho}$.  If the reduction takes
  place in $t_1$ or in $t_2$, then we apply {CR2}
  and the induction hypothesis.
  
  Otherwise, the proof-term $t_1$ has the form $\lambda x^C.u$ and the reduct
  is $(t_2/x)u$.  As $\lambda x^C.u \in \llbracket A \multimap B
  \rrbracket_{\rho}$, we have $(t_2/x)u \in \llbracket B \rrbracket_{\rho}$. \qed
  \end{proof}
  
  \begin{lemma}[Adequacy of $\elimtens$]
  \label{lem:elimtens}
  {For every valuation $\rho$,} if $t_1 \in \llbracket A \otimes B \rrbracket_{\rho}$,
  for all $u$ in $\llbracket A \rrbracket_{\rho}$,
  for all $v$ in $\llbracket B \rrbracket_{\rho}$,
  $(u/x,v/y)t_2 \in \llbracket C \rrbracket_{\rho}$, 
  then $\elimtens(t_1, x^D y^E.t_2) \in \llbracket C \rrbracket_{\rho}$.
  \end{lemma}
  
  \begin{proof}
  By Lemma~\ref{lem:Var}, $x \in \llbracket A \rrbracket_{\rho}$ and $y \in
  \llbracket B \rrbracket_{\rho}$, thus $t_2 = (x/x,y/y)t_2 \in \llbracket C
  \rrbracket_{\rho}$.  Hence, $t_1$ and $t_2$ strongly normalise.  We prove, by
  induction on $|t_1| + |t_2|$, that $\elimtens(t_1,
  x^D y^E.t_2) \in \llbracket C \rrbracket_{\rho}$.  Using
  {CR3}, we only need to prove that every of its one
  step reducts is in $\llbracket C \rrbracket_{\rho}$.  If the reduction takes
  place in $t_1$ or $t_2$, then we apply {CR2} and the
  induction hypothesis. Otherwise, either:
  \begin{itemize}
  \item The proof-term $t_1$ has the form $w_2 \otimes w_3$ and the reduct is
    $(w_2/x,w_3/y)t_2$. As
    $w_2 \otimes w_3 \in \llbracket A \otimes B \rrbracket_{\rho}$, we
    have $w_2 \in \llbracket A \rrbracket_{\rho}$
  and $w_3 \in \llbracket B \rrbracket_{\rho}$. 
    Hence, $(w_2/x,w_3/y)t_2 \in
    \llbracket C \rrbracket_{\rho}$.
  
  \item The proof-term $t_1$ has the form $t_1' \plus t''_1$ and the
    reduct is\\ $\elimtens(t'_1, x^D y^E.t_2) \plus
    \elimtens(t''_1, x^D y^E.t_2)$. As $t_1
    \lra t'_1$ with an ultra-reduction rule, we have by
    {CR2}, $t'_1 \in \llbracket A \otimes B
    \rrbracket_{\rho}$.  Similarly, $t''_1 \in \llbracket A \otimes B
    \rrbracket_{\rho}$.  Thus, by induction hypothesis, $\elimtens(t'_1,
    x^D y^E.t_2) \in \llbracket A \otimes B \rrbracket_{\rho}$
    and $\elimtens(t''_1, x^D y^E.t_2) \in \llbracket A
    \otimes B \rrbracket_{\rho}$.  We conclude with Lemma~\ref{lem:sum}.
  
  \item The proof-term $t_1$ has the form $a \bullet t_1'$ and the
    reduct is $a \bullet \elimtens(t'_1, x^D y^E.t_2)$. As $t_1
    \lra t'_1$ with an ultra-reduction rule, we have by
    {CR2}, $t'_1 \in \llbracket A \oplus B
    \rrbracket_{\rho}$.  
    Thus, by induction hypothesis, $\elimtens(t'_1,
    x^D y^E.t_2) \in \llbracket A \otimes B \rrbracket_{\rho}$.
    We conclude with Lemma~\ref{lem:prod}. \qed
  \end{itemize}
  \end{proof}
  
  \begin{lemma}[Adequacy of $\elimzero$]
  \label{lem:elimzero}
  {For every valuation $\rho$,} if $t \in \llbracket \zero \rrbracket_{\rho}$, 
  then $\elimzero(t) \in \llbracket C \rrbracket_{\rho}$.
  \end{lemma}
  
  \begin{proof}
  The proof-term $t$ strongly normalises.  Consider a reduction sequence
  issued from $\elimzero(t)$.  This sequence can only reduce $t$, hence it
  is finite.  Thus, $\elimzero(t)$ strongly normalises.  Moreover, it
  never reduces to an introduction. \qed
  \end{proof}
  
  \begin{lemma}[Adequacy of $\elimwith^1$]
  \label{lem:elimwith1}
  {For every valuation $\rho$,} if $t_1 \in \llbracket A \with B \rrbracket_{\rho}$
  and, for all $u$ in $\llbracket A \rrbracket_{\rho}$,
  $(u/x)t_2 \in \llbracket C \rrbracket_{\rho}$, 
  then $\elimwith^1(t_1, x^D.t_2) \in \llbracket C \rrbracket_{\rho}$.
  \end{lemma}
  
  \begin{proof}
  By Lemma~\ref{lem:Var}, $x \in \llbracket A \rrbracket_{\rho}$
  thus $t_2 = (x/x)t_2 \in \llbracket C
  \rrbracket_{\rho}$.  Hence, $t_1$ and $t_2$ strongly normalise.  We prove, by
  induction on $|t_1| + |t_2|$, that $\elimwith^1(t_1, x^D.t_2)
  \in \llbracket C \rrbracket_{\rho}$.  Using {CR3}, we only
  need to prove that every of its one step reducts is in $\llbracket C
  \rrbracket_{\rho}$.  If the reduction takes place in $t_1$ or $t_2$, then we
  apply {CR2} and the induction hypothesis.
  
  Otherwise, the proof-term $t_1$ has the form $\pair{u}{v}$ and the
  reduct is $(u/x)t_2$.  As $\pair{u}{v} \in \llbracket A
  \with B \rrbracket_{\rho}$, we have $u \in \llbracket A \rrbracket_{\rho}$.
  Hence, $(u/x)t_2 \in \llbracket C \rrbracket_{\rho}$. \qed
  \end{proof}
  
  \begin{lemma}[Adequacy of $\elimwith^2$]
  \label{lem:elimwith2}
  {For every valuation $\rho$,} if $t_1 \in \llbracket A \with B \rrbracket_{\rho}$ and,
  for all $u$ in $\llbracket B \rrbracket_{\rho}$,
  $(u/x)t_2 \in \llbracket C \rrbracket_{\rho}$, 
  then $\elimwith^2(t_1, x^D.t_2) \in \llbracket C \rrbracket_{\rho}$.
  \end{lemma}
  
  \begin{proof}
  Similar to the proof of Lemma~\ref{lem:elimwith1}. \qed
  \end{proof}
  
  \begin{lemma}[Adequacy of $\elimplus$]
  \label{lem:elimplus}
  {For every valuation $\rho$,} if $t_1 \in \llbracket A \oplus B \rrbracket_{\rho}$, for all $u$ in $\llbracket A
  \rrbracket_{\rho}$, $(u/x)t_2 \in \llbracket C \rrbracket_{\rho}$, and, for all $v$
  in $\llbracket B \rrbracket_{\rho}$, $(v/y)t_3 \in \llbracket C \rrbracket_{\rho}$,
  then $\elimplus(t_1, x^D.t_2, y^E.t_3) \in \llbracket C \rrbracket_{\rho}$.
  \end{lemma}
  
  \begin{proof}
  By Lemma~\ref{lem:Var}, $x \in \llbracket A \rrbracket_{\rho}$, thus $t_2 =
  (x/x)t_2 \in \llbracket C \rrbracket_{\rho}$. In the same way, $t_3 \in
  \llbracket C \rrbracket_{\rho}$.  Hence, $t_1$, $t_2$, and $t_3$ strongly
  normalises.  We prove, by induction on $|t_1| + |t_2| + |t_3|$,
  that $\elimplus(t_1, x^D.t_2,
  y^E.t_3) \in \llbracket C \rrbracket_{\rho}$.  Using {CR3}, we
  only need to prove that every of its one step reducts
  is in $\llbracket C \rrbracket_{\rho}$.  If the reduction takes place in
  $t_1$, $t_2$, or $t_3$, then we apply {CR2} and
  the induction hypothesis. Otherwise, either:
  \begin{itemize}
  \item The proof-term $t_1$ has the form $\inl(w_2)$ and the reduct is
    $(w_2/x)t_2$. As $\inl(w_2) \in \llbracket A \oplus B \rrbracket_{\rho}$, we
    have $w_2 \in \llbracket A \rrbracket_{\rho}$.  Hence, $(w_2/x)t_2 \in
    \llbracket C \rrbracket_{\rho}$.
  
  \item The proof-term $t_1$ has the form $\inr(w_3)$ and the reduct is
    $(w_3/x)t_3$. As $\inr(w_3) \in \llbracket A \oplus B \rrbracket_{\rho}$, we
    have $w_3 \in \llbracket B \rrbracket_{\rho}$.  Hence, $(w_3/x)t_3 \in
    \llbracket C \rrbracket_{\rho}$.
  
  \item The proof-term $t_1$ has the form $t_1' \plus t''_1$ and the
    reduct is $\elimplus(t'_1, x^D.t_2, y^E.t_3) \plus
    \elimplus(t''_1, x^D.t_2, y^E.t_3)$. As $t_1
    \lra t'_1$ with an ultra-reduction rule, we have by
    {CR2}, $t'_1 \in \llbracket A \oplus B
    \rrbracket_{\rho}$.  Similarly, $t''_1 \in \llbracket A \oplus B
    \rrbracket_{\rho}$.  Thus, by induction hypothesis, $\elimplus(t'_1,
    x^D.t_2, y^E.t_3) \in \llbracket A \oplus B \rrbracket_{\rho}$
    and $\elimplus(t''_1, x^D.t_2, y^E.t_3) \in \llbracket A
    \oplus B \rrbracket_{\rho}$.  We conclude with Lemma~\ref{lem:sum}.
  
  \item The proof-term $t_1$ has the form $a \bullet t_1'$ and the
    reduct is $a \bullet \elimplus(t'_1, x^D.t_2, y^E.t_3)$. As $t_1
    \lra t'_1$ with an ultra-reduction rule, we have by
    {CR2}, $t'_1 \in \llbracket A \oplus B
    \rrbracket_{\rho}$.  
    Thus, by induction hypothesis, $\elimplus(t'_1,
    x^D.t_2, y^E.t_3) \in \llbracket A \oplus B \rrbracket_{\rho}$.
    We conclude with Lemma~\ref{lem:prod}. \qed
  \end{itemize}
  \end{proof}

\begin{lemma}[Adequacy of $\elimbang$]\label{lem:elimbang}
  For every valuation $\rho$, if $t_1 \in \llbracket \oc A \rrbracket_\rho$
  and, for all $u$ in $\llbracket A \rrbracket_\rho$,
  $(u/x)t_2 \in \llbracket B \rrbracket_\rho$, 
  then $\elimbang(t_1, x^C.t_2) \in \llbracket B \rrbracket_\rho$.
\end{lemma}

\begin{proof}
  By Lemma~\ref{lem:Var}, $x \in \llbracket A \rrbracket_\rho$
  thus $t_2 = (x/x)t_2 \in \llbracket B
  \rrbracket_\rho$.  Hence, $t_1$ and $t_2$ strongly normalise.  We prove, by
  induction on $|t_1| + |t_2|$, that $\elimbang(t_1, x^C.t_2)
  \in \llbracket B \rrbracket_\rho$.  Using CR3, we only
  need to prove that every of its one step reducts is in $\llbracket B
  \rrbracket_\rho$.  If the reduction takes place in $t_1$ or $t_2$, then we
  apply CR2 and the induction hypothesis.

  Otherwise, the proof-term $t_1$ has the form $\oc u$ and the
  reduct is $(u/x)t_2$.  As $\oc u \in \llbracket \oc A \rrbracket_\rho$, we have $u \in \llbracket A \rrbracket_\rho$.
  Hence, $(u/x)t_2 \in \llbracket B \rrbracket_\rho$. \qed
\end{proof}

\begin{lemma}[Adequacy of proposition application]
  \label{lem:typeapplication}
  If $t \in \llbracket \forall X.A \rrbracket_\rho$, then $t~B \in \llbracket (B/X)A \rrbracket_\rho$.
\end{lemma}
\begin{proof}
  Since $t \in \llbracket \forall X.A \rrbracket_\rho$, $t \in \SN$, therefore $t~B \in \SN$. By Lemma~\ref{lem:substRho}, it suffices to show that $t~B \in \llbracket A \rrbracket_{\rho, \llbracket B \rrbracket_\rho/X}$. We proceed by induction on $|t|$. Using CR3, we need to prove that each of its one step reducts is in $\llbracket A \rrbracket_{\rho, \llbracket B \rrbracket_\rho/X}$.
  \begin{itemize}
    \item If $t = \Lambda X.u$, then $t~B \lra (B/X)u$. Since $t \in \llbracket \forall X.A \rrbracket_\rho$, then $(B/X)u \in \llbracket A \rrbracket_{\rho, \llbracket B \rrbracket_\rho/X}$.
    \item If $|t| = 0$ and $t \neq \Lambda X.u$, then $Red(t~B)$ is empty.
    \item If $|t| > 0$, let $t \lra t'$. Then, $t~B \lra t'~B$. By CR2, $t' \in \llbracket \forall X.A \rrbracket_\rho$. Since $|t'| < |t|$, by the induction hypothesis $t'~B \in \llbracket \forall X.A \rrbracket_\rho$.
      \qed
  \end{itemize}
\end{proof}

\section{Proof of Section~\ref{sec:introductionthm}}\label{proof:introductionthm}
\introductions*
\begin{proof}
  By induction on $t$.

  We first remark that, as the proof-term $t$ is closed, it is not a
  variable. Then, we prove that it cannot be an elimination.
  \begin{itemize}
    \item If $t = \elimone(u,v)$, then $u$ is a closed
      irreducible proof-term of $\one$, hence, by induction
      hypothesis, it has the form $a.\star$
      and the proof-term $t$ is reducible.

    \item If $t = u~v$, then $u$ is a closed  irreducible proof-term of $B
      \multimap A$, hence, by induction hypothesis, it has the form
      $\lambda \abstr{x:B}u_1$ and the proof-term $t$ is reducible.

    \item If $t = \elimtens(u,x^B y^C.v)$, then $u$ is a closed
      irreducible proof-term of $B \otimes C$, hence, by induction hypothesis, it
      has the form $u_1 \otimes u_2$, $u_1 \plus u_2$, or $a \bullet
      u_1$ and the proof-term $t$ is reducible.

    \item If $t = \elimzero(u)$, then $u$ is a closed irreducible proof-term of
      $\zero$ and, by induction hypothesis, no such proof-term exists.

    \item If $t = \elimwith^1(u,x^B.v)$, then $u$ is a closed
      irreducible proof-term of $B \with C$, hence, by induction
      hypothesis, it has the form $\pair{u_1}{u_2}$
      and the proof-term $t$ is reducible.

    \item If $t = \elimwith^2(u,x^C.v)$, then $u$ is a closed
      irreducible proof-term of $B \with C$, hence, by induction
      hypothesis, it has the form $\pair{u_1}{u_2}$
      and the proof-term $t$ is reducible.

    \item If $t = \elimplus(u,x^B.v,y^C.w)$, then $u$ is a closed
      irreducible proof-term of $B \oplus C$, hence, by induction hypothesis, it
      has the form $\inl(u_1)$, $\inr(u_1)$, $u_1 \plus u_2$, or $a \bullet
      u_1$ and the proof-term $t$ is reducible.

    \item If $t = \elimbang(u,x^B.v)$, then $u$ is a closed
      irreducible proof-term of $\oc B$, hence, by induction
      hypothesis, it has the form $\oc u_1$
      and the proof-term $t$ is reducible.

    \item If $t = u~C$, then $u$ is a closed  irreducible proof-term of $\forall X.B$, hence, by induction hypothesis, it has the form
      $\Lambda \abstr{X}u_1$ and the proof-term $t$ is reducible.
  \end{itemize}

  Hence, $t$ is an introduction, a sum, or a product.

  It $t$ has the form $a.\star$, then $A$ is $\one$. 
  If it has the form $\lambda x^B.u$, then $A$ has the form $B \multimap C$.
  If it has the form $u \otimes v$, then $A$ has the form $B \otimes C$.
  If it is $\langle \rangle$, then $A$ is $\top$.
  If it has the form $\pair{u}{v}$, then $A$ has the form $B \with C$.
  If it has the form $\inl(u)$ or $\inr(u)$, then $A$ has the form $B \oplus C$.
  If it has the form $\oc u$, then $A$ has the form $\oc B$.
  If it has the form $\Lambda \abstr{X}u$, then $A$ has the form $\forall X.B$.
  We prove that, if it has the form $u \plus v$ or $a \bullet u$, $A$ has
  the form $B \otimes C$ or $B \oplus C$.
  \begin{itemize}
    \item 
      If $t = u \plus v$, then the proof-terms $u$ and $v$ are two closed and irreducible proof-terms of $A$.
      If $A = \one$ then, by induction hypothesis, they both have the form $a.\star$ and the proof-term $t$ is reducible.
      If $A$ has the form $B \multimap C$ then, by induction hypothesis, they are both abstractions and the proof-term $t$ is reducible.
      If $A = \top$ then, by induction hypothesis, they both are $\langle \rangle$ and the proof-term $t$ is reducible.
      If $A = \zero$ then, they are irreducible proof-terms of $\zero$ and, by induction hypothesis, no such proof-terms exist.
      If $A$ has the form $B \with C$, then, by induction hypothesis, they are both pairs and the proof-term $t$ is reducible.  Hence, $A$ has the form $B \otimes C$ or $B \oplus C$.
      If $A$ has the form $\oc B$, then, by induction hypothesis, they both have exponential connectives in head position, and the proof-term $t$ is reducible.
      If $A$ has the form $\forall X.B$ then, by induction hypothesis, they are both universal abstractions and the proof-term $t$ is reducible.

    \item
      If $t = a \bullet u$, then the proof-term $u$ is a closed and irreducible proof-term of $A$.
      If $A = \one$ then, by induction hypothesis, $u$ has the form $a.\star$ and the proof-term $t$ is reducible.
      If $A$ has the form $B \multimap C$ then, by induction hypothesis, it is an abstraction and the proof-term $t$ is reducible.
      If $A = \top$ then, by induction hypothesis, it is $\langle \rangle$ and the proof-term $t$ is reducible.
      If $A = \zero$ then, it is an irreducible proof-term of $\zero$ and, by induction hypothesis, no such proof-term exists.
      If $A$ has the form $B \with C$, then, by induction hypothesis, it is a pair and the proof-term $t$ is reducible.  Hence, $A$ has the form $B \otimes C$ or $B \oplus C$.
      If $A$ has the form $\oc B$, then, by induction hypothesis, it has an exponential connective in head position and the proof-term $t$ is reducible.
      If $A$ has the form $\forall X.B$ then, by induction hypothesis, it is a universal abstraction and the proof-term $t$ is reducible.
      \qed
  \end{itemize}
\end{proof}

\section{Proof of Section~\ref{sec:seclinearity}}\label{proof:seclinearity}
\begin{lemma}
  \label{lem:msubst}
If $\Gamma, x^A \vdash t:B$ and $\Delta \vdash u:A$ then $\mu((u/x)t) \leq
\mu(t)+\mu(u)$.
\end{lemma}

\begin{proof}
  By induction on $t$. 
  \begin{itemize}
    \item If $t$ is a variable, then $\Gamma$ is empty, $t = x$, 
      $(u/x)t = u$ and $\mu(t) = 0$.
      Thus, $\mu((u/x)t) = \mu(u) = \mu(t)+\mu(u)$.

    \item If $t = t_1 \plus t_2$, then $\Gamma, x^A \vdash t_1:B$ and $\Gamma, x^A \vdash t_2:B$.  Using the induction hypothesis, we get 
      $\mu((u/x)t)
      % = \mu((u/x)t_1 \plus (u/x)t_2)
      = 1 +
      \max(\mu((u/x)t_1),\mu((u/x)t_2)) \leq 1 + \max(\mu(t_1) + \mu(u),
      \mu(t_2) + \mu(u))
      % = 1 + \max(\mu(t_1), \mu(t_2)) + \mu(u)
      = \mu(t) + \mu(u)$.

    \item If $t = a \bullet t_1$, then $\Gamma, x^A \vdash t_1:B$.
      Using the induction hypothesis, we get 
      $\mu((u/x)t)
      % = \mu(a \bullet (u/x)t_1)
      = 1 + \mu((u/x)t_1)
      \leq 1 + \mu(t_1) + \mu(u) = \mu(t) + \mu(u)$.

    \item The proof-term $t$ cannot be of the form
      $a.\star$, that is not well-typed in $\Gamma, x^A$.

      \item If $t = \elimone(t_1,t_2)$, then
      $\Gamma = \Gamma_1, \Gamma_2$ and there are two cases.
      \begin{itemize}
	\item 
	  If $\Gamma_1, x^A \vdash t_1:\one$ and $\Gamma_2 \vdash t_2:B$,
	  then, using the induction hypothesis, we get
	  $\mu((u/x)t)
          % = \mu(\elimone((u/x)t_1,t_2))
	  = 1 + \mu((u/x)t_1) + \mu(t_2)
	  \leq 1 + \mu(t_1) + \mu(u) + \mu(t_2)
	  = \mu(t) + \mu(u)$.

	\item 
	  If $\Gamma_1 \vdash t_1:\one$ and $\Gamma_2, x^A \vdash t_2:B$,
	  then, using the induction hypothesis, we get
	  $\mu((u/x)t)
          % = \mu(\elimone(t_1,(u/x)t_2))
	  = 1 + \mu(t_1) + \mu((u/x)t_2)
	  \leq 1 + \mu(t_1) + \mu(t_2) + \mu(u)
	  = \mu(t) + \mu(u)$.
      \end{itemize}

    \item If $t = \lambda y^C. t_1$,
we apply the same method as for the case $t = a \bullet t_1$.

%      then $\Gamma, y:B_1, x:A \vdash t_1:B_2$.
%      Using the induction hypothesis, we get
%      $\mu((u/x)t)
%      = \mu(\lambda \abstr{y} (u/x)t_1)
%      = 1 +
%      \mu((u/x)t_1) \leq 1 + \mu(t_1) + \mu(u) = \mu(t) + \mu(u)$.

    \item If $t = t_1~t_2$, we apply the same method as for the
      case $t = \elimone(t_1,t_2)$.

%     then $\Gamma = \Gamma_1, \Gamma_2$ and there are two cases.
%      \begin{itemize}
%	\item 
%	  If $\Gamma_1, x:A \vdash t_1:B_1 \multimap B$
%	  and $\Gamma_2 \vdash t_2:B_1$, then using the induction hypothesis, we get
%	  $\mu((u/x)t)
% = \mu((u/x)t_1~t_2)
%= 1 + \mu((u/x)t_1) + \mu(t_2)
%	  \leq 1 + \mu(t_1) + \mu(u) + \mu(t_2) = \mu(t) + \mu(u)$.

%	\item 
%	  If $\Gamma_1 \vdash t_1:B_1 \multimap B$
%	  and $\Gamma_2, x:A \vdash t_2:B_1$, then using the induction hypothesis,
%	  we get
%	  $\mu((u/x)t)
 %         % = \mu(t_1~(u/x)t_2)
%          = 1 + \mu(t_1) + \mu((u/x)t_2)
%	  \leq 1 + \mu(t_1) + \mu(t_2) + \mu(u) = \mu(t) + \mu(u)$.
%      \end{itemize}

    \item If $t = t_1 \otimes t_2$, we apply the same method as for the
      case $t = \elimone(t_1,t_2)$.

%      then
%      $\Gamma = \Gamma_1, \Gamma_2$ and there are two cases.
%      \begin{itemize}
%	\item 
%	  If $\Gamma_1, x:A \vdash t_1:B_1$ and
%	  $\Gamma_2 \vdash t_2:B_2$, then using the induction hypothesis,
%	  we get
%	  $\mu((u/x)t)
%          = \mu((u/x)t_1 \otimes t_2)
%	  = 1 + \mu((u/x)t_1) + \mu(t_2)
%	  \leq 1 + \mu(t_1) + \mu(u) + \mu(t_2)
%	  = \mu(t) + \mu(u)$.
%	\item
%	  If $\Gamma_1 \vdash t_1:B_1$ and
%	  $\Gamma_2, x:A \vdash t_2:B_2$, then we apply the same method.
%      \end{itemize}

    \item If $t = \elimtens(t_1, y^C z^D.t_2)$,
      we apply the same method as for the
      case $t = \elimone(t_1,t_2)$.

%      then
%      $\Gamma = \Gamma_1, \Gamma_2$ and there are two cases.
%      \begin{itemize}
%	\item 
%	  If $\Gamma_1, x:A \vdash t_1:C_1 \otimes C_2$ and
%	  $\Gamma_2, y:C_1, z:C_2 \vdash t_2:A$, then using the induction hypothesis,
%	  we get
%	  $\mu((u/x)t)
%	  % = \mu(\elimtens((u/x)t_1,\abstr{yz}t_2))
%	  = 1 + \mu((u/x)t_1) + \mu(t_2)
%	  \leq 1 + \mu(t_1) + \mu(u) + \mu(t_2)
%	  = \mu(t) + \mu(u)$.
%	\item
%	  If $\Gamma_1 \vdash t_1:C_1 \otimes C_2$ and
%	  $\Gamma_2, y:C_1, y:C_2, x:A \vdash t_2:A$,
%	  then using the induction hypothesis,
%	  we get
%	  $\mu((u/x)t)
%         = \mu(\elimtens(t_1,\abstr{yz}(u/x)t_2))
%	  = 1 + \mu(t_1) + \mu((u/x)t_2)
%	  \leq 1 + \mu(t_1) + \mu(t_2) + \mu(u) 
%	  = \mu(t) + \mu(u)$.
%      \end{itemize}

    \item If $t = \topintro$, then
      $\mu((u/x)t)
      % = \mu(\topintro)
      = 1
      \leq 1 + \mu(u) 
      = \mu(t) + \mu(u)$.

    \item If $t = \elimzero(t_1)$, then
      $\Gamma = \Gamma_1, \Gamma_2$ and there are two cases.
      \begin{itemize}
	\item 
	  If $\Gamma_1, x^A \vdash t_1:\zero$,
we apply the same method as for the case $t = a \bullet t_1$.

%          then, using the induction
%	  hypothesis, we get $\mu((u/x)t)
% = \mu(\elimzero((u/x)t_1))
%          = 1 + \mu((u/x)t_1) \leq 1 + \mu(t_1) + \mu(u) = \mu(t) + \mu(u)$.

	\item 
	  If $\Gamma_1 \vdash t_1:\zero$, then, we get $\mu((u/x)t) = \mu(t)
	  \leq \mu(t) + \mu(u)$.
      \end{itemize}

    \item If $t = \pair{t_1}{t_2}$, we apply the same method as for the
      case $t = t_1 \plus t_2$.

    \item   
      If $t = \elimwith^1(t_1,y^C.t_2)$, we apply the same method
      as for the case $t = \elimone(t_1,t_2)$.

%      then
%      $\Gamma = \Gamma_1, \Gamma_2$ and there are two cases.
%      \begin{itemize}
%	\item 
%	  If $\Gamma_1, x:A \vdash t_1:C_1 \with  C_2$ and
%	  $\Gamma_2, y:C_1 \vdash t_2:A$, then using the induction hypothesis,
%	  we get
%	  $\mu((u/x)t)
%         = \mu(\elimwith^1((u/x)t_1,\abstr{y}t_2))
%	  = 1 + \mu((u/x)t_1) + \mu(t_2)
%	  \leq 1 + \mu(t_1) + \mu(u) + \mu(t_2)
%	  = \mu(t) + \mu(u)$.
%	\item
%	  If $\Gamma_1 \vdash t_1:C_1 \with  C_2$ and
%	  $\Gamma_2, y:C_1, x:A \vdash t_2:A$, then using the induction hypothesis,
%	  we get
%	  $\mu((u/x)t)
%         = \mu(\elimwith^1(t_1,\abstr{y}(u/x)t_2))
%	  = 1 + \mu(t_1) + \mu((u/x)t_2)
%	  \leq 1 + \mu(t_1) + \mu(t_2) + \mu(u) 
%	  = \mu(t) + \mu(u)$.
%      \end{itemize}

    \item   
      If $t = \elimwith^2(t_1,y^C.t_2)$, 
      we apply the same method as for the
      case $t = \elimone(t_1,t_2)$.

%      we apply the same method as for the case
%      $t = \elimwith^1(t_1,\abstr{y}t_2)$.

    \item If $t = \inl(t_1)$ or $t = \inr(t_1)$,
     we apply the same method as for the case $t = a \bullet t_1$.

    \item   
      If $t = \elimplus(t_1,y^{C_1}.t_2,z^{C_2}.t_3)$ then
      $\Gamma = \Gamma_1, \Gamma_2$ and there are two cases.
      \begin{itemize}
	\item 
	  If $\Gamma_1, x^A \vdash t_1:C_1 \oplus C_2$,
	  $\Gamma_2, y^{C_1} \vdash t_2:A$,
	  $\Gamma_2, z^{C_2} \vdash t_3:A$, then using the induction hypothesis,
	  we get
	  $\mu((u/x)t)
          % = \mu(\elimplus((u/x)t_1,\abstr{y}t_2,\abstr{z}t_3))
	  = 1 + \mu((u/x)t_1) + \max(\mu(t_2),\mu(t_3))
	  \leq 
	  1 + \mu(t_1) + \mu(u) + \max(\mu(t_2),\mu(t_3))
	  = \mu(t) + \mu(u)$.

	\item 
	  If $\Gamma_1 \vdash t_1:C_1 \oplus C_2$,
	  $\Gamma_2, y^{C_1}, x^A \vdash t_2:A$,
	  $\Gamma_2, z^{C_2}, x^A \vdash t_3:A$, then using the induction hypothesis,
	  we get
	  $\mu((u/x)t)
          % = \mu(\elimplus(t_1,\abstr{y}(u/x)t_2,\abstr{z}(u/x)t_3))
	  = 1 + \mu(t_1) + \max(\mu((u/x)t_2),\mu((u/x)t_3))
	  \leq 
	  1 + \mu(t_1) + \max(\mu(t_2) + \mu(u),\mu(t_3) + \mu(u))
	  1 + \mu(t_1) + \max(\mu(t_2),\mu(t_3)) + \mu(u)
	  = \mu(t) + \mu(u)$.
      \end{itemize}
    \item If $t = \Lambda X.t_1$, then $B = \forall X.C$, $\Gamma, x^A \vdash t_1:C$ and $X \notin \FV(\Gamma, A)$. Using the induction hypothesis, we get 
    $\mu((u/x)t) = 1 + \mu((u/x)t_1)
    \leq 1 + \mu(t_1) + \mu(u)
    = \mu(t) + \mu(u)$.
    \item If $t = t_1~C$, then $B = (C/X)D$ and $\Gamma, x^A \vdash t_1:D$. Using the induction hypothesis, we get
    $\mu((u/x)t) = 1 + \mu((u/x)t_1)
    \leq 1 + \mu(t_1) + \mu(u)
    = \mu(t) + \mu(u)$. \qed
  \end{itemize}
\end{proof}

As a corollary, we get a similar measure preservation theorem for
reduction.

\begin{lemma}
  \label{lem:mured}
  If $\Gamma \vdash t:A$ and $t \lra u$, then $\mu(t) \geq \mu(u)$.
\end{lemma}
\begin{proof}
  By induction on $t$. The context cases are trivial because the
  functions used to define $\mu(t)$ in function of $\mu$ of the
  subproof-terms of $t$ are monotone. We check the rules one by one, using
  Lemma~\ref{lem:msubst}.

  \begin{itemize}
    \item $\mu(\elimone(a.\star,t)) = 2 + \mu(t) > 1 + \mu(t) = \mu(a
      \bullet t)$
    \item $\mu((\lambda x^B.t)~u) = 2 + \mu(t) + \mu(u) > \mu(t) +
      \mu(u) \geq \mu((u/x)t)$
    \item $
      \begin{aligned}[t]
	\mu(\elimtens(u \otimes v,x^B y^C.w))
	&= 2 + \mu(u) + \mu(v) + \mu(w)
	> \mu(u) + \mu(v) + \mu(w)\\
	&\geq \mu(u) + \mu((v/y)w)
	\geq  \mu((u/x)(v/y)w)\\
	&= \mu((u/x,v/y)w)
      \end{aligned}
      $

      as $x$ does not occur in $v$
    \item $\mu(\elimwith^1(\pair{t}{u}, x^B.v)) = 2 +
      \max(\mu(t),\mu(u)) + \mu(v) > \mu(t) + \mu(v) \geq \mu((t/x)v)$
    \item $\mu(\elimwith^2(\pair{t}{u}, x^B.v)) = 2 +
      \max(\mu(t),\mu(u)) + \mu(v) > \mu(u) + \mu(v) \geq \mu((u/x)v)$
    \item $
      \begin{aligned}[t]
	\mu(\elimplus(\inl(t),x^B.v,y^C.w)) 
	&= 2 + \mu(t) + \max(\mu(v), \mu(w)) 
	> \mu(t) + \mu(v) \\
	& \geq \mu((t/x)v)
      \end{aligned}
      $
    \item $
      \begin{aligned}[t]
	\mu(\elimplus(\inr(t),x^B.v,y^C.w)) 
	&= 2 + \mu(t) + \max(\mu(v),\mu(w))
	> \mu(t) + \mu(w)\\
	&\geq \mu((t/y)w)
      \end{aligned}
      $
    \item $\mu((\Lambda X.t)~B) = 2 + \mu(t) > \mu(t) = \mu((B/X)t)$
    \item $\mu({a.\star} \plus b.\star) = 2 > 1 = \mu((a+b).\star)$
    \item $
      \begin{aligned}[t]
	\mu((\lambda x^B.t) \plus (\lambda x^B.u))
	&= 1 + \max (1 + \mu(t), 1 + \mu(u)) \\
	&= 2 + \max(\mu(t),\mu(u))
	= \mu(\lambda x^B.(t \plus u))
      \end{aligned}
      $
    \item $
      \begin{aligned}[t]
	\mu(\elimtens(t \plus u,x^B y^C.v)) 
	&= 2 + \max(\mu(t), \mu(u)) + \mu(v) \\
	&= 1 + \max(1 + \mu(t) + \mu(v), 1 + \mu(u) + \mu(v)) \\
	&= \mu(\elimtens(t,x^B y^C.v) \plus \elimtens(u,x^B y^C.v))
      \end{aligned}
      $
    \item $\mu(\topintro \plus \topintro) = 2 > 1 = \mu(\topintro)$
    \item $
      \begin{aligned}[t]
	\mu(\pair{t}{u} \plus \pair{v}{w}) 
	&= 1 + \max (1 + \max(\mu(t),\mu(u)), 1 + \max(\mu(v),\mu(w)))\\
	&= 2 + \max (\mu(t),\mu(u),\mu(v),\mu(w)) \\
	&= 1 + \max(1+ \max(\mu(t),\mu(v)), 1 + \max(\mu(u),\mu(w)))\\
	&= \mu(\pair{t \plus v}{u \plus w})
      \end{aligned}
      $
    \item $
      \begin{aligned}[t]
	&\mu(\elimplus(t \plus u,x^B.v,y^C.w)) \\
	&= 2 + \max(\mu(t),\mu(u)) + \max(\mu(v), \mu(w)) \\
	&= 1 + \max(1 + \mu(t) + \max(\mu(v), \mu(w)), 1 + \mu(u) + \max(\mu(v), \mu(w))) \\
	&= \mu(\elimplus(t,x^B.v,y^C.w) \plus \elimplus(u,x^B.v,y^C.w))
      \end{aligned}
      $
    \item $
      \begin{aligned}[t]
	\mu((\Lambda X.t) \plus (\Lambda X.u)) 
	&= 1 + \max(1 + \mu(t), 1 + \mu(u)) 
	= 2 + \max(\mu(t), \mu(u)) \\
	&= \mu(\Lambda X.(t \plus u))
      \end{aligned}
      $
    \item $\mu(a \bullet b.\star) = 2 > 1 = \mu((a \times b).\star)$
    \item $\mu(a \bullet \lambda x^B. t) = 2 + \mu(t) = \mu(\lambda
      x^B. a \bullet t)$
    \item $\mu(\elimtens(a \bullet t,x^B y^C.v)) = 2 + \mu(t) + \mu(v) =
      a \bullet \elimtens(t,x^B y^C.v)$
    \item $\mu(a \bullet \topintro) = 2 > 1 = \mu(\topintro)$
    \item $\mu(a \bullet \pair{t}{u}) = 2 + \max(\mu(t),\mu(u)) = 1 +
      \max(1 + \mu(t),1 + \mu(u)) = \mu(\pair{a \bullet t}{a \bullet u})$
    \item
      $
      \begin{aligned}[t]
	\mu(\elimplus(a \bullet t,x^B.v,y^C.w)) 
	&= 2 + \mu(t) + \max(\mu(v),\mu(w)) \\
	&= \mu(a \bullet \elimplus(t,x^B.v,y^C.w))
      \end{aligned}
      $
    \item {$\mu(a \bullet \Lambda X.t) = 2 + \mu(t) = \mu(\Lambda X. a \bullet t)$} \qed
  \end{itemize} 
\end{proof}

In the case of elimination contexts, Lemma~\ref{lem:msubst} can be
strengthened.

\begin{lemma}
  \label{lem:strengtheningmsubst}
  $\mu(K\{t\}) = \mu(K) + \mu(t)$
\end{lemma}
\begin{proof}
  By induction on $K$.

  \begin{itemize}
    \item If $K = \_$, then $\mu(K) = 0$ and $K\{t\} = t$.  We have
      $\mu(K\{t\}) = \mu(t) = \mu(K) + \mu(t)$.

    \item If $K = \elimone(K_1,u)$ then $K\{t\} = \elimone(K_1\{t\},u)$.
      We have, by
      induction hypothesis, $\mu(K\{t\}) = 1 + \mu(K_1\{t\}) + \mu(u)
      = 1 + \mu(K_1) + \mu(t) + \mu(u) = \mu(K) + \mu(t)$.

    \item If $K = K_1~u$ then $K\{t\} = K_1\{t\}~u$.  We have, by
      induction hypothesis, $\mu(K\{t\}) = 1 + \mu(K_1\{t\}) + \mu(u) = 1 +
      \mu(K_1) +   \mu(t) + \mu(u) = \mu(K) + \mu(t)$.

    \item 
      If $K = \elimtens(K_1,x^A y^B.v)$, then $K\{t\} =
      \elimtens(K_1\{t\},x^A y^B.v)$.  We have, by induction
      hypothesis, $\mu(K\{t\}) = 1 + \mu(K_1\{t\}) + \mu(v) = 1 +
      \mu(K_1) + \mu(t) + \mu(v) = \mu(K) + \mu(t)$.

    \item If $K = \elimzero(K_1)$, then $K\{t\} = \elimzero(K_1\{t\})$. We
      have, by induction hypothesis, $\mu(K\{t\}) = 1 + \mu(K_1\{t\})= 1 +
      \mu(K_1) + \mu (t) = \mu(K) + \mu(t)$.

    \item If $K = \elimwith^1(K_1,x^A.r)$, then $K\{t\} =
      \elimwith^1(K_1\{t\},x^A.r)$. We have, by induction hypothesis,
      $\mu(K\{t\}) = 1 + \mu(K_1\{t\}) + \mu(r) = 1 + \mu(K_1) + \mu (t) +
      \mu(r) = \mu(K) + \mu(t)$.

      The same holds if $K = \elimwith^2(K_1,y^A.s)$.

    \item 
      If $K = \elimplus(K_1,x^A.r,y^B.s)$, then $K\{t\} =
      \elimplus(K_1\{t\},x^A.r,y^B.s)$.  We have, by induction
      hypothesis, $\mu(K\{t\}) = 1 + \mu(K_1\{t\}) + \max(\mu(r), \mu(s)) = 1 +
      \mu(K_1) + \mu(t) + \max(\mu(r), \mu(s)) = \mu(K) + \mu(t)$.

    \item 
       If $K = K_1~A$ then $K\{t\} = K_1\{t\}~A$. We have, by induction hypothesis, $\mu(K\{t\}) = 1 + \mu(K_1\{t\}) = 1 + \mu(K_1) + \mu(t) = \mu(K) + \mu(t)$. \qed
  \end{itemize}
\end{proof}

\decomp*
\begin{proof}
  By induction on $t$.

\begin{itemize}
    \item If $t$ is the variable $x$, an introduction, a sum, or a
     product, we take $K = \_$, $u = t$, and $B = A$.

    \item If $t = \elimone(t_1,t_2)$, then $t_1$ is not a closed proof-term as
      otherwise it would be a closed irreducible proof-term of $\one$, hence,
      by Theorem~\ref{thm:introductions}, it would be an introduction and $t$
      would not be irreducible. Thus, by the inversion property, $x^C
      \vdash t_1:\one$ and $\vdash t_2:A$.

      By induction hypothesis, there exist $K_1$, $u_1$ and $B_1$ such
      that $\_^{B_1} \vdash K_1:\one$, $x^C \vdash u_1:B_1$, and $t_1 =
      K_1\{u_1\}$.  We take $u = u_1$, $K = \elimone(K_1,t_2)$, and $B =
      B_1$.  We have $\_^B \vdash K:A$, $x^C \vdash u:B$, and $K\{u\} =
      \elimone(K_1\{u_1\},t_2) = t$.
      
    \item If $t = t_1~t_2$,
          we apply the same method as for the case $t = \elimone(t_1,t_2)$.

%    \item If $t = t_1~t_2$, then $t_1$ is not a closed proof-term as otherwise it
%      would be a closed irreducible proof-term of an implication, hence, by
%      Theorem~\ref{thm:introductions}, it would be an introduction and $t$
%      would not be irreducible. Thus, by the inversion property, $x:C
%      \vdash t_1:D \multimap A$ and $ \vdash t_2:D$.

%      By induction hypothesis, there exist $K_1$, $u_1$ and $B_1$ such
%      that $\_:B_1 \vdash K_1:D \multimap A$, $x:C \vdash u_1:B_1$, and
%      $t_1 = K_1\{u_1\}$.  We take $u = u_1$, $K = K_1~t_2$, and $B =
%      B_1$.  We have $\_:B \vdash K:A$, $x:C \vdash u:B$, and $K\{u\} =
%      K_1\{u_1\}~t_2 = t$.

    \item If $t = \elimtens(t_1,y^{D_1} z^{D_2}.t_2)$, then $t_1$ is not a
      closed proof-term as otherwise it would be a closed irreducible proof-term of a
      multiplicative conjunction $\otimes$, hence,
      by Theorem~\ref{thm:introductions}, it would be an introduction, a sum, or a
      product, and $t$ would not be irreducible. Thus, by the inversion
      property, $x^C \vdash t_1:D_1 \otimes D_2$ and $y^{D_1}, z^{D_2} \vdash
      t_2:A$.

      By induction hypothesis, there exist $K_1$, $u_1$ and $B_1$ such
      that $\_^{B_1} \vdash K_1:C \otimes
      D$, $x^C \vdash u_1:B_1$, and $t_1 = K_1\{u_1\}$.  We take $u =
      u_1$, $K = \elimtens(K_1,y^{D_1} z^{D_2}.t_2)$, and $B = B_1$.  We have\\
      $\_^B \vdash K:A$, $x^C \vdash u:B$, and $K\{u\} =
      \elimtens(K_1\{u_1\},y^{D_1} z^{D_2}.t_2) = t$.

    \item If $t = \elimzero(t_1)$, then,
      by Theorem~\ref{thm:introductions},
      $t_1$ is not a closed proof-term as there is
      no closed irreducible proof-term of $\zero$. Thus, by the inversion
      property, $x^C \vdash t_1:\zero$.

      By induction hypothesis, there exist $K_1$, $u_1$, and $B_1$ such
      that $\_^{B_1} \vdash K_1:\zero$, $x^C \vdash u_1:B_1$, and $t_1 =
      K_1\{u_1\}$.  We take $u = u_1$, $K = \elimzero(K_1)$, and $B =
      B_1$.  We have $\_^B, \vdash K:A$, $x^C \vdash u:B$, and $K\{u\} =
      \elimzero(K_1\{u_1\}) = t$.

    \item If $t = \elimwith^1(t_1,y^D.t_2)$ or
      $t = \elimwith^2(t_1,y^D.t_2)$,
      we apply the same method as for the case $t = \elimone(t_1,t_2)$.

%    \item If $t = \elimwith^1(t_1,\abstr{y}t_2)$, then $t_1$ is not a
%      closed proof-term as otherwise it would be a closed irreducible proof-term of an
%      additive conjunction $\with $, hence, by Theorem~\ref{thm:introductions}, it
%      would be an introduction and $t$ would not be irreducible. Thus, by
%      the inversion property, $x:C \vdash t_1:D_1 \with  D_2$ and $y:D_1
%      \vdash t_2:A$.

%      By induction hypothesis,there exist $K_1$, $u_1$, and $B_1$ such that
%      $ \_:B_1 \vdash K_1:D_1 \with  D_2$, $x:C \vdash u_1:B_1$, and $t_1 =
%      K_1\{u_1\}$.  We take $u = u_1$, $K =
%      \elimwith^1(K_1,\abstr{y}t_2)$, and $B = B_1$.  We have $\_:B \vdash
%      K:A$, $x:C \vdash u:B$, and $K\{u\} =
%      \elimwith^1(K_1\{u_1\},\abstr{y}t_2) = t$.

%    \item If $t = \elimwith^2(t_1,\abstr{y}t_2)$, the proof is similar.

\item If $t =  \elimplus(t_1,y^{D_1}.t_2,z^{D_2}.t_3)$,
we apply the same method as for the case
$t = \elimtens(t_1,y^{D_1} z^{D_2}.t_2)$.

      %    \item If $t = \elimplus(t_1,\abstr{y}t_2,\abstr{z}t_3)$, then $t_1$ is
%      not a closed proof-term as otherwise it would be a closed irreducible proof-term of
%      an additive disjunction $\oplus$, hence, by
%      Theorem~\ref{thm:introductions}, it would be an introduction, a sum, or
%      a product, and $t$ would not be irreducible. Thus, by the inversion
%      property, $x:C \vdash t_1:D_1 \oplus D_2$, $y:D_1 \vdash t_2:A$, and
%      $z:D_2 \vdash t_3:A$.

%      By induction hypothesis, there exist $K_1$, $u_1$, and $B_1$ such
%      x  that $\_:B_1 \vdash K_1:C \oplus D$, $x:C \vdash u_1:B_1$, and $t_1
%      = K_1\{u_1\}$.  We take $u = u_1$, $K =
%      \elimplus(K_1,\abstr{y}t_2,\abstr{z}t_3)$, and $B = B_1$.  We have
%      $\_:B \vdash K:A$, $x:C \vdash u:B$, and $K\{u\} =
%      \elimplus(K_1\{u_1\},\abstr{y}t_2,\abstr{z}t_3) = t$.
{
\item If $t = t_1~D$, then by the inversion property $A = (D/X)E$ and $x^C
     \vdash t_1:\forall X.E$.
     By induction hypothesis, there exist $K_1$, $u_1$ and $B_1$ such
     that $\_^{B_1} \vdash K_1:\forall X.E$, $x^C \vdash u_1:B_1$, and
     $t_1 = K_1\{u_1\}$.  We take $u = u_1$, $K = K_1~D$, and $B =
     B_1$.  We have $\_^B \vdash K:A$, $x^C \vdash u:B$, and $K\{u\} =
     K_1\{u_1\}~D = t$.\qed
}
\end{itemize}
\end{proof}

A final lemma shows that we can always decompose an elimination context $K$
different from $\_$ into a smaller elimination context $K_1$ and a last
elimination rule $K_2$.  This is similar to the fact that we can always
decompose a non-empty list into a smaller list and its last element.

\begin{lemma}[Decomposition of an elimination context]
  \label{lem:horrible}
  If $K$ is an elimination context such that $\_^A \vdash K:B$
  and $K \neq \_$, then $K$ has the form $K_1\{K_2\}$ and
  $K_2$ is an elimination of the top symbol of $A$.
\end{lemma}
\begin{proof}
  As $K$ is not $\_$, it has the form $K = L_1\{L_2\}$.
  If $L_2 = \_$, we take $K_1 = \_$, $K_2 = L_1$ and, as the proof-term is
  well-typed, $K_2$ must be an elimination of the top symbol of $A$.
  Otherwise, by induction hypothesis, $L_2$ has the form $L_2 = K'_1\{K'_2\}$,
  and $K'_2$ is an elimination of the top symbol of $A$.
  Hence, $K = L_1\{K'_1\{K'_2\}\}$. We take $K_1 = L_1\{K'_1\}$, $K_2 = K'_2$. \qed
\end{proof}

\linearity*
\begin{proof}
  Without loss of generality, we can assume that $t$ is irreducible.
  We proceed by induction on $\mu(t)$. 

  Using Lemma~\ref{lem:elim}, the proof-term $t$ can be decomposed as $K\{t'\}$
  where $t'$ is either the variable $x$, an introduction, a sum, or a
  product.

  \begin{itemize}
    \item 
      If $t'$ is an introduction, as $t$ is irreducible, $K = \_$ and
      $t'$ is a proof-term of $B \in
      {\mathcal V}$, $t'$ is either $a.\star$ or $\pair{t_1}{t_2}$. However,
      since $a.\star$ is not well-typed in $x:A$, it is $\pair{t_1}{t_2}$.
      Using the induction hypothesis with $t_1$ and with $t_2$
      ($\mu(t_1) < \mu(t')$, $\mu(t_2) < \mu(t')$), 
      we get
      \begin{align*}
	t\{u_1 \plus u_2\}
%	&= \pair{t_1\{u_1 \plus u_2\}}{t_2\{u_1 \plus u_2\}}
	\equiv
	\pair{t_1\{u_1\} \plus t_1\{u_2\}}{t_2\{u_1\} \plus t_2\{u_2\}}
%	&
	\lla 
%	\pair{t_1\{u_1\}}{ t_2\{u_1\}} \plus \pair{t_1\{u_2\}}{t_2\{u_2\}}
	%=
	t\{u_1\} \plus t\{u_2\}
      \end{align*}
      And
      \begin{align*}
	t\{a \bullet u_1\}
%	& = \pair{t_1\{a \bullet u_1\}}{t_2\{a \bullet u_1\}}
	\equiv \pair{a \bullet t_1\{u_1\}}{a \bullet t_2\{u_1\}}
	&
	\lla
%	a \bullet \pair{t_1\{u_1\}}{t_2\{u_1\}}
%	=
	a \bullet t\{u_1\}
      \end{align*}

    \item If $t' = t_1 \plus t_2$, then using the induction hypothesis
      with $t_1$, $t_2$, and $K$ ($\mu(t_1) < \mu(t)$, $\mu(t_2) <
      \mu(t)$, and $\mu(K) < \mu(t)$) and Lemma~\ref{lem:vecstructure} (1.,
      2., and 7.), we get
      \begin{align*}
	t\{u_1 \plus u_2\}
%	&= K\{t_1\{u_1 \plus u_2\} \plus t_2\{u_1 \plus u_2\}\}
	&	\equiv K\{(t_1\{u_1\} \plus t_1\{u_2\})
	\plus (t_2\{u_1\} \plus t_2\{u_2\})\}\\
	&\equiv 
	K(( t_1\{u_1\} \plus t_2\{u_1\})
	\plus (t_1\{u_2\} \plus t_2\{u_2\}))
	\equiv  
%	K( t_1\{u_1\} \plus t_2\{u_1\})
%	\plus K(t_1\{u_2\} \plus t_2\{u_2\})
%	=
	t\{u_1\} \plus t\{u_2\}
      \end{align*}
      And 
      \begin{align*}
	t\{a \bullet u_1\}
	&	\equiv K\{a \bullet t_1\{u_1\} \plus a \bullet t_2\{u_1\}\}
	\equiv K\{a \bullet (t_1\{u_1\} \plus t_2\{u_1\})\}\\
	&\equiv 
	a \bullet t\{u_1\}
      \end{align*}

    \item
      If $t' = b \bullet t_1$, then using the induction hypothesis
      with $t_1$ and $K$ ($\mu(t_1) < \mu(t)$, $\mu(K) < \mu(t)$) and
      $K$ and Lemma~\ref{lem:vecstructure} (7. and 5.), we get
      \begin{align*}
	t\{u_1 \plus u_2\}
	\equiv K\{b \bullet (t_1 \{u_1\} \plus t_1\{u_2\})\}
	& \equiv K\{b \bullet t_1 \{u_1\} \plus b \bullet t_1\{u_2\}\}\\
	&\equiv 
	t\{u_1\} \plus t\{u_2\}
      \end{align*}
      And 
      \begin{align*}
	t\{a \bullet u_1\}
%	&= K\{b \bullet t_1 \{a \bullet u_1\}\}
	&	\equiv K\{b \bullet a \bullet t_1 \{u_1\}\}
%	\equiv K\{(b \times a) \bullet t_1 \{u_1\}\}
%	 = K\{(a \times b) \bullet t_1 \{u_1\}\}\\
%&
	\equiv K\{a \bullet b \bullet t_1 \{u_1\}\}
	\equiv 
%	a \bullet K\{b \bullet t_1 \{u_1\}\}
%	=
	a \bullet t\{u_1\}
      \end{align*}

    \item If $t'$ is the variable $x$, we need to prove
      \[
	K\{u_1 \plus u_2\} \equiv K\{u_1\} \plus K\{u_2\}
	\qquad\textrm{and}\qquad
	K\{a \bullet u_1\} \equiv a \bullet K\{u_1\}
      \]
      By Lemma~\ref{lem:horrible},
      $K$ has the form $K_1\{K_2\}$ and $K_2$ is an elimination of the top symbol of $A$.
      We consider the various cases for $K_2$. 
      \begin{itemize}
	\item If $K = K_1\{\elimone(\_,r)\}$, then $u_1$ and $u_2$ are closed
	  proof-terms of $\one$, thus $u_1 \lras b.\star$ and $u_2\lras c.\star$.
	  Using the induction hypothesis with the proof-term $K_1$
	  ($\mu(K_1) < \mu(K) = \mu(t)$) and Lemma~\ref{lem:vecstructure} (8.~and 5.)
	  \begin{align*}
	    K\{u_1 \plus u_2\}
	    &\lras K_1 \{\elimone({b.\star} \plus c.\star,r)\}
	    \lras K_1 \{(b + c) \bullet r\}\\
	    &\equiv (b + c) \bullet K_1 \{r\}
	    \equiv b \bullet K_1 \{r\} \plus c \bullet K_1 \{r\}\\
	    &\equiv K_1 \{b \bullet r\} \plus K_1 \{c \bullet r\}\\
	    & \llas K_1 \{\elimone(b.\star,r)\} \plus K_1 \{\elimone(c.\star,r)\}\\
	    &\llas K\{u_1\} \plus K\{u_2\}
	  \end{align*}
	  And
	  \begin{align*}
	    K\{a \bullet u_1\}
	    &\lras K_1 \{\elimone(a \bullet b.\star,r)\}
	    \lras K_1 \{(a \times b) \bullet r\}\\
	   & \equiv (a \times b) \bullet  K_1 \{r\}
	     \equiv a \bullet b \bullet  K_1 \{r\}
	    \equiv a \bullet  K_1 \{b \bullet r\}\\
	    &\llas a \bullet K_1 \{\elimone(b.\star,r)\}
	    \llas a \bullet K\{u_1\}
	  \end{align*}

	\item
	  If $K = K_1 \{\_~s\}$, then $u_1$ and $u_2$ are closed 
	  proof-terms of an implication, thus 
	  $u_1 \lras \lambda \abstr{y^C} u'_1$
	  and
	  $u_2 \lras \lambda \abstr{y^C} u'_2$.
	  Using the induction hypothesis with the proof-term $K_1 $
	  ($\mu(K_1 ) < \mu(K) = \mu(t)$), we get
	  \begin{align*}
	    K\{u_1 \plus u_2\}
	    &\lras K_1 \{(\lambda \abstr{y^C} u'_1 \plus \lambda \abstr{y^C} u'_2)~s\}
	    \lras K_1 \{u'_1\{s\} \plus u'_2\{s\}\}\\
	    &\equiv K_1 \{u'_1\{s\}\} \plus K_1 \{u'_2\{s\}\}\\
	    &\llas K_1 \{(\lambda \abstr{y^C} u'_1)~s\} \plus K_1 \{(\lambda \abstr{y^C} u'_2)~s\}\\
	    & \llas K\{u_1\} \plus K\{u_2\}
	  \end{align*}
	  And
	  \begin{align*}
	    K\{a \bullet u_1\}
	    &\lras K_1 \{(a \bullet \lambda \abstr{y^C} u'_1)~s\}
	    %	    \lra K_1 \{\lambda \abstr{y} (a \bullet u'_1)~s\}
	    \lras K_1 \{a \bullet u'_1\{s\}\}\\
	    & \equiv a \bullet  K_1 \{u'_1\{s\}\}
	    \lla
	    a \bullet K_1 \{(\lambda \abstr{y^C} u'_1)~s\}
	    \llas a \bullet K\{u_1\}
	  \end{align*}

	\item
	  If $K = K_1 \{\elimtens(\_,y^C z^D.r)\}$, then, using
	  the induction hypothesis with the proof-term $K_1 $ 
	  ($\mu(K_1 ) < \mu(K) = \mu(t)$), we get
	  \begin{align*}
	    K\{u_1 \plus u_2\}
	    &
	    \lra K_1 \{\elimtens(u_1,y^C z^D.r)
	    \plus \elimtens(u_2,y^C z^D.r)\}\\
	    &\equiv
	    K\{u_1\} \plus K\{u_2\}
	  \end{align*}
	  And 
	  \begin{align*}
	    K\{a \bullet u_1\}
	    %&= K_1 \{\elimtens(a \bullet u_1,\abstr{y z}r)\}
	    \lra K_1 \{a \bullet \elimtens(u_1,y^C z^D.r)\}
	    &\equiv
	    %a \bullet K_1 \{\elimtens(u_1,\abstr{y z}r)\} =
	    a \bullet K\{u_1\}
	  \end{align*}

	\item The case $K = K_1 \{\elimzero(\_)\}$ is not possible as $u_1$ would
	  be a closed proof-term of $\zero$ and there is no such proof-term.

	\item
	  If $K = K_1 \{\elimwith^1(\_,\abstr{y^C}r)\}$, then $u_1$ and $u_2$
	  are closed proof-terms of an additive conjunction $\with $, thus $u_1 \lras
	  \pair{u_{11}}{u_{12}}$ and $u_2 \lras \pair{u_{21}}{u_{22}}$.

	  Let $r'$ be the irreducible form of $K_1 \{r\}$.
	  Using the induction hypothesis with the proof-term $r'$
	  (because, with Lemmas~\ref{lem:mured} and~\ref{lem:strengtheningmsubst}, we have
	    $\mu(r') \leq \mu(K_1 \{r\}) = \mu(K_1 ) + \mu(r) < \mu(K_1 ) + \mu(r) + 1
	  = \mu(K) = \mu(t)$)
	  \begin{align*}
	    &K\{u_1 \plus u_2\}\\
	    &\lras K_1 \{\elimwith^1(\pair{u_{11}}{u_{12}} \plus \pair{u_{21}}{u_{22}}, \abstr{y^C}r)\}
	    \lras K_1 \{r\{u_{11} \plus u_{21}\}\}\\
	    & \lras r'\{u_{11} \plus u_{21}\}
	    \equiv r'\{u_{11}\} \plus r'\{u_{21}\}
	    \llas K_1 \{r\{u_{11}\}\} \plus K_1 \{r\{u_{21}\}\}\\
	    & \llas K_1 \{\elimwith^1(\pair{u_{11}}{u_{12}},\abstr{y^C}r)\}\plus K_1 \{\elimwith^1(\pair{u_{21}}{u_{22}},\abstr{y^C}r)\}\\
	    &\llas K\{u_1\} \plus K\{u_2\} 
	  \end{align*}
	  And
	  \begin{align*}
	    K\{a \bullet u_1\}
	    &\lras K_1 \{\elimwith^1(a \bullet \pair{u_{11}}{u_{12}},\abstr{y^C}r)\}
	    \lra^* K_1 \{r\{a \bullet u_{11}\}\}\\
	    &\lras r'\{a \bullet u_{11}\}
	     \equiv a \bullet r'\{u_{11}\}
	    \llas a \bullet K_1 \{r\{u_{11}\}\}\\
	    &\lla a \bullet K_1 \{\elimwith^1(\pair{u_{11}}{u_{12}},\abstr{y^C}r)\}
	     \llas a \bullet K\{u_1\}
	  \end{align*}

	\item If $K = K_1 \{\elimwith^2(\_,\abstr{y^C}r)\}$, the proof is similar.

	\item
	  If $K = K_1 \{\elimplus(\_,\abstr{y^C}r, \abstr{z^D}s)\}$, then, using
	  the induction hypothesis with the proof-term $K_1 $ 
	  ($\mu(K_1 ) < \mu(K) = \mu(t)$), we get
	  \begin{align*}
	    K\{u_1 \plus u_2\}
	    &\lra K_1 \{\elimplus( u_1,\abstr{y^C}r, \abstr{z^D}s) \plus \elimplus( u_2,\abstr{y^C}r, \abstr{z^D}s)\}\\
	    &\equiv
	    K\{u_1\} \plus K\{u_2\}
	  \end{align*}
	  And 
%	  \begin{align*}
	  \[
	    K\{a \bullet u_1\}
	    %&= K_1 \{\elimplus(a \bullet u_1,\abstr{y}r, \abstr{z}s)\}
	    \lra K_1 \{a \bullet \elimplus(u_1,\abstr{y^C}r, \abstr{z^D}s)\}
	    %&
	    \equiv
	    %a \bullet K_1 \{\elimplus(u_1,\abstr{y}r, \abstr{z}s)\} =
	    a \bullet K\{u_1\}
	  \]
%	  \end{align*}
	  \item If $K = K_1\{\_~C\}$, then $u_1$ and $u_2$ are closed proof-terms of $\forall X.D$, thus $u_1 \lras \Lambda X.u_1'$ and $u_2 \lras \Lambda X.u_2'$. Using the induction hypothesis with the proof-term $K_1$ ($\mu(K_1) < \mu(K) = \mu(t)$), we get
	    \begin{align*}
	      K\{u_1 \plus u_2\}
	      &\lras K_1 \{(\Lambda X.u'_1 \plus \Lambda X.u'_2)~C\}\\
	      &\lras K_1 \{(C/X)u'_1 \plus (C/X)u'_2\}
	      = (C/X)K_1\{u'_1 \plus u'_2\} \\
	      &\equiv (C/X)(K_1 \{u'_1\} \plus K_1 \{u'_2\})\\
	      &= K_1 \{(C/X)u'_1\} \plus K_1 \{(C/X)u'_2\}\\
	      & \llas K_1\{(\Lambda X.u'_1)~C\} \plus K_1\{(\Lambda X.u'_2)~C\}\\
	      &\llas K\{u_1\} \plus K\{u_2\}
	    \end{align*}
	    And
	    \begin{align*}
	      K\{a \bullet u_1\}
	      &\lras K_1 \{(a \bullet \Lambda X.u'_1)~C\}
	      \lras K_1 \{a \bullet (C/X)u'_1\}\\
	      &= (C/X)K_1 \{a \bullet u'_1\}
	       \equiv (C/X)(a \bullet  K_1 \{u'_1\})\\
	      &= a \bullet K_1 \{(C/X)u'_1\}
	       \llas a \bullet K_1 \{(\Lambda X.u'_1)~C\}\\
	      &\llas a \bullet K\{u_1\} \qedHERE
	    \end{align*}
      \end{itemize}
  \end{itemize}
\end{proof}

\corollarylin*
\begin{proof}
  Let $C \in {\mathcal V}$ and $c$ be a proof-term such that
  $\_^B \vdash c:C$. Then applying Theorem~\ref{thm:linearity} to
  the proof-term $c\{t\}$ we get 
  \[
    c\{t\{u_1 \plus u_2\}\} \equiv c\{t\{u_1\}\} \plus c\{t\{u_2\}\}
    \quad\qquad\textrm{and}\qquad\quad
    c\{t\{a \bullet u_1\}\} \equiv a \bullet c\{t\{u_1\}\}
  \]
  and applying it again to the proof-term $c$ we get
  \[
    c\{t\{u_1\} \plus t\{u_2\}\} \equiv c\{t\{u_1\}\} \plus c\{t\{u_2\}\}
    \quad\quad\textrm{and}\quad\quad
    c\{a \bullet t\{u_1\}\} \equiv a \bullet c\{t\{u_1\}\}
  \]
  Thus
  \[
    c\{t\{u_1 \plus u_2\}\} \equiv c\{t\{u_1\} \plus t\{u_2\}\}
    \qquad\qquad\textrm{and}\qquad\qquad
    c\{t\{a \bullet u_1\}\} \equiv c\{a \bullet t\{u_1\}\}
  \]
  that is 
  \[
    t\{u_1 \plus u_2\} \sim t\{u_1\} \plus t\{u_2\}
    \qquad\qquad\textrm{and}\qquad\qquad
    t\{a \bullet u_1\} \sim a \bullet t\{u_1\}
    \qedHEREHERE
  \]
\end{proof}

\end{document}